\documentclass[12pt]{article}

\usepackage{amssymb,amsmath,amsfonts,bbm,eurosym,geometry,ulem,graphicx,caption,color,setspace,comment,footmisc,caption,natbib,pdflscape,array,enumerate,nomencl,etoolbox,mathtools,booktabs,makecell,multirow,multicol,float}
\usepackage{subfigure}
\usepackage{makecell}

\usepackage[ruled]{algorithm2e}
\usepackage[thmmarks,amsmath,thref]{ntheorem}
\newtheorem{theorem}{Theorem}[section]
\newtheorem{definition}{Definition}[section]
\newtheorem{lemma}{Lemma}[section]
\newtheorem{prop}{Proposition}[section]

\newtheorem{corollary}{Corollary}[section]
\newtheorem{remark}{Remark}[section]
\newtheorem{open problem}{Open Problem}
\theorembodyfont{\normalfont}

\theoremsymbol{{$\square$}}
\newtheorem*{proof}{Proof}

\numberwithin{equation}{section}

\usepackage[colorlinks=true,linkcolor=blue,citecolor=blue,urlcolor=blue]{hyperref}
\usepackage{rotating}
\usepackage{makecell}
\usepackage{graphicx}

\title{A Lagrangian Approach to Optimal Randomization\thanks{An earlier version of this work, which this paper supersedes, circulated under the title ``A Lagrangian Approach to Optimal Lotteries in Non-Convex Economies''. We thank Matias Bayas-Erazo, Timo Boppart, Hal Cole, Albert Marcet, Ben Moll, Nicola Pavoni, Ned Prescott, Marek Pycia, Simon Scheidegger, Florian Scheuer, Liangjie Wu, and seminar participants at BSE Summer Forum, CRETA conference at Warwick, PKU, SED, SFI Research Days, and Zurich for helpful discussions. Kubler and Yang thank the Swiss National Science Foundation (\#10003091) for financial support. Shen and Zhou thank the support by the National Key R\&D Program of China with project number 2021YFA1001200, and the NSFC with grant number 12171013.}}
\author{Chengfeng Shen\footnote{Peking University. Email: \href{mailto:shencf1999@pku.edu.cn}{shencf1999@pku.edu.cn}} \and Felix K\"{u}bler\footnote{University of Zurich. Email: \href{mailto:fkubler@gmail.com}{fkubler@gmail.com}} \and Yucheng Yang\footnote{University of Zurich and SFI. Email: \href{mailto:yucheng.yang@uzh.ch}{yucheng.yang@uzh.ch}} \and Zhennan Zhou\footnote{Westlake University. Email: \href{mailto:zhouzhennan@westlake.edu.cn}{zhouzhennan@westlake.edu.cn}}}
\date{\vspace{0.5cm}
First version: April 22, 2025
\\ \vspace{0.25cm} This version: \today
}

\begin{document}
	\maketitle

 \onehalfspacing
 
\begin{abstract}
We develop an efficient method for solving non-convex constrained optimization problems that are pervasive in economics. The optimal solution to these problems often involves randomization. We employ a Lagrangian framework and prove that the value of the saddle point characterizing the optimal random solution equals the value of the deterministic dual problem. Our algorithm solves this dual via subgradient descent and recovers the optimal random solution directly from deterministic optima computed along the iterations. For many non-convex economic problems, our method is orders of magnitude faster than linear programming, making previously intractable lottery problems feasible. As an application, we solve for optimal Mirrleesian income taxation with multi-dimensional types. We show that heterogeneity in productivity and Frisch elasticity can make randomization welfare-improving over the optimal deterministic schedule.
\end{abstract}

        

\newpage

\section{Introduction}
Many interesting economic questions can be formulated as non-convex constrained optimization problems -- prominent examples include contracting under moral hazard \citep{mirrlees1999theory}, optimal taxation \citep{mirrlees1971exploration,cremer2001direct}, optimal fiscal policy \citep{aiyagari2002optimal}, or macro-prudential regulation \citep{davila2018pecuniary}. In these environments, deterministic solutions are both hard to compute -- first-order conditions are generally insufficient to characterize global optima -- and potentially suboptimal. A large literature has established that randomization can improve welfare when non-convexities are present, with applications ranging from stochastic tax schedules to indivisible labor supply and labor market programs (e.g., \cite{weiss1976desirability}; \cite{myerson1982optimal};  \cite{arnott1988randomization}; \cite{rogerson1988indivisible}; \cite{phelan1991computing}; \cite{brito1995randomization}; \cite{pavoni2007optimal}; \cite{garratt2023efficient}; \cite{citanna2024taxspots}).\footnote{Throughout the paper, we use the terms ``lottery solution'' and ``randomized solution'' interchangeably to refer to the optimal, possibly random, solution of a nonlinear optimization problem.}
Although explicit lottery contracts are uncommon, there are no economic arguments to exclude contracts involving a lottery. Many real-world institutions implement closely related devices, including stochastic audits and enforcement, randomized eligibility or assignment rules, and tie-breaking procedures. 
It is then a quantitative question to explore under which conditions randomized contracts can lead to large welfare gains.
  
Since \cite{townsend1987economic}, the standard approach to computing optimal lotteries has been linear programming (LP). 
However, LP methods suffer from a severe curse of dimensionality as the action and consumption space becomes large.
In this paper, we develop a new method to solve for optimal lottery solutions in a wide class of models with non-convexities. 
Our method is based on several key theoretical insights. First, we show that the lottery solution assigns positive probability only to the deterministic allocations that maximize the Lagrangian. Second, we show that this implies that the value of the saddle point of the Lagrangian of the lottery system must be the same as the value of the dual of the deterministic problem. This will be the same as the value of the primal problem if and only if no lottery solution yields a higher value than the deterministic solution. Using the Lagrangian dual to solve constrained optimization problems is a very common approach in applied mathematics (see, e.g., \cite{boyd2004convex}).   By reformulating the primal optimization problem into the Lagrangian dual problem, a difficult non-convex problem can be turned into a simpler problem that can be efficiently solved (see, e.g. \cite{lemarechal2001lagrangian}). We use sub-gradient descent to solve the dual problem iteratively. Our third theoretical insight is that the frequency of different optimal points along the iterations approximates the probability distribution of the optimal lottery solution. As in fictitious play for normal form games (see, e.g., \cite{fudenberg1998theory}), the (weighted) empirical distribution of past actions converges to the optimal distribution in the lottery solution. 
Computationally, this approach avoids optimizing over the high-dimensional probability simplex and instead repeatedly solves a deterministic Lagrangian subproblem that, in many economic environments, inherits decomposability and partial concavity, delivering large speed and memory gains over linear programming.

We first illustrate the method in a textbook principal agent problem with moral hazard as in \citet[Chapter 5]{salanie2005economics}. We use this simple example to give an intuitive explanation for our method: the Lagrangian maximization step selects a deterministic contract, the multipliers move to enforce incentive and participation constraints, and the algorithm cycles among a small set of contracts whose weighted frequencies converge to the optimal lottery. The example also clarifies the source of our computational advantage. The Lagrangian subproblem decomposes across outcome states and can exploit curvature in consumption, whereas the linear-programming formulation must optimize over an exponentially large space of lotteries.

Our main application studies optimal Mirrleesian taxation with multi-dimensional hidden types. 
Multidimensional heterogeneity is central to optimal taxation, yet it poses fundamental challenges: the first-order approach in one-dimensional screening, which relies on single-crossing and monotonicity to reduce global incentive constraints to local ones, does not extend to higher dimensions. All incentive constraints must be enforced globally, and the computation of optimal mechanisms becomes markedly harder \citep{judd2017optimal}.
We consider a setting with standard preferences over consumption and leisure, with heterogeneity in both productivity and labor supply elasticity. We show that nondegenerate optimal randomization arises naturally in this environment. Our Lagrangian iteration approach makes the problem computationally tractable and yields new economic predictions: optimal randomization reduces bunching, smooths consumption toward the full-information benchmark, and can reduce the welfare cost of private information substantially.
Since the Lagrangian decomposes across types, the complexity of our algorithm increases only polynomially in the number of types. This is in stark contrast to global constrained optimization methods for non-convex problems that are generally known to be NP-hard (e.g. \cite{sahni1974computationally}). The lottery relaxation together with a {\it decomposability} property makes the problem feasible.
Using the calibration from \cite{judd2017optimal} with 25 types and 600 incentive constraints, we illustrate the efficiency of our method and clarify the roles that lotteries play in the optimal solution.

\paragraph{Related Literature.} The seminal work by \cite{myerson1982optimal} established that allowing lotteries in principal-agent problems can enhance welfare and improve tractability.
Prescott and Townsend's (\cite{prescott1984general}, \cite{prescott1984pareto}) key methodological insight 
was to treat each individual’s choice set as the set of lotteries over allocations, which makes prices linear in lotteries.
This restores convexity and allows standard competitive equilibrium analysis to be applied to economies with moral hazard and adverse selection problems.
 
To investigate the potential quantitative importance of randomized contracts, one needs to be able to compute optimal contracts with lotteries. 
Most of the existing approaches for solving lottery problems in models with private information focus on improving linear programming techniques. For example, \citet{prescott2004computing} exploits the block-triangular structure of the constraint matrix and applies Dantzig-Wolfe decomposition to reduce memory and computational cost, while \citet{doepke2006dynamic} proposes additional simplifications tailored to specific environments. An alternative approach by \citet{su2007computation}, based on mathematical programming with equilibrium constraints (MPEC), improved the numerical accuracy in small-scale moral hazard problems. Our method advances the literature by introducing a novel solution method that transcends the linear programming framework, enabling the analysis of more complex and computationally demanding models that were previously infeasible.  Our Lagrangian iteration method offers significant computational advantages compared to traditional linear programming approaches. 

Our main application contributes to the literature on optimal tax design under private information. Previous work has shown that lotteries can arise in optimal tax schedules with one-dimensional heterogeneity (e.g. \cite{weiss1976desirability}, \cite{brito1995randomization}, \cite{gauthier2014value}). 
With one-dimensional heterogeneity in preferences over consumption, in the plausible case of decreasing absolute risk aversion, equilibrium lotteries could also be degenerate (see, e.g. \cite{kehoe2002lotteries} or \cite{hellwig2007undesirability}). More recent studies, such as \cite{judd2017optimal}, \cite{moser2019optimal}, and \cite{boerma2022bunching}, analyze optimal deterministic tax systems with multidimensional types. Our framework allows for a general analysis of optimal taxation with multidimensional heterogeneity and reveals that lotteries naturally arise for a range of parameter values.

It is established in convex analysis that the Lagrangian dual problem corresponds to the closed convexification of the perturbed primal problem (see Chapter XII in \cite{urruty1993convex}, and Theorem 2.2 and Lemma 2.1 in \cite{bi2020analysis}). In this paper, we demonstrate that the lottery relaxation provides a means of achieving this convexification and thereby provide economic foundations for convexification.
Lastly, our method builds on the subgradient descent literature for nondifferentiable convex optimization \citep{bertsekas2009convex,shor2012minimization,nedic2001convergence}, and combines it with a key theoretical insight: the value of the saddle point characterizing the optimal lottery equals the value of the dual of the deterministic problem. 
This connection allows us to construct the optimal lottery from the sequence of deterministic solutions generated along the Lagrangian iterations. The approach is related in spirit to algorithms used to compute mixed-strategy equilibria in mean-field games \citep[e.g.,][]{shen2023fictitious}, but is tailored to the structure of constrained planning problems.

\paragraph{Roadmap.} In  Section \ref{sec:method}, we describe the general setup, introduce the Lagrangian iteration method, and state our theorems. In Section \ref{sec:comper} we discuss the computational performance of our method and compare it with existing methods. In Section \ref{sec:pri-ag}, we illustrate the method using a textbook principal-agent problem. 
In Section \ref{sec:Mirrleesian}, we apply the method to an optimal tax problem with multi-dimensional heterogeneity. Section \ref{sec:concl} concludes. The proofs are collected in the Appendix \ref{app:proofs}. In Appendix \ref{app:complex}, we provide worst-case complexity bounds for our method and compare them to bounds from linear programming.

\section{Lagrangian Iteration}\label{sec:method}
This section introduces a general class of non-convex planning problems and presents the Lagrangian iteration method for computing the optimal lottery. We begin by stating the maximization problem in Section~\ref{sec3:general}. We then establish a key connection between the lottery problem and the deterministic Lagrangian in Section~\ref{sec3:Lagrange}, which motivates our approach and provides the basis for the algorithm. Next, Section~\ref{sec3:algo} presents the algorithm and Section~\ref{sec:theorem} proves its convergence properties.  Readers primarily interested in the method and its implementation may focus on Sections~\ref{sec3:general} and~\ref{sec3:algo}.
\subsection{General Setup}\label{sec3:general}
In this section, we present a general constrained optimization problem that subsumes many planning problems considered in economics. Let $A$ denote a finite set of actions, and let $C$ be a $n$-dimensional box of consumption vectors in the Euclidean space $\mathbb{R}^n$. 
The payoff function $f:A\times C \rightarrow \mathbb{R}$ is continuous. 
In the deterministic problem, the social planner allocates action and consumption to solve the following optimization problem:
\begin{equation}\label{math_pure_mh}
\begin{aligned}
&\max_{a \in A,\, c \in C} f(a, c),\\
\textbf{s.t. } &g_i(a, c) \le 0 \quad i \in \{1, \dots, m\},\\
&h_j(a, c) \le 0 \quad j \in \{1, \dots, \ell\},
\end{aligned}
\end{equation}
where the functions $ g_i(.), i=1,\ldots,m $, and $ h_j (.), j=1,\ldots,\ell $ are assumed to be continuous in $c$.\footnote{The reason why we distinguish between $g$-constraints and $ h$-constraints and why $A$ is assumed to be finite will become clear when we describe the lottery problem.}
In the rest of the paper, we refer to \eqref{math_pure_mh} as the \textit{deterministic problem}. A central challenge in such problems arises from the fact that the set of allocations that satisfy the constraints is generally non-convex. This non-convexity makes it computationally challenging to solve the problem since first order conditions can generally not be used. In examples where the set $C$ can also be taken to be approximated by a finite set, an optimal solution can be found by searching over all elements in $ A \times C $, but the computational costs of this method increase exponentially in the dimension of $ A \times C $.

Motivated by these difficulties, economists have explored the concept of  \textit{lottery solutions}, where the planner chooses a probability distribution over actions and consumption rather than a deterministic allocation \citep{prescott1984general, myerson1982optimal,arnott1988randomization}. This approach offers a way to relax the non-convexities and can improve the value of the objective function.
In this setting, instead of selecting a deterministic action-consumption pair from \(A \times C\) as in \eqref{math_pure_mh}, the social planner chooses a probability distribution over the set of possible outcomes, \(A \times C\). The economic model then dictates which constraints can be relaxed to having to hold only in expectations and which constraints have to hold for each realization of the lottery. In our abstract formulation, assume that the $g-$constraints can be relaxed while the $h-$constraints have to hold for each realization of $ a \in A$.

Formally, let \( \mathcal{P}(A \times C) \) denote the set of Borel probability measures over \(A \times C\). 
An element of \( \mathcal{P}(A \times C) \) is denoted by \(x(a, dc)\), which specifies the probability of choosing an action \(a \in A\) and allocating consumption \(c \in C\). Here, we assume that
the variable $a$ is discrete and the variable $c$ is continuous. 
The objective of the social planner in the lottery problem is to maximize the expected payoff, which can be written as:
$$
\sum_{a\in A}\int_{c\in C}f(a,c)x(a,dc).
$$
The $g$-constraints from the original maximization problem are now assumed to hold in expectations, i.e.
\[
\sum_{a \in A} \int_{c \in C} g_i(a, c) x(a, dc) \leq 0 \quad \forall i \in \{1, \dots, m\}.
\]
In economic models, one typical justification for this assumption is that there is a large number of identical agents and the $g-$constraints describe the resource constraint.
 In contrast, we assume that the $h$-constraints have to be satisfied for each action $a$ that has positive probability, i.e. 
for each action, $a$, in the support of $x$, the $h$-constraints must be satisfied in expectation over consumption.
\[
\int_{c \in C} h_j(a, c) x(a, dc) \leq 0 \quad \forall j \in \{1, \dots, \ell\}.
\]
Note that allowing for continuous actions $ a \in A $ would lead to measure-theoretic problems since classic measure theory defines marginal integration and addresses problems in terms of almost every $a\in A$ when $A$ is continuous, rather than considering every $a\in A$, which is required in our problem.

Thus, the lottery optimization problem for the social planner can be written as follows.
\begin{equation}\label{math_lot_mh}
\begin{aligned}
    &\max_{x \in \mathcal{P}(A \times C)} \sum_{a \in A} \int_{c \in C} f(a, c) x(a, dc), \\
    \text{s.t. } &\sum_{a \in A} \int_{c \in C} g_i(a, c) x(a, dc) \leq 0 \quad \forall i \in \{1, \dots, m\}, \\
    &\int_{c \in C} h_j(a, c) x(a, dc) \leq 0 \quad \forall a \in A, \, j \in \{1, \dots, \ell\}.
\end{aligned}
\end{equation}

Note that in this formulation, we do not require the marginal distributions of $ x \in {\mathcal P}(A \times C)$ across the components $A$ and $C$ to be  independent. This requirement is, of course, an important
part of the definition of a mixed strategy Nash equilibrium, where allowing for correlation significantly alters the set of equilibria; see \cite{aumann1974subjectivity}. 
In a social planning framework, it seems that not allowing for correlation would artificially restrict the set of possible contracts (see also \cite{myerson1982optimal} for a discussion in the context of the principal-agent problem). In the examples below, correlation is never used in the optimal solution, but this is a result and not an assumption.

\subsection{A Lagrange Approach}\label{sec3:Lagrange}

The essence of our method is to establish a connection between the deterministic system \eqref{math_pure_mh} and the lottery system \eqref{math_lot_mh}.
Given Lagrange parameters $\lambda,\gamma$, we define the Lagrangian function in the probability space $\mathcal{P}(A\times C)$ as
    \begin{equation}\label{defL}
    \begin{aligned}
    &L(x;\lambda,\gamma):=\\
    &\sum_{a\in A}\int_{c\in C}f(a,c)x(a,dc)-\sum_{i=1}^{m}\lambda_i\sum_{a\in A}\int_{c\in C}g_i(a,c)x(a,dc)-\sum_{j=1}^{\ell}\sum_{a\in A}\gamma_{j,a}\int_{c\in C}h_{j}(a,c)x(a,dc),
    \end{aligned}
    \end{equation}
We define the Lagrangian function in the pure strategy space $A\times C$ as
    \begin{equation}\label{defcalL}
\mathcal{L}(a,c;\lambda,\gamma):=f(a,c)-\sum_{i=1}^{m}\lambda_i g_i(a,c)-\sum_{j=1}^{\ell}\gamma_{j,a}h_j(a,c).
\end{equation}
The main result of this subsection is that under the assumption that Slater's condition holds, we have the following equivalence
\begin{equation}
    \label{dual-equivalence}
\max_{x \in {\cal P}(A\times C)} \min_{(\lambda,\gamma)\in\mathbb{R}_+^{m}\times \mathbb{R}_+^{\ell|A|}} L(x;\lambda,\gamma)=\min_{(\lambda,\gamma)\in\mathbb{R}_+^{m}\times \mathbb{R}_+^{\ell|A|}} \max_{a\in A,\,c\in C}\mathcal{L}(a,c;\lambda,\gamma),
    \end{equation}
bridging the gap between the lottery problem (the left-hand side) and the dual of the deterministic problem (the right-hand side).

\subsubsection{Prerequisites for a Lagrange approach}
One well-known obstacle to solving the deterministic problem \eqref{math_pure_mh} is that standard constraint qualifications, which are necessary for applying optimization methods from numerical analysis, often fail to hold at the solution (see, e.g., \cite{su2007computation} or \cite{judd2017optimal}). 
We therefore first demonstrate the validity of the Lagrange multiplier method for the lottery system \eqref{math_lot_mh}. This is a linear programming problem in $\mathcal{P}(A\times C)$ and hence a convex optimization problem in $\mathcal{P}(A\times C)$.
According to the Karush-Kuhn-Tucker Theorem, a convex optimization problem that satisfies Slater's conditions can be solved by the Lagrange multiplier method (see Theorem 1 in Section 8.3 in \cite{luenberger1997optimization} for details).\footnote{Because the decision variable is a probability measure, \eqref{math_lot_mh} is an infinite-dimensional linear program. In finite dimensions, strong duality for linear programs can be established without Slater-type conditions. In infinite dimensions, a constraint qualification is typically invoked to ensure existence of multipliers.}
The definition of the condition is as follows.
\begin{definition}\label{def:slater}{(Slater's condition)}
    We say that Slater's condition holds for the lottery system \eqref{math_lot_mh}, if the feasible set to \eqref{math_lot_mh} includes one inner point, i.e.,  there exists $x\in \mathcal{P}(A\times C)$ such that 
    $$
    \sum_{a\in A}\int_{c\in C}g_i(a,c)x(a,dc)<0,\quad \forall i\in\{1,\cdots,m\},
    $$
    and
    $$
    \int_{c\in C}h_{j}(a,c)x(a,dc)< 0, \quad \forall j\in\{1,\cdots, \ell\},\,a\in A.
    $$
    Note here both inequalities must hold strictly.
\end{definition}

 It is standard to show (see, e.g., Theorem 1 in Section 8.3 in \cite{luenberger1997optimization}\footnote{We need to verify the existence of the maximizer of the Lagrangian in order to apply the Theorem in \cite{luenberger1997optimization}. This follows directly from the fact that $\mathcal{P}(A\times C)$ is compact with respect to the weak* topology. The existence of the maximizer can then be deduced by the standard compactness argument. }) that Slater's condition implies that
 the solution of problem \eqref{math_lot_mh} $x^*\in \mathcal{P}(A\times C)$ exists, and that there exist Lagrangian multipliers $\lambda^*_i, \gamma^*_{j,a}\ge 0$ ($i\in\{1,\cdots,m\},\, j\in\{1,\cdots,l\},\,a\in A$)  such that
 $x^*$ is the maximizer of 
        $L(x;\lambda^*,\gamma^*)$.
Generally, it is difficult to verify whether Slater's condition holds in \eqref{math_lot_mh}. However, if we relax the constraints in \eqref{math_lot_mh} to 
\begin{equation}
\label{math_lot_mh_relax1}
\sum_{a\in A}\int_{c\in C}g_i(a,c)x(a,dc)\le \epsilon,\quad \forall i\in\{1,\cdots,m\},
\end{equation}
and
\begin{equation}
\label{math_lot_mh_relax2}
\int_{c\in C}h_{j}(a,c)x(a,dc)\le  \epsilon, \quad \forall j\in\{1,\cdots, \ell\},\,a\in A,
\end{equation}
for some $\epsilon>0$, then any feasible measure $x\in\mathcal{P}(A\times C)$ for the original lottery system \eqref{math_lot_mh} will strictly satisfy these two relaxed constraints. 
We therefore consider the following relaxed system.
\begin{equation}\label{math_lot_mh_relax}
    \max_{x \in \mathcal{P}(A \times C)} \sum_{a \in A} \int_{c \in C} f(a, c) x(a, dc) \mbox{ s.t. }
(\ref{math_lot_mh_relax1})  \mbox{ and }
    (\ref{math_lot_mh_relax2})
    \end{equation}

The Slater's condition for \eqref{math_lot_mh_relax} is then satisfied when there is at least one feasible point for the original lottery system \eqref{math_lot_mh}. We denote by $x^\epsilon$ the optimal solution to \eqref{math_lot_mh_relax}. The following theorem shows that $x^\epsilon$ will converge to $x^*$, which is the optimal solution to the original lottery system \eqref{math_lot_mh}.
\begin{theorem}\label{thm:relax}
    Assume that there exists at least one feasible point for the original lottery problem \eqref{math_lot_mh}, hence for any $\epsilon>0$, the solution to the relaxed lottery problem \eqref{math_lot_mh_relax} exists, denoted as $x^\epsilon$. Then we can choose a sequence $\{\epsilon_n\}\rightarrow 0$, s.t. $x^{\epsilon_n}$ converges in the weak* topology in the finite Borel measures space on $A\times C$, denoted as $\mathcal{M}(A\times C)$, to some $x^*\in \mathcal{P}(A\times C)\subset\mathcal{M}(A\times C)$ as $n\rightarrow \infty$,  i.e. for any $\varphi\in C^0(A\times C)$, we have
    $$
    \sum_{a\in A}\int_{c\in C}\varphi(a,c)x^{\epsilon_n}(a,dc)\rightarrow  \sum_{a\in A}\int_{c\in C}\varphi(a,c)x^*(a,dc), \text{\quad as }n\rightarrow
    \infty.
    $$
    Furthermore, $x^*$ is an optimal solution to \eqref{math_lot_mh}.
\end{theorem}
\begin{proof}
    See Appendix \ref{app:thm:relax}.
\end{proof}

Theorem \ref{thm:relax} shows that one can approximate the solution to system \eqref{math_lot_mh} by solving the relaxed problem \eqref{math_lot_mh_relax} with a sufficiently small $\epsilon>0$. Importantly, the existence of a feasible point for the original system \eqref{math_lot_mh} — a very mild assumption — implies that the relaxed problem satisfies Slater's condition, which in turn guarantees the existence of Lagrangian multipliers. Hence, the Lagrange multiplier method applies to the relaxed problem. Consequently, in the remainder of this paper, we directly assume that Lagrange multipliers exist for the original system \eqref{math_lot_mh}. This assumption is not restrictive: if it fails, we can instead rely on the existence of a feasible point and compute an approximate solution via the $\epsilon$-relaxed problem, as guaranteed by Theorem \ref{thm:relax}. Unfortunately, as is typically the case in computational economics, we cannot establish a direct relationship between the $\epsilon$ solutions and the exact solution.

 \subsubsection{The Relation between the Lottery Problem and the Deterministic Problem} 

To understand why (\ref{dual-equivalence}) holds, we first establish that for given $ \lambda,\gamma$, the optimal solutions to the maximization of $L(.;\lambda,\gamma) $ and $ {\cal L}(.;\lambda,\gamma)$ have the same value in the following theorem.

\begin{theorem}\label{thm:simp_L} Given $\lambda\ge 0,\,\gamma\ge 0$. Let $L$ and $\mathcal{L}$ be defined as in \eqref{defL} and \eqref{defcalL} respectively. Then we have
\begin{equation}\label{eq:thm:sim_L_1}
\max_{x\in \mathcal{P}(A\times C)}L(x;\lambda,\gamma)=\max_{a\in A,c\in C}\mathcal{L}(a,c;\lambda,\gamma).
\end{equation}
Furthermore, if we define $Z=\arg\max_{a\in A,c\in C}\mathcal{L}(a,c;\lambda,\gamma)$, then  $$x^*\in \arg\max_{x\in\mathcal{P}(A\times C)}L(x;\lambda,\gamma)$$ if and only if the measure of $Z^c$ with respect to $x^*$ is zero, i.e.
\begin{equation}\label{eq:thm:sim_L_2}
x^*(Z^c)=0.
\end{equation}
\end{theorem}
\begin{proof}
    See Appendix \ref{app:thm:simp_L}.
\end{proof}

Similar to a mixed-strategy Nash equilibrium, here an optimal solution only puts weights on different actions that yield the same value of the objective function.
This implies directly that
$$ \min_{\lambda,\gamma} \max_{x\in \mathcal{P}(A\times C)}L(x;\lambda,\gamma)= \min_{\lambda,\gamma}\max_{a\in A,c\in C}\mathcal{L}(a,c;\lambda,\gamma). $$
Since the lottery problem is convex and strong duality holds, this implies that the lottery problem and the dual problem of the deterministic problem are essentially the same problem. To formalize this, it is useful to explicitly define the solution to the dual problem. 
\begin{definition}\label{def_dual_problem}
    We define the dual problem of the deterministic problem \eqref{math_pure_mh} as
    \begin{equation}\label{dual_problem}
        \inf_{(\lambda,\gamma)\in\mathbb{R}_+^{m}\times \mathbb{R}_+^{\ell |A|}} \max_{a\in A,\,c\in C}\mathcal{L}(a,c;\lambda,\gamma),
    \end{equation}
    where $\mathcal{L}$ is defined in \eqref{defcalL}. If the infimum is attained at $(\lambda^*,\gamma^*)$, then we call $(\lambda^*,\gamma^*)$ a solution to the dual problem \eqref{dual_problem}.
\end{definition}
The following theorem shows that these two problems share the same Lagrangian multipliers and the same optimal value.

\begin{theorem}\label{thm:dual=lot} We assume that $x^*$ is the solution to system \eqref{math_lot_mh}, and the Lagrangian multipliers corresponding to $x^*$ exist, denoted as $(\lambda^*,\gamma^*)\in \mathbb{R}^{m}\times\mathbb{R}^{\ell |A|}$. Then $(\lambda^*,\gamma^*)$ is the solution to the dual problem of the deterministic problem, i.e.
    \begin{equation}\label{dual}
       \max_{a\in A,\,c\in C} \mathcal{L}(a,c;\lambda^*,\gamma^*)=\inf_{(\lambda,\gamma)\in\mathbb{R}_+^{m}\times \mathbb{R}_+^{\ell |A|}} \max_{a\in A,\,c\in C}\mathcal{L}(a,c;\lambda,\gamma),
    \end{equation}
where $\mathcal{L}$ is defined in \eqref{defcalL}. Furthermore, the optimal objective value of the dual problem \eqref{dual_problem} is the same as the optimal objective value of the lottery problem \eqref{math_lot_mh}, i.e.
\begin{equation}\label{dual2}
 \inf_{(\lambda,\gamma)\in\mathbb{R}_+^{m}\times \mathbb{R}_+^{\ell |A|}} \max_{a\in A,\,c\in C}\mathcal{L}(a,c;\lambda,\gamma)=\sum_{a\in A}\int_{c\in C}f(a,c)x^*(a,dc).
\end{equation}
\begin{proof}
    See Appendix \ref{app:dual=lot}.
\end{proof}
\end{theorem}

The theorems imply the following corollary that gives general conditions for nondegenerate lotteries to be optimal. 
\begin{corollary}\label{cor:nondelot}
    The Lagrangian function in $\mathcal{P}(A\times C)$, $L$, admits a nondegenerate lottery maximizer and no deterministic maximizer if and only if one of the following holds.
    \begin{enumerate}
    \item The Lagrangian function in $A\times C$, $\mathcal{L}$, has at least two different maximal points at the optimal $ \lambda^*,\gamma^* $, and none of these points are feasible.
    \item The deterministic saddle point problem does not satisfy strong duality, i.e.
    
$$   \sup_{a\in A,\,c\in C}  \inf_{(\lambda,\gamma)\in\mathbb{R}_+^{m}\times \mathbb{R}_+^{\ell |A|}} \mathcal{L}(a,c;\lambda,\gamma) \ne   \inf_{(\lambda,\gamma)\in\mathbb{R}_+^{m}\times \mathbb{R}_+^{\ell |A|}} \sup_{a\in A,\,c\in C}\mathcal{L}(a,c;\lambda,\gamma). $$
    \end{enumerate}
\end{corollary}

\subsection{The Lagrangian iteration algorithm}
\label{sec3:algo}
 We utilize the results above to construct an algorithm for solving the lottery problem. Theorem \ref{thm:simp_L} implies that for given $ (\lambda,\gamma) $ the Lagrangian optimization problem in the probability space $\mathcal{P}(A\times C)$,
$
\max_{x\in \mathcal{P}(A\times C)}L(x;\lambda,\gamma)
$
can be simplified to the Lagrangian optimization problem in the pure strategy space $A\times C$ and
Theorem \ref{thm:dual=lot} shows how the optimal $ \lambda^*, \gamma^* $ can be obtained from the deterministic problem:
  We define
    \begin{equation}\label{def:dualfunc}
    V(\lambda,\gamma):=\max_{a\in A,\,c\in C}\mathcal{L}(a,c;\lambda,\gamma),
    \end{equation}
and minimize $ V(\cdot) $ via  sub-gradient descent. 
\begin{definition}[Sub-gradient]\label{def:subgradient}
    We consider a convex function $F:\mathbb{R}^{m}_{+}\times \mathbb{R}_{+}^{\ell |A|}\rightarrow\mathbb{R}$. For any $(\lambda,\gamma)\in\mathbb{R}_{+}^{m}\times \mathbb{R}_{+}^{\ell |A|}$, the sub-gradient of $F$ at the point $(\lambda,\gamma)$ is defined as
    $$
    \partial F(\lambda,\gamma)=\{d\in \mathbb{R}^m\times \mathbb{R}^{\ell|A|}|\,F(\lambda',\gamma')\ge F(\lambda,\gamma)+\left(\lambda'-\lambda,\gamma'-\gamma\right)\cdot d,\,\forall (\lambda',\gamma')\in \mathbb{R}_+^m\times \mathbb{R}_+^{\ell|A|}\}.
    $$
\end{definition}
\begin{lemma}\label{lem:subgradient}
  The dual function $V$ defined in \eqref{def:dualfunc} is a convex function. Furthermore, for any $(\lambda,\gamma)\in\mathbb{R}_{+}^{m}\times \mathbb{R}_{+}^{\ell |A|}$ and for any $(a_\lambda,c_{\lambda})$ that maximize $\mathcal{L}(a,c;\lambda,\gamma)$, i.e.
    $$
    (a_{\lambda},c_{\lambda})\in\arg\max_{a\in A,\, c\in C}\mathcal{L}(a,c;\lambda,\gamma),
    $$
    the  following  is a negative sub-gradient of $V$ at $(\lambda,\gamma)$. 
    \begin{equation}\label{def_updating_direction}
    (\Delta \lambda,\,\Delta \gamma):=\begin{cases}
        \Delta \lambda_i= g_i(a_{\lambda},c_{\lambda}),&i\in\{1,\cdots,m\};\\
        \Delta \gamma_{j,a}=h_j(a_{\lambda},c_{\lambda}), &j\in\{1,\cdots,\ell\},\,a=a_{\lambda};\\
        \Delta\gamma_{j,a}=0,&j\in\{1,\cdots,\ell\},\,a\ne a_{\lambda.}
    \end{cases}
    \end{equation}

\end{lemma}
\begin{proof}
    See Appendix \ref{app:lem:sub}.
\end{proof}

Assuming that the Lagrangian multipliers at the $k$-th iteration are denoted as $\lambda^k$ and $\gamma^k$, we first solve the unconstrained optimization problem in $A\times C$
$$(a^k,c^k) \in \arg\max_{a\in A,c\in C}\mathcal{L}(a,c;\lambda,\gamma).$$ 
We then update the multipliers using the projected sub-gradient descent method:
$$
\lambda_i^{k+1}=\max\{\lambda_i^k+\mu^kg_i(a^k,c^k),0\},\,\forall i\in\{1,\cdots,m\};
$$
$$
\gamma_{j,a^k}^{k+1}=\max\{\gamma_{j,a^k}^k+\mu^kh_{j}(a^k,c^k),0\},\,\forall j\in\{1,\cdots,\ell\};
$$
and
$$
\gamma_{j,a}^{k+1}=\gamma_{j,a}^k,\,\forall j\in\{1,\cdots,\ell\},\,a\ne a^k.
$$
These updates follow the negative subgradient of the dual function, with a projection onto the nonnegative orthant to enforce the non-negativity of multipliers as required by the KKT conditions. Intuitively, when a constraint is violated, its associated multiplier increases, tightening the penalty; when satisfied, the multiplier relaxes.

Finally, we collect all these pure strategies $(a^k,c^k)$ during the iterations to construct an approximate lottery solution as
\begin{equation}\label{eq:lottery_weight}
    x^N:=\frac{1}{\sum_{k=1}^{N}\mu^k}\sum_{k=1}^{N}\mu^k\delta_{(a^k,c^k)}.
\end{equation}
It seems surprising at first that the empirical frequencies of different actions along the iterations converge to the optimal probabilities.
The key to understanding this construction lies in the proof of Proposition  \ref{math_thm_heu_mh}, where it is shown that for every constraint $i$ and iteration $n$ 
$$
\lambda_i^{n+1}\ge \lambda_i^1+\sum_{k=1}^{n}\mu_kg_i(a^k,c^k)= \lambda_i^{1}+\left(\sum_{k=1}^{n}\mu^k\right)\sum_{a\in A}\int_{c\in C}g_i(a,c)x^n(a,dc),
$$
which directly implies that
$$
\sum_{a\in A}\int_{c\in C}g_i(a,c)x^n(a,dc)\le\frac{\lambda_i^{n+1}-\lambda_i^1}{\sum_{k=1}^{n}\mu^k}.
$$
The updating rule for the Lagrangian multipliers can be regarded as estimating the contribution to the constraints for accumulating $(a^k,c^k)$ as a new point in the support of the optimal lottery with the weight $\mu^k$. 
We formally present the full algorithm in Algorithm \ref{math_alg_mh}. \\

\begin{algorithm}[ht!]
\caption{Lagrangian Iteration Algorithm}
\label{math_alg_mh}
Given $\lambda_i^1(i\in \{1,\cdots,m\}),\,\gamma_{j,a}^1(a\in A,\,j\in\{1,\cdots,\ell\}),\,\mu^1,\mu^2,...\in\mathbb{R}_+$, $N\in \mathbb{N}_+$. 

\textbf{For $k=1:N$}

\quad\textbf{Step 1. Solve the Lagrangian problem.}
$$(a^k,c^k)\in\arg \max_{a\in A,c\in C} \mathcal{L}(a,c;\lambda^k,\gamma^k). $$

\quad\textbf{Step 2. Update the Lagrangian multipliers.}
$$
\lambda_i^{k+1}=\max\{\lambda_i^k+\mu^kg_i(a^k,c^k),0\},\,\forall i\in\{1,\cdots,m\}.
$$
$$
\gamma_{j,a^k}^{k+1}=\max\{\gamma_{j,a^k}^k+\mu^kh_{j}(a^k,c^k),0\},\,\forall j\in\{1,\cdots,\ell\}.
$$
$$
\gamma_{j,a}^{k+1}=\gamma_{j,a}^k,\,\forall j\in\{1,\cdots,\ell\}\,,a\ne a^k.
$$
\textbf{End} 

\textbf{Step 3. Construct the lottery solution.}
$$
x^N:=\frac{1}{\sum_{k=1}^{N}\mu^k}\sum_{k=1}^{N}\mu^k\delta_{(a^k,c^k)},
$$
where $\delta_{(a^k,c^k)}$ is the $\delta-$measure at the point $(a^k,c^k)$.
\end{algorithm}

It is important to note that there are $\ell|A|$ Lagrange multipliers associated with the constraints related to unobserved actions, as the system \eqref{math_lot_mh} requires the incentive constraints to be satisfied for all $a\in A$. However, in each iteration, we update only $\ell$ of these multipliers, specifically those corresponding to the form $\gamma_{\cdot, a^k}$. 

A key part of the algorithm consists of the maximization problem 
$$(a^k,c^k)\in\arg \max_{a\in A,c\in C} \mathcal{L}(a,c;\lambda^k,\gamma^k). $$
Since it is assumed that the set $A$ is finite, the feasibility of the algorithm is based on solving
$\max_{c\in C} \mathcal{L}(\bar a,c;\lambda^k,\gamma^k)$ for a given action $ \bar a \in A$. This is generally a nonconvex problem, and in the worst-case scenario, one needs to discretize the set $ C $ and search over all $ A \times C $. In this case, the algorithm's computational efficiency is comparable to linear programming (see Appendix \ref{app:complex}).
However, as we argue in Section \ref{sec:comper}, a central claim of this paper is that this maximization step is typically much more tractable in economic applications because standard economic structure can be exploited to solve it efficiently.

\subsection{Theoretical Properties of the Algorithm}\label{sec:theorem}

In this section, we establish that, under appropriate conditions on the learning rates $\mu^k$, the lottery solution $x^N$ generated by Algorithm \ref{math_alg_mh} provides an approximation to the optimal lottery solution, with a precise definition of ``approximation'' to follow. 
From our results above and general convergence results for subgradient algorithms (see, e.g. \cite{nedic2001convergence}), we can expect our algorithm to have the following convergence properties. The Lagrangian multiplier $(\lambda^k,\gamma^k)$ during iterations converges to the optimal Lagrange multiplier, $ (\lambda^*,\gamma^*)$; the value of the objective function for the constructed lottery converges to the optimal value of the objective function; and the constructed lotteries $ x^n $ converge to the optimal lottery. These results will be established formally in this section.

For any $\epsilon>0$, we first define the $\epsilon$-optimal solution to the system \eqref{math_lot_mh}, to describe a probability measure on $A\times C$, that approximately satisfies all the constraints, and approximately attains the maximal objective function value with error $\epsilon$. 
\begin{definition}{($\epsilon$-optimal solution)}\label{def:eps_opt_sol}
We denote by $x^*$ the optimal solution to system \eqref{math_lot_mh}. We call  $\tilde{x}^\epsilon\in\mathcal{P}(A\times C)$ an $\epsilon$-optimal solution to system \eqref{math_lot_mh}, if it satisfies the relaxed constraints (\ref{math_lot_mh_relax1}) and (\ref{math_lot_mh_relax2}) and if 
 \begin{equation}\label{eq:property3}
 \sum_{a\in A}\int_{c\in C}f(a,c)\tilde{x}^{\epsilon}(a,dc)-\sum_{a\in A}\int_{c\in C}f(a,c)x^*(a,dc)\ge -\epsilon.
 \end{equation}
\end{definition}

Note that $\tilde{x}^\epsilon$ is different from the solution to the relaxed system \eqref{math_lot_mh_relax}, since we do not require $\tilde{x}^{\epsilon}$ to exactly attain the maximal objective function value among all feasible probability measures for \eqref{math_lot_mh_relax}.  However, the property \eqref{eq:property3} in this definition actually implies that $\tilde{x}^{\epsilon}$ can be regarded as an approximate optimal solution to \eqref{math_lot_mh_relax} and we have the following analog to Theorem \ref{thm:relax}.
\begin{theorem}\label{thm:eps_optimal}
     We denote by $\tilde{x}^\epsilon$ an $\epsilon$-optimal solution to the system \eqref{math_lot_mh}. Then we can choose a sequence $\{\epsilon_n\}\rightarrow 0$, s.t. $\tilde{x}^{\epsilon_n}$ converges in the weak* topology in the finite Borel measures space on $A\times C$, denoted as $\mathcal{M}(A\times C)$, to some $\tilde{x}^*\in \mathcal{P}(A\times C)\subset \mathcal{M}(A\times C)$, as $n\rightarrow \infty$, i.e. for any $\varphi\in C^0(A\times C)$, we have
    $$
    \sum_{a\in A}\int_{c\in C}\varphi(a,c)\tilde{x}^{\epsilon_n}(a,dc)\rightarrow  \sum_{a\in A}\int_{c\in C}\varphi(a,c)\tilde{x}^*(a,dc),\,\text{as }n\rightarrow\infty.
    $$
    Furthermore, $\tilde{x}^*$ is an optimal solution to \eqref{math_lot_mh}.
\end{theorem} 
\begin{proof}
    See Appendix \ref{app:eps_optimal}.
\end{proof}

We now examine the property of the lottery solution $x^N$ generated by Algorithm \ref{math_alg_mh} and have the following result.
\begin{theorem}\label{math_thm_final_mh}
    We assume that the sequence $(\mu^k)_{k=1}^\infty$ satisfies
$$
\sum_{k=1}^{\infty}\mu^k=\infty, \text{\quad and\quad}\sum_{k=1}^{\infty}(\mu^k)^2<\infty.
$$
Let $x^*$ be the solution to system \eqref{math_lot_mh}, and suppose the corresponding Lagrangian multipliers to $x^*$ exist. Then for any $\epsilon>0$, there exists $N\in \mathbb{N}_+$, such that when $n>N$, $x^n$ obtained in Algorithm \ref{math_alg_mh} is an $\epsilon-$optimal solution to system \eqref{math_lot_mh}.
\end{theorem}
\begin{proof}
    See Appendix \ref{app: math_thm_final_mh}.
\end{proof}
Theorem \ref{math_thm_final_mh} demonstrates that, under the appropriate assumptions on the learning rates $\mu^k$, the lottery we constructed by the algorithm is an approximation of the optimal lottery.

The final theoretical result derives a bound on the number of iterations needed to reach a given accuracy of the solution.
Let $\lambda=(\lambda_i)_{i\in\{1,\cdots,m\}}$ and $\gamma=(\gamma_{j,a})_{j\in\{1,\cdots,\ell\},\,a\in A}$. We define the function
$$
\Lambda(\lambda, \gamma) = \sum_{i=1}^{m} \lambda_i^2 + \sum_{j=1}^{\ell}\sum_{ a \in A} \gamma_{j, a}^2.
$$

Suppose that all the assumptions in Theorem \ref{math_thm_final_mh} hold, then  $(\lambda^k,\gamma^k)$ converges and hence remains bounded during the iterations in Algorithm \ref{math_alg_mh}. Additionally, we assume there exist two constants $M\ge 0, \bar{\Lambda}\ge 0$ such that $|g_i(a,c)|\le M,|h_{j}(a,c)|\le M$ and \(\|(\lambda^k,\gamma^k)\|_{\infty}+\Lambda(\lambda^1,\gamma^1)\le \bar{\Lambda} \) throughout the iterations. We use the symbol $x\succsim y$ to denote $x$ is greater than or equal to a constant multiple of $y$. The following proposition provides a rough estimate of the number of iterations required to obtain an $ \epsilon $-solution.
\begin{prop}\label{thm_complexity}
    We take all the assumptions in Theorem \ref{math_thm_final_mh}. Additionally, we assume that there exist two constants $M\ge 0, \bar{\Lambda}\ge 0$ such that $\|g_i\|_{\infty}\le M\,(i=1,\cdots,m),\,\|h_{j}\|_{\infty}\le M\,(j=1,\cdots,\ell),\,$ and \(\|(\lambda^k,\gamma^k)\|_{\infty}+\Lambda(\lambda^1,\gamma^1)\le \bar{\Lambda} \) throughout the iterations. We take $\mu^k\sim k^{-\frac{1}{2}(1+\rho)}$ for $0<\rho\le 1$. For $\epsilon>0$, if
    \begin{equation}\label{eq:n_itertime_mh}
n \succsim 
\begin{cases}
\left(\frac{M(\frac{1}{\rho}+\bar{\Lambda})}{\epsilon}\right)^{\frac{2}{1 - \rho}}(m+\ell)^{\frac{1}{1-\rho}}, & \rho < 1 \\
e^{\frac{M(\frac{1}{\rho}+\bar{\Lambda})\sqrt{m+\ell}}{\epsilon}}, & \rho = 1,
\end{cases}
\end{equation}
then the $x^n$ obtained from Algorithm \ref{math_alg_mh} is an $\epsilon$-optimal solution.
\end{prop}
\begin{proof}
    See Appendix \ref{app:thm_complexity}.
\end{proof}
From Equation (\ref{eq:n_itertime_mh}) it appears as if 
the learning rate parameter $ \rho $  should be set relatively close to zero. However,
by Proposition 2.8 in \cite{nedic2001convergence}, if we additionally assume that the dual problem is strongly convex in a neighborhood of the optimal Lagrangian multipliers, then for any \( \epsilon>0 \), we can select a learning rate satisfying \( \mu^k\sim O(1/k) \)  (i.e. with $ \rho=1$) such that when  
\begin{equation}\label{ncon:sub_gradient_conv}
n\succsim \frac{M^2(m+\ell)}{\epsilon^2},
\end{equation}
it follows that \( \text{dist}((\lambda^n,\gamma^n),\Lambda^*)<\epsilon \). 

In general, it is important to point out that sub-gradient descent is generally much less efficient than gradient descent applied to convex and smooth problems. In the latter case, the algorithm can be refined by using line-search or other methods and, in general, the number of iterations needed to achieve precision $ \epsilon$ is of the order $ \frac{1}{\epsilon}$ or even $ \frac{1}{\sqrt{\epsilon}} $ (see \cite{shor2012minimization}). If the optimal solution is deterministic, our problem becomes smooth (at least in a neighborhood of the optimal solution) and there are more efficient ways to find the optimal solution. In this paper, we focus mainly on the more challenging general case where the optimal solution involves lotteries.

\section{Computational Performance}\label{sec:comper}
In this section, we explain how the performance of our algorithm is superior to existing methods for solving economic problems with non-convexities. We first discuss worst-case complexity bounds for grid search, linear programming, and our method. We then explain how economic problems often contain additional structure that can only be exploited by our method. Finally, we provide a classic economic example to illustrate our claims.

\subsection{Computational Performance without Additional Structure}
It is well established at least since \cite{sahni1974computationally} that without any further assumptions, solving the maximization problem (\ref{math_pure_mh}) for a deterministic solution 
is not computationally tractable.
Linear programming methods are typically used to solve for a lottery solution to (\ref{math_pure_mh}), but there is little gain in terms of computational efficiency without any additional structure. To be precise, to use linear programming, we first need to discretize the set $C$ into a finite grid $\hat{C}$.  This transforms \eqref{math_lot_mh} into a finite-dimensional problem:
\begin{equation}\label{math_lot_mh_lp1}
\begin{aligned}
    &\max_{\hat{x}\in \mathcal{P}(A\times \hat{C})}\sum_{a\in A}\sum_{\,c\in \hat{C}}f(a,c)\hat{x}(a,c),\\
    \textbf{s.t. }&\sum_{a\in A}\sum_{c\in \hat{C}}g_i(a,c)\hat{x}(a,c)\le0\text{ }(i\in\{1,\cdots,m\}),\\
    &\sum_{c\in \hat{C}} h_j(a,c)\hat{x}(a,c)\le0 \text{ }(j\in\{1,\cdots,\ell\},\,a\in A).
\end{aligned}
\end{equation}
This formulation constitutes a standard linear programming problem, where $\hat{x}$ is a vector of dimension $|A||\hat{C}|$, and the number of inequality constraints is $\ell|A| + m$. 
The deterministic problem can be approximately solved by exhaustive search over $ A \times \hat{C} $, while using linear programming to obtain the lottery solution means solving a linear program with $ |A| |\hat{C}|$ variables -- this is significantly more computationally intensive than grid search. In applications where a very fine discretization of $ C$ is needed to obtain an economically meaningful solution, neither method is tractable. Furthermore, if $C$ is a high-dimensional set, and one wants to maintain the same number of grid points along each dimension of $C$, both methods' complexity is obviously exponential in the dimension of $C$.\footnote{This is true even though the complexity of linear programming is polynomial in the number of variables.}
 
For the abstract mathematical problem, our Lagrangian algorithm exhibits better worst-case complexity bounds than the linear programming approach. The number of iterations required to obtain a given accuracy $ \epsilon $ is bounded in Proposition \ref{thm_complexity} and does not depend on the number of variables, but increases polynomially in the number of constraints. The overall complexity of our algorithm can then be obtained by combining the required number of iterations (Proposition \ref{thm_complexity}), and a straightforward estimate of the complexity of the grid search in Step 1 -- see Appendix \ref{sect_Bcompare} for details.  However, without further structure, our algorithm does not perform significantly better in practical computations. For finite $ \hat{C}$, the grid search in Step 1 of
Algorithm \ref{math_alg_mh} suffers from the same curse of dimensionality as do linear programming and exhaustive search for a deterministic solution.

One of the key insights of this paper is that typical economic problems contain additional structure which stems from the standard assumption that the utility and production functions are concave and that helps overcome the curse of dimensionality for the Step 1 of our algorithm.

\subsection{Computational Performance with Additional Structure}
Two key properties of many economic problems that can significantly improve the efficiency of our algorithm are \textit{decomposability} and \textit{partial concavity}.

\textit{Decomposability} allows a high-dimensional consumption set $C$ to be decomposed into lower-dimensional components $C = \times_{k=1}^{d} C_k$, with the objective function and all constraint functions being additively separable across these components. This structure enables us to drastically reduce the scale of grid search in Step 1 of our algorithm. Although linear programming can also exploit decomposability to reduce dimensionality (see Appendix \ref{sec: decoLP}), the mechanism is less direct. Importantly, decomposability does not improve efficiency if one attempts to apply grid search to \eqref{math_pure_mh} to solve for a deterministic solution. The formal definition is as follows.
\begin{definition}\label{def:decomposability}
We call an economic problem \textit{decomposable}, if there exists $C_1,\cdots,C_d$, such that $C=\times_{k=1}^{d}C_k$, and for any $c=(c_1,\cdots, c_d)$, we have
$$
f(a,c)=\sum_{k=1}^{d}f_k(a,c_k),\,g_i(a,c)=\sum_{k=1}^{d}g_{i,k}(a,c_k)(\forall 1\le i\le m),\,\,h_{j}(a,c)=\sum_{k=1}^{d}h_{j,k}(a,c_k)(\forall 1\le j\le \ell).
$$
\end{definition}
Decomposability is a common structure in multi-agent models; we illustrate it with specific examples from principal-agent problems in Section \ref{sec:pri-ag} and optimal taxation problems in Section \ref{sec:Mirrleesian}.

\textit{Partial concavity} refers to the property that the Lagrangian of problem \eqref{math_pure_mh} is pseudo-concave in some dimensions of $C$. This permits a first-order approach for those dimensions in Step 1, avoiding the need for discretization. The formal definition is as follows.
\begin{definition}\label{def:partially_concavity}
    We call an economic problem partially concave, if there exists $C_0,\,C_1$, s.t. $C=C_0\times C_1$, and for any $\lambda\ge 0,\,\gamma\ge 0,\,a\in A,\,c_0\in C_0$, the Lagrangian $\mathcal{L}(a,(c_0,c_1);\lambda,\gamma)$ is continuously differentiable and pseudo-concave in $c_1$.\footnote{Following \cite{mangasarian1965pseudo} we say a differentiable function $f: X\subset {\mathbb R}^n \rightarrow {\mathbb R} $ is pseudo-concave if 
$$ x,y\in X:\quad D_x f(x)\cdot (x-y)\geq 0\Rightarrow f(x)\geq f(y).$$}
\end{definition}
A sufficient condition for this property is that for all $ a\in A$, $c_0\in C_0$, the objective function $ f$ and the constraints $ -g,-h $ of the original problem (\ref{math_pure_mh}) are concave in $ c_1$. As we will illustrate in a simple example in Section \ref{sec:pri-ag} this is not a necessary condition. In Appendix \ref{complex:pfoc} we  give a more general sufficient condition for a partially concave problem, and derive the worst case complexity of our method for this case. Intuitively, partial concavity is a natural feature of many economic problems since many economic problems are likely to be convex in some variables. More specific examples will be discussed in Section \ref{sec:pri-ag} and Section \ref{sec:Mirrleesian}.
Partial concavity cannot be effectively exploited by linear programming or grid‑search methods because they rely on a full discretization of all actions.
The applicability of different methods to the two structures is illustrated in Table \ref{tab:applicability}.

\begin{table}[htbp]
\centering
\begin{tabular}{lccc}
\toprule
 & \textbf{Lagrangian Iteration} & \textbf{LP} & \textbf{Deterministic Solution} \\
\midrule
\textbf{Decomposability}    & $\checkmark$ & $\checkmark$ & $\times$ \\
\textbf{Partial Concavity}   & $\checkmark$ & $\times$     & $\times$ \\
\bottomrule
\end{tabular}
\caption{Applicability of different methods to the two structures}
\label{tab:applicability}
\end{table}

\subsection{A classical example}
We consider a variation of the example in Section 6 of \cite{prescott1984general}, which introduces the production of a public good.  There are $n$ private commodities, one public good, and finitely many types of agents $\theta \in \Theta$. Each type $\theta$ has a continuous and increasing utility function $u^{\theta}:{\mathbb R}^{n+1}_+ \rightarrow {\mathbb R}$ over private and public consumption.

The utilitarian planner maximizes total welfare $\sum_{\theta\in \Theta} u^{\theta}(c^{\theta}, G)$,
and faces the following resource constraints:
\[
 \sum_{\theta \in \Theta}  c^{\theta} \le \omega - x, \mbox{ and }
F(x) -G \ge 0,
\]
where $ F:\mathbb{R}^{n+1}\rightarrow \mathbb{R} $ is assumed to be a differentiable and strictly concave  production function.
When an agent's type is unobservable,
the deterministic problem can be written as:
\begin{align}
&\max_{x, (c^{\theta})_{\theta \in \Theta},G} \quad \sum_{\theta\in \Theta} u^{\theta}(c^{\theta},G), \label{eq:objective_sect4} \\
\text{s.t.}\quad & \sum_{\theta \in \Theta}  c^{\theta} \le \omega-x \label{eq:resource1_sect4} \\
& F(x)-G \ge 0, \label{eq:resource_aggregate_sect4} \\
& u^{\theta}(c^{\theta},G) \ge u^{\theta}(c^{\theta'},G), \quad \forall \theta, \theta' \in \Theta. \label{eq:incentive_sect4}
\end{align}

We introduce Lagrangian multipliers: $\lambda$ for the resource constraints \eqref{eq:resource1_sect4}; $\mu$ for the aggregate resource constraint \eqref{eq:resource_aggregate_sect4};  and $(\gamma_{\theta,\theta'})_{\theta,\theta'\in \Theta}$ for the incentive constraints \eqref{eq:incentive_sect4}. The Lagrangian of this problem is:
\begin{align}
&\mathcal{L}\bigl(x,(c^{h})_{h\in H},G;\lambda,\gamma,\mu\bigr) \notag \\
&= 
 \sum_{\theta}\left[ \left(1 + \sum_{\theta'\in \Theta} \gamma_{\theta,\theta'}\right) u^{\theta}(c^{\theta},G) 
 - \lambda  c^{\theta} - \sum_{\theta'\in \Theta} \gamma_{\theta',\theta} u^{\theta'}(c^{\theta},G) \right] \label{eq:Lagrangian_part2} \\
&\quad -\lambda x + \mu F(x) -\mu G +\lambda \omega \label{eq:Lagrangian_part3} 
\end{align}

We observe that, given the multipliers \((\lambda,  \gamma, \mu)\), the Lagrangian is not directly decomposable. However, the Lagrangian is concave in \( x \); hence, the economy exhibits partial concavity.
For given multipliers, there is a unique $x$ that maximizes the Lagrangian, and its value can be directly determined from the equation $ \frac{\lambda}{\mu}=D_x F(x) $. This yields an upper bound for the one-dimensional variable $G$. For each type $ \theta $, the optimal $ c^{\theta} $ maximizes the following expression which is independent of all other types' consumption, given $G$ and the multipliers.
$$  f^{\theta}(c;G,\lambda,\gamma)=\left(1 + \sum_{\theta'\in \Theta} \gamma_{\theta,\theta'}\right) u^{\theta}(c^{\theta},G) 
 - \lambda  c^{\theta} - \sum_{\theta'\in \Theta} \gamma_{\theta',\theta} u^{\theta'}(c^{\theta},G) .$$ 
In Step 1, the algorithm then searches over $G$ in the admissible interval and for each $G$ maximizes $ f^{\theta}(c; G,\lambda,\gamma)$ for each type $ \theta $.
  Since linear programming cannot take advantage of partial concavity, this simplification only improves efficiency for the Lagrangian iteration algorithm. This algorithm can be used to compute the optimal solution to 
(\ref{eq:objective_sect4})--(\ref{eq:incentive_sect4}) even when the number of types becomes large.

\section{Application to a Textbook Principal-Agent Problem}
\label{sec:pri-ag}
In this section, we illustrate our method as well as its advantages over the conventional linear programming approach using a textbook principal-agent problem with moral hazard (e.g., in \citet[Section 5.2]{salanie2005economics}).
The agent can take unobserved actions $a \in A$, which influence the probability distribution of their output $q \in Q$ through conditional probabilities $p(q|a)$. The action set $A$ and the output set $Q$ are finite sets. The total amount of consumption goods available to the social planner is the sum of the outputs of all agents. An agent has utility $u(c, a)$ from action $a$ and consumption $c \in C$, where the consumption set $C$ is a closed interval. For a given output $q$ and a given consumption contract $c(q)$, the principal has utility $v(q-c(q))$. In the deterministic solution, the principal's goal is to allocate an action and a consumption contract $\{a,c(q)\}$ for different states of output to each agent to maximize her expected utility:
\begin{equation}\label{moral_hazard_pure_obj}
    \max_{a,c(q)}\sum_{q}p(q|a)v(q-c(q)),
\end{equation}
subject to the participation constraint and incentive compatibility constraint:
\begin{equation}\label{moral_hazard_pure_cons}
\begin{aligned}
    &\sum_{q}p(q|a)u(c(q),a)\ge \underline{U};\\
    &\sum_{q}p(q|a)u(c(q),a)\ge \sum_{q}p(q|\hat{a})u(c(q),\hat{a}),\,\forall \hat{a}\in A,
\end{aligned}
\end{equation}
where $\underline{U}$ is the value of the outside option.
 Compared to the general form \eqref{math_pure_mh}, the function $f(a, c)$ is defined as $f(a, c) = \sum_{q} v(q- c(q)) p(q | a)$. The participation constraint function is $g_1(a, c) = \sum_{q} u(c(q),a) p(q | a)-\underline{U}$. The incentive constraint functions are $h_{\hat{a}}(a, c) = \sum_{q} p(q | \hat{a}) u(c(q), \hat{a}) - \sum_{q} p(q | a) u(c(q),a),\,\forall \hat{a}\in A$.

\paragraph{Lottery Solutions.}
We now consider the lottery solution for the moral hazard problem above (see, e.g. \cite{arnott1988randomization}). Instead of choosing a deterministic action and outcome-contingent consumption contract $\{a, c(q)\}$, the planner chooses a probabilistic allocation over actions and  consumption contracts. Mathematically, the problem is formulated in the space $\mathcal{P}(A \times C^{|Q|})$. An element $x\in \mathcal{P}(A\times C^{|Q|})$ can be expressed as $x=x(a,c)=x(a,c(q_1),\cdots,c(q_{|Q|}))$. The objective function \eqref{moral_hazard_pure_obj} becomes
\begin{equation}\label{moral_hazard_mix_obj}
\sum_{a\in A}\int_{c\in C^{|Q|}}x(a,dc)\sum_{q\in Q}p(q|a)v(q-c(q)).
\end{equation}
The participation constraint in \eqref{moral_hazard_pure_cons} is assumed to hold only in expectation,
\begin{equation}\label{moral_hazard_mix_res}
\sum_{a\in A}\int_{c\in C^{|Q|}}x(a,dc)\sum_{q\in Q}p(q|a)u(c(q),a)\ge \underline{U},
\end{equation}
and the incentive constraints in \eqref{moral_hazard_pure_cons} become
\begin{equation}\label{moral_hazard_mix_icc}
\int_{c\in C^{|Q|}}x(a,dc)\sum_{q}p(q|a)u(c(q),a)\ge \int_{c\in C^{|Q|}}x(a,dc)\sum_{q}p(q|\hat{a})u(c(q),\hat{a}),\,\forall (a,\hat{a})\in A\times A.
\end{equation}

The interpretation is that the principal randomizes over actions and compensation schedules, commits to a possibly random compensation schedule to each agent, and announces a recommended action to the agent. The agent undertakes the recommended action if the expected payoff is weakly larger than for all other actions, taking the random compensation schedule as given. The agent commits to participating before  he observes the realization of the lottery over recommended actions. The incentive constraint must hold for each recommended action.

We note that the economy is decomposable as in Definition \ref{def:decomposability}; hence we can define $\pi \in \mathcal{P}(C \times Q \times A)$ such that $\pi(dc(q), q, a) = \int_{c_{-q}}x(a, dc) p(q|a)$, where $c_{-q}:=(c(q'))_{q'\in Q,\,q'\ne q}$. With a slight abuse of notation, we now let $c \in C$ (rather than $c \in C^{|Q|}$), so that the planner's objective function \eqref{moral_hazard_mix_obj} can be written as
\begin{equation}\label{moral_hazard_mixfinal_obj}
    \sum_{q,a}\int_{c}\pi(dc,q,a)v(q-c),
\end{equation}
while the participation constraint \eqref{moral_hazard_mix_res} and the incentive constraint \eqref{moral_hazard_mix_icc} become
\begin{equation}\label{moral_hazard_mixfinal_res}
    \sum_{q,a}\int_{c}\pi(dc,q,a)u(c,a)\ge \underline{U},
\end{equation}
\begin{equation}\label{moral_hazard_mixfinal_icc}
\sum_{q}\int_{c}\pi(dc,q,a)u(c,a)\ge \sum_{q}\int_{c}\pi(dc,q,a)\frac{p(q|\hat{a})}{p(q|a)}u(c,\hat{a}), \quad \forall (a, \, \hat{a}) \in A \times A.
\end{equation}
respectively.

To apply the Lagrangian iteration algorithm, note that the deterministic Lagrangian can be written as
\begin{align}
 {\mathcal L}(a,c;\lambda,\gamma)=&\sum_{q}p(q|a)v(q- c(q))-  \sum_{\hat a \in A} \gamma_{a,\hat{a}} \sum_{q}(p(q|\hat{a})u(c(q),\hat{a})-p(q|a)u(c(q),a))\\
 & -\lambda(\underline{U}-\sum_{q}p(q|a)u(c(q),a)). \end{align}

 It is easy to see that the problem is decomposable in outcomes $ q \in Q$ and that we can write
$$ {\mathcal L}(a,c;\lambda,\gamma)=\sum_{q \in Q}  {\mathcal L}_q(a,c(q);\lambda,\gamma) .$$
For the discretized version of the problem with finite $C$, decomposability implies that
for a given $ a,\lambda,\gamma $ one needs to consider $ |Q| |C| $ possible $c$ instead of $ |C|^{|Q|}$ points.
However, under mild assumptions on fundamentals, the problem becomes partially convex and one does not need to discretize the set $C$ at all.
The following lemma provides sufficient conditions for partial concavity in this setting.
\begin{lemma}
\label{lem:pcmh}
    Suppose $ v(.) $ is $C^2$, concave, and increasing; for each $ a \in A$, $ u(.,a) $ is $ C^2 $ concave and increasing. Also assume that
     for each $ a\in A $,  $ \frac{\partial u(c,a)}{\partial c} = C(a) \tilde{u}'(c)$ for a decreasing function $ \tilde{u}'(.)$. Then the economy is partially concave, i.e., the Lagrangian ${\mathcal L}(a,c;\lambda,\gamma)$ is continuously differentiable and pseudo-concave in $c$. 
\end{lemma} 
\begin{proof}
    See Appendix \ref{app:pcmh}.
\end{proof}

\subsection{A numerical example}
We adopt the parameterization following Example 1 in \cite{prescott1998computing}.\footnote{The formulation of the principal agent problem in \cite{prescott1998computing} is slightly non-standard in that he maximizes the agent's expected utility subject to the resource constraint.
However, it is easy to see that if we set the outside option $ \underline{U} $ to the maximizing value in \cite{prescott1998computing}, his problem is mathematically equivalent to our framework.} The consumption set is \(C = [0, 2]\), the output set is \(Q = \{0.5, 1.5\}\), and the action set is \(A = 0.05:\Delta a:1.95\), where \(\Delta a=0.025\). The agents' utility function is defined as \(u(a,c) = \sqrt{c} + 0.8\sqrt{2 - a}\). The principal is risk neutral, and her utility is $v(q-c)=q-c$. The relationship between output and action is given by:

\[
p(q = 1.5 \mid a) = \begin{cases}
\frac{1 - (1 - a)^{0.2}}{2}, & \text{if } a < 1, \\
\frac{1 + (a - 1)^{0.2}}{2}, & \text{if } a \ge 1.
\end{cases}
\]

Clearly, the assumptions in Lemma \ref{lem:pcmh} are satisfied, and the problem is partially concave.
Therefore, in each Step 1 of our algorithm, the optimal consumption \(c^k(a,q,\lambda^k,\gamma^k)\) is determined by the first-order condition (FOC):
$$
c^k(a,q,\lambda^k,\gamma^k) = \min \left\{ c_{max}, \left(\tilde{u}'\right)^{-1} \left(\frac{\mathcal{B}(a,q)}{\mathcal{A}(a,q,\lambda^k,\gamma^k)}\right) \right\},
$$
where
$$
\mathcal{A}(a,q,\lambda^k,\gamma^k)=\lambda^k p(q|a)-\sum_{\hat{a}\in A}\gamma_{\hat{a},a}^k\left(p(q|\hat{a})-p(q|a)\right),\,\mathcal{B}(a,q)= p(q|a),$$
when $\mathcal{A}(a,q,\lambda^k,\gamma^k)>0$; 
and the optimal consumption is determined by
$$
c^k(a,q,\lambda^k,\gamma^k) = c_{min},
$$
when $\mathcal{A}(a,q,\lambda^k,\gamma^k)\le 0$.
Thus, we can analytically determine the optimal consumption for each iteration using the FOC.

To solve this problem with linear programming, \cite{prescott1998computing} discretizes the set \(C\)  to \(\hat{C} = 0:0.01:2\), and obtains the following solution:
\[
\pi(a = 0.050) = 0.0924, \, \pi(a = 1.075) = 0.9076
\]
\[
\pi(c = 1.20 \mid q = 0.5, a = 0.05) = 1, \, \pi(c = 1.19 \mid q = 1.5, a = 0.05) = 1.
\]
\[
\pi(c = 0.54 \mid q = 0.5, a = 1.075) = 0.5311, \, \pi(c = 0.55 \mid q = 0.5, a = 1.075) = 0.4689.
\]
\[
\pi(c = 1.40 \mid q = 1.5, a = 1.075) = 1.
\]
In this solution, the expected utility of the agent is approximately 1.8950. 
We use our algorithm with $\underline{U}=1.8950$ and continuous $c$. We choose the initial Lagrangian multipliers as \(\lambda^1 = 0.5\) and \(\gamma^1 = 0\). The learning rate is chosen as 
$$
\mu^k=\frac{1}{\left(k+\frac{1}{|\Delta a|^2}\right)^{0.7}}.$$

 The total number of iterations, \( N \), is chosen from the relation \( N \propto \ell\sim \frac{1}{\Delta a} \), as derived from the condition \eqref{ncon:sub_gradient_conv} and we take $100/\Delta a=4000$ iterations.
The trajectories of \(\lambda^k,\,a^k,\,c^k,\,\text{and }V(\lambda^k,\gamma^k)\) are shown in Figure \ref{Fig_mhfoc}.
\begin{figure}[ht!]
\centering
    \includegraphics[width=0.95\linewidth]{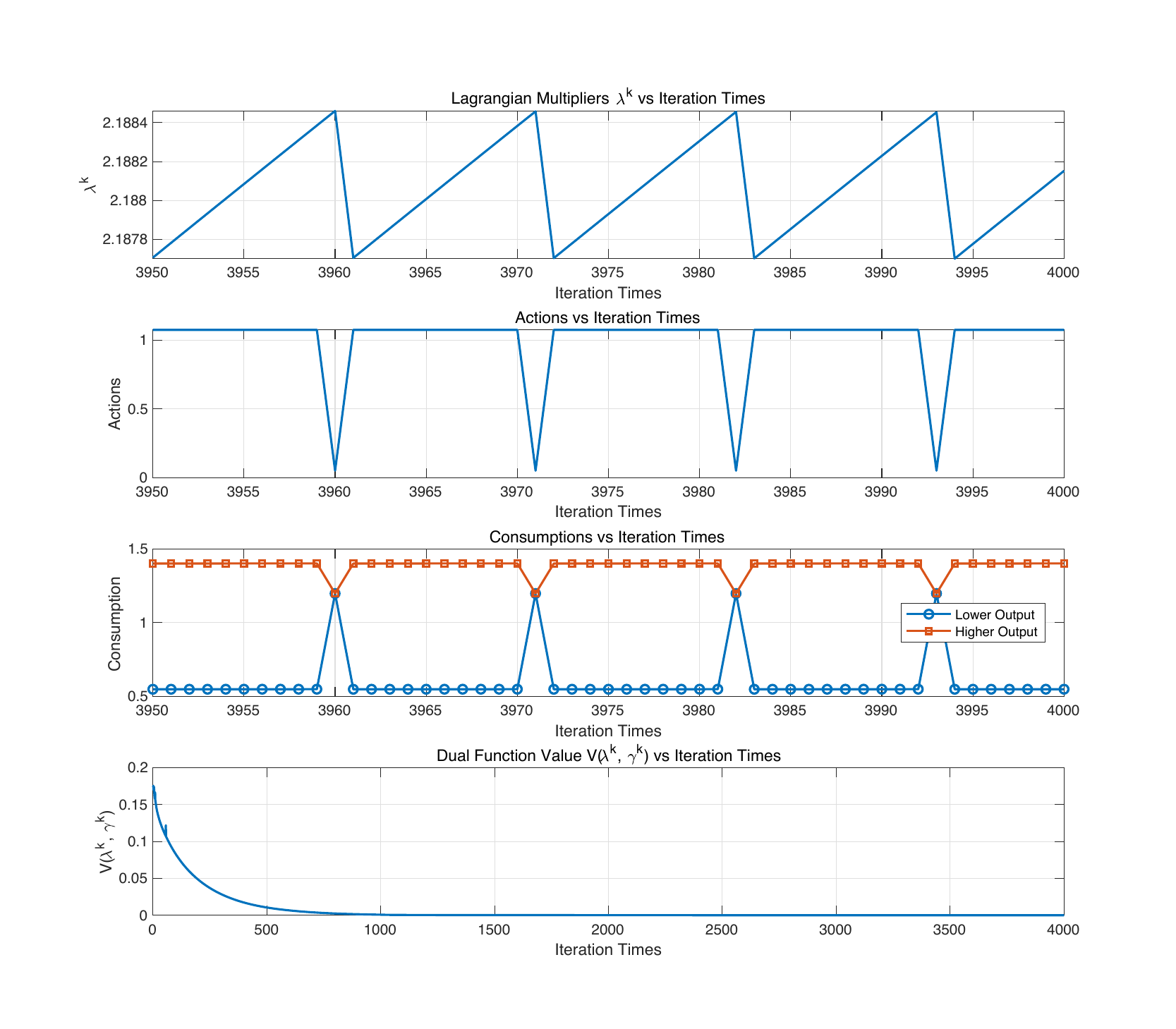}
\caption{Results from Lagrangian iteration when $\Delta a=0.025$. The four panels plot the Lagrangian multipliers $\lambda^k$, optimal actions $a^k$, optimal consumptions $c^k$, and the values of the dual function $V(\lambda^k,\gamma^k)$ over iteration.}\label{Fig_mhfoc}
\end{figure}
The action \(a^k\) oscillates between 0.050 and 1.075 during iterations. Notably, there exists a relationship between \( \lambda^k \) and \( a^k \): when \( \lambda^k \) is below a certain threshold, indicating that the principal places smaller emphasis on satisfying the participation constraint, the resulting action \( a^k \) stabilizes around 1.075. Conversely, when \( \lambda^k \) exceeds this threshold, reflecting an increased focus on the participation constraint, \( a^k \) stabilizes around 0.050. We use the results from rounds 3750 to 4000 to construct a lottery \(\hat{\pi}\). The resulting lottery is:
\[
\hat{\pi}(a = 0.050) = \frac{\sum_{k=3750}^{4000}\mu^k \mathbf{1}_{a^k=0.050}}{\sum_{k=3750}^{4000}\mu^k} \approx 0.0917.
\]
This value is close to the solution obtained from linear programming, where \(\pi(a = 0.050) = 0.0924\).  When \(a^k = 0.050\), the consumption contract approximates \(c^k(q=0.5) = c^k(q=1.5) = 1.20\); and when \(a^k = 1.075\), the consumption contract approximates \(c^k(q=0.5) = 0.55\) and \(c^k(q=1.5) = 1.40\). 

\subsection{Efficiency and comparison with linear programming}

In this textbook example, the Lagrangian iteration method exhibits highly efficient computational performance not only by exploiting the low dimensionality but also the decomposability and partial concavity of the problem.\footnote{As \cite{prescott1998computing} and \cite{su2007computation} point out, the decomposability in $c(q) $ also simplifies the linear programming approach and the number of variables can be taken to be roughly $ |Q| |C| |A|$ (instead of $  |C|^{|Q|} |A|$ in the general problem). 
However, the linear programming approach cannot exploit partial concavity, and the size of the problem still grows dramatically in $ |C|$, $|Q|$, and $|A|$. }
In our numerical example, our method takes only 0.05 seconds to run 4000 Lagrangian iterations to obtain the approximate lottery solution, while the linear programming method takes 1.69 seconds\footnote{This CPU time accounts solely for the LP solver runtime, excluding the matrix construction time.} using the optimization toolbox in MATLAB. As shown in the last sub-figure of Figure \ref{Fig_mhfoc}, the dual function \(V(\lambda^k, \gamma^k)\) (defined in \eqref{def:dualfunc}) rapidly converges to the minimum \(V(\lambda^*, \gamma^*)\approx 0\) within the first few hundred iterations.  
The subsequent iterations serve primarily to build lotteries and do not require large learning rates.  A careful design of learning rates can therefore improve computational efficiency in practical applications.

It is clear that a finer discretization of $C$ would dramatically increase the running time of linear programming, but this does not affect our method since we do not need to discretize $C$. Slightly more surprising is the fact that using a larger set of possible actions slows down the linear programming approach dramatically while having only modest effects on the run-time of our algorithm.
To demonstrate this, we vary \(\Delta a\) and perform $100/\Delta a$ iterations for each specification.
The computational time for each choice of \(\Delta a\)
is presented in Table \ref{Table_lag}. In addition, the table provides the sizes of the corresponding linear programming problems, highlighting the direct implications on memory requirements.

\begin{table}[ht!]
\centering
\begin{tabular}{cccccc}
\hline\hline
\(\Delta a\)&Iterations & CPU time & \multicolumn{3}{c}{Size of LP} \\
\cline{4-6}
& & & \#Vars & \#Equality Cons & \#Inequality Cons \\
\hline
0.2    & 500& 0.006 & 4020   & 21  & 91   \\
0.1    & 1000& 0.01 & 8040   & 41  & 381  \\
0.05   & 2000& 0.02 & 15678  & 79  & 1483 \\
0.025  & 4000& 0.05 & 30954  & 155 & 5853 \\
0.0125  & 8000& 0.16 & 61506  & 307 & 23257 \\
0.00625  & 16000& 0.98 & 122610 & 611 & 92721 \\
\hline\hline
\end{tabular}
\caption{Computational Performance for Different \(\Delta a\). ``\# Vars'' refers to the number of variables in the LP method, ``\# Equality Cons'' refers to the number of equality constraints, and ``\# Inequality Cons'' refers to the number of inequality constraints, excluding the non-negativity constraints.}
\label{Table_lag}
\end{table}

\begin{figure}[ht!]
\centering
   \includegraphics[width=0.9\linewidth]{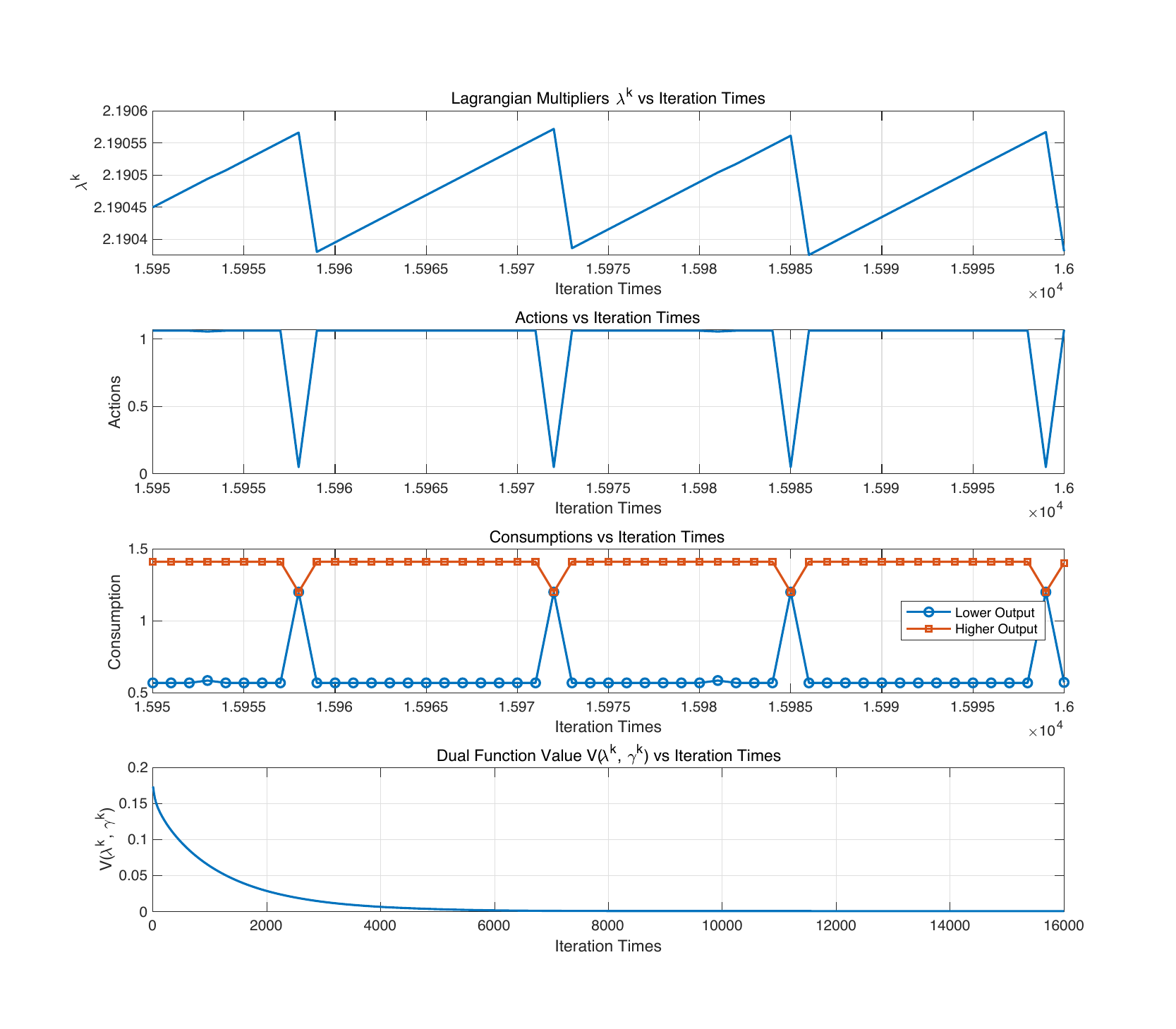}
\caption{Results from Lagrangian iteration when $\Delta a=0.00625$. The four panels plot the Lagrangian multipliers $\lambda^k$, optimal actions $a^k$, optimal consumptions $c^k$, and the values of the dual function $V(\lambda^k,\gamma^k)$ over iteration. 
}\label{Fig_mhfoc_largea}
\end{figure}

Although the linear programming solver can also handle the problem with $\Delta a=0.0125$ within 12.93 seconds, it cannot solve the finer grid case of $\Delta a=0.00625$ on a regular laptop due to exceeding memory limits. In contrast, our Lagrangian iteration method remains computationally efficient and requires significantly less memory. Furthermore, Figure \ref{Fig_mhfoc_largea} illustrates the trajectories of \( \lambda^k \), \( a^k \), \( c^k \), and \( V(\lambda^k, \gamma^k) \) for \( \Delta a = 0.00625 \). In this case, \( a^k \) oscillates between 1.0625 and 0.05, which differs from the case when \( \Delta a = 0.025 \), where \( a^k \) oscillates between 1.075 and 0.05. The discretization of the action set has a significant effect on the optimal solution.

\section{Mirrleesian optimal taxation with multidimensional types}\label{sec:Mirrleesian}
Following \cite{mirrlees1971exploration}, there is a large literature on optimal income taxation with information frictions. An important open question is how the classical results change when private information is multi-dimensional (see, e.g. \cite{judd2017optimal}, \cite{bergstrom2021optimal}). 
In such environments, the classical first-order approach, which uses single-crossing and monotonicity to reduce global incentive constraints to local ones,  generally breaks down, and computing the optimal deterministic mechanism is challenging
(see \cite{judd2017optimal} and \cite{boerma2022bunching}).
In this section, we show that our Lagrangian iteration approach makes the multidimensional problem computationally tractable and yields new economic predictions: heterogeneity in productivity and labor supply elasticity can generate nondegenerate optimal randomization, reduce bunching, and smooth consumption toward the full-information benchmark.

Following \cite{weiss1976desirability}, \cite{brito1995randomization}, and \cite{hellwig2007undesirability}, we study an economy in which the planner can offer random tax schedules. Agents must commit to a schedule before observing the random outcome and before choosing labor supply.
There are finitely many types of agents $ \theta \in \Theta $, a finite set, and a continuum of identical agents within each type.
A type is characterized by a preference parameter $ \eta $ and a productivity parameter $ \omega $. We use $ \eta_{\theta}$ and $ \omega_{\theta} $ to denote the parameters of type $ \theta $.
Each agent of type $\theta=(\eta,\omega) $ has preferences over consumption ($c$) and effort ($a$), which we represent by a continuous utility function $ \bar{u}_{\eta}(c,a) $. 
Given effort $a$, an agent of type $ \theta=(\eta,\omega) $ receives a wage $\omega$ and has pre-tax income $ y=\omega a$. While income is observable, the agent's effort and type parameters $(\omega, \eta)$ are not.
We assume that $ a \in [0,a_{\max}]$. Unlike above, here $a$ is a continuous variable. 
It is useful to write the indirect utility function over consumption and income as
\begin{equation} \label{indirectutil} 
u_\theta(c,y)=\bar{u}_{\eta_\theta}(c,\frac{y}{\omega_\theta}).
\end{equation}
A deterministic allocation is therefore  $ (c_\theta,y_\theta)_{\theta\in \Theta} \in C\times Y $, where $ C \subset {\mathbb R}^{|\Theta|}$ is compact and $ Y=\times_{\theta\in\Theta}[0, \omega_{\theta}a_{\max}] \subset {\mathbb R}^{|\Theta|}$ is a $|\Theta|$-dimensional box. 

The planner chooses a lottery $x$ over  $ (c_\theta,y_\theta)_{\theta\in \Theta} $ to solve 
\begin{align}\
&\max_{x \in \mathcal {P}(C\times Y)} \int_{(c, y) \in C \times Y} \left( \sum_{\theta\in \Theta}  u_\theta(c_\theta,y_\theta) \right) x(dc, dy),\label{eq:tax_obj2}\\
\text{s.t.} & \int_{(c, y) \in C\times Y} \left(u_\theta(c_\theta,y_\theta)-u_\theta(c_{\theta'},y_{\theta'}) \right) x(dc, dy)\ge 0 \mbox{ for all } \theta,\theta'\in \Theta.\label{eq:tax_ic2}\\
 & \int_{(c, y) \in C\times Y} \left( \sum_{\theta \in \Theta}  c_\theta -  \sum_{\theta \in \Theta} y_\theta\right) x(dc,dy) \le 0. \label{eq:tax_rc} 
\end{align}

The assumption of expected utility implies that the incentive constraint (\ref{eq:tax_ic2}) only has to hold in expectation. Because there is a continuum of agents within each type, the resource constraints for consumption and for labor supply only have to hold in expectation, and they are combined into a single constraint (\ref{eq:tax_rc}). 
This problem is equivalent to the optimal taxation problem where the planner offers a possibly stochastic income tax schedule $T_{\theta}(y)$ for each reported type $ \theta $. Each agent chooses between a deterministic tax schedule and a lottery that depends on his reported type. If the agent chooses the lottery, after reporting her type $\hat{\theta}$, the planner uses a public randomization device to generate, for each agent reporting $\hat{\theta}$, an i.i.d. public random variable $s$ with a known distribution. This then implies a tax schedule $ T_{s,\hat \theta}(.)$. For simplicity and without loss of generality, we think of the tax schedule as being characterized by an income level $ \bar y(s,\hat\theta) $ and being of the form
$ T_{s\hat \theta}(y) = y $ if $ y \ne  \bar y(s,\hat\theta) $ and $ T_{s\hat \theta}(y)=\bar T_{s \hat \theta} $ if $ y = \bar y (s,\hat \theta)$  (see, e.g., \cite{brito1995randomization} or \cite{hellwig2007undesirability}).
By reporting type $ \hat \theta $, an agent commits to earning income $\bar y (s,\hat \theta)$ and obtaining consumption $ c(s,\hat \theta)=\bar y (s,\hat \theta) - \bar T_{s,\hat \theta } $ when the realization of her random signal is $s$. We assume the law of large numbers applies, so that among agents reporting type $ \hat \theta$, the empirical distribution of realized tax schedules is equal to the distribution of the signals.

It is easy to verify that this problem fits directly into our general setup. Moreover, our algorithm can be applied even when the number of types becomes large. The reason for this is that the problem is decomposable as described in Section \ref{sec:comper}, and the Lagrangian to the deterministic problem can be written as
\begin{eqnarray*} {\cal L}(c,y,\lambda,\gamma)&=&
\sum_{\theta \in \Theta}  u_\theta(c_\theta,y_\theta) +\sum_{\theta,\theta'\in \Theta}\lambda_{\theta,\theta'}
(u_\theta(c_\theta,y_\theta)-u_\theta(c_{\theta'},y_{\theta'}))-\gamma(
\sum_{\theta \in \Theta}  c_\theta - \sum_{\theta \in \Theta}  y_\theta)\\
&=& \sum_{\theta \in \Theta} {\cal L}_\theta(c_\theta,y_\theta,\lambda,\gamma),
\end{eqnarray*}
where
$ \lambda=(\lambda_{\theta,\theta'})_{\theta, \theta' \in \Theta}$ and
$$ {\cal L}_\theta(c_\theta,y_\theta,\lambda,\gamma)=
(1+\sum_{\theta'\in \Theta, \theta'\ne \theta}\lambda_{\theta,\theta'})u_\theta(c_\theta,y_\theta) - \sum_{\theta'\in \Theta, \theta' \ne \theta} \lambda_{\theta',\theta} u_{\theta'}(c_{\theta},y_{\theta})-\gamma (c_\theta-y_\theta).
$$
In Step 1 of Algorithm \ref{math_alg_mh}, one therefore has to solve $|\Theta|$ maximization problems in two variables. The computational costs in Step 1 increase linearly in the number of types. The number of constraints grows polynomially, and therefore Proposition \ref{thm_complexity} implies that the number of iterations required also grows polynomially.

Note that decomposability can also be exploited when using linear programming to solve the problem. The probability that one needs to assign to $ (c,y) $ can be decomposed into independent probabilities for each $ (c_\theta,y_\theta)$. 
Suppose that for each $ \theta $, the set of feasible $ (c_\theta,y_\theta)  $ is discretized into $M$ points.
Instead of solving a linear program with $ M^{|\Theta|} $ variables that would quickly become infeasible, one can solve a linear program with
$ M |\Theta| $ variables that remains feasible for moderate $M$.
In contrast, a simple exhaustive search for the optimal deterministic solution is clearly infeasible, since one needs to search in $M^{|\Theta|}$ dimensional space.

We now assume that for each type $\theta=(\eta,\omega) \in \Theta $, the utility function is given by
\begin{equation}\label{eq:utility_func}
  u_{(\eta,\omega)}(c,y) =\log(c)-\frac{\left(\frac{y}{\omega}\right)^{\frac{1}{\eta}+1}}{\frac{1}{\eta}+1}, 
\end{equation}

Heterogeneity in the Frisch elasticity $ \eta $ naturally creates scope for lotteries in this setting. In calibrations where it is optimal for lower-$\eta$ agents to supply more labor, the planner can use lotteries to prevent them from misreporting their type.
To understand how lotteries can improve  upon the deterministic solution for this utility function, it is useful to consider a relaxed problem that imposes only the following subset of incentive constraints:
\begin{equation}\label{eq:tax_ic_new_eta}
   \int_{(c, y) \in C\times Y} \left(u_\theta(c_\theta,y_\theta)-u_\theta(c_{\theta'},y_{\theta'})\right) x(dc, dy)\ge 0 \mbox{ for all } \theta, \theta' \mbox{ with } \eta_\theta \ge \eta_{\theta'}.
\end{equation}

The key insight that we formalize in Proposition \ref{prop:etaconv} below is that, if $ a_{\max} $ is sufficiently large, none of the other incentive constraints beyond \eqref{eq:tax_ic_new_eta} will be binding at the optimum. We can therefore focus on the problem subject only to the subset of constraints in \eqref{eq:tax_ic_new_eta}. 
Intuitively, to deter a low-$\eta$ agent from mimicking a higher-$\eta$ type, the planner can offer the higher-$\eta$ type a contract that, with very small probability, requires extremely high earnings (equivalently, very high labor effort). As $\eta$ increases, this tail risk becomes more extreme: the required earnings rise while the probability falls. In this contract, a low-$\eta$ agent does not want to imitate because the higher curvature of disutility from labor supply makes these rare high-effort realizations disproportionately costly. At the same time, the planner can choose the tail probabilities sufficiently small so that the associated welfare loss is second order.
The following proposition formalizes this. Section \ref{ssec:judd-example} provides a concrete example of this reasoning under a specific calibration. 

\begin{prop}\label{prop:etaconv}
The maximum objective value for the planner problem \eqref{eq:tax_obj2} subject to \eqref{eq:tax_rc} and \eqref{eq:tax_ic2}, and the maximum objective value for the planner problem \eqref{eq:tax_obj2} subject to \eqref{eq:tax_rc} and \eqref{eq:tax_ic_new_eta} have the same limit as $a_{\max}\rightarrow\infty$.
\end{prop}
\begin{proof}
    See Appendix \ref{app:etaconv}.
\end{proof}

\subsection{Lottery solution in a calibrated example}
\label{ssec:judd-example}
Following Section 5.1 in \cite{judd2017optimal}, we choose five values of $\omega_\theta \in \{ 1,2,3,4,5 \}$ and five values of $\eta_\theta\in \{ \frac{1}{8}, \frac{1}{5}, \frac{1}{3}, \frac{1}{2}, 1 \}$, resulting in 25 distinct types indexed by $ (\omega_\theta,\eta_\theta) $ and 600 incentive compatibility constraints.\footnote{To the best of our knowledge, this is the only paper that solved the deterministic Mirrleesian problem for multi-dimensional types without additional simplifying assumptions. We use their deterministic optimal solutions as a comparison benchmark to show the possible welfare gains from lotteries.} 
It is useful to first consider the ``first best'' solution where the planner faces no incentive constraints. In this case, consumption is constant across all types and income increases with productivity $ \omega $. For effort-levels above one, it is optimal for the planner to require less effort for lower $ \eta $-types while for effort-levels below one, the planner requires more effort for lower $ \eta $ types.
Columns 2 and 3 of Table \ref{Table_Judd} report the optimal deterministic solution and the first best solution. It is interesting to note that in the deterministic solution, there is substantial dispersion in consumption and a lot of \lq\lq bunching\rq\rq where different types receive identical $(c,y)$ allocation.

\begin{table}[ht!]
\centering
{
\begin{tabular}{ccccc}\hline\hline
\makecell{Type\\ $(\omega,\eta)$} & \makecell{Deterministic solution\\ $(c, y)$}  & \makecell{First Best \\  $(c, y)$} & \makecell{$a_{\max}=\infty$\\ $(c, \mathbb{E}[y])$} & \makecell{$a_{\max}=1.2$\\ $(c, \mathbb{E}[y])$}\\ \hline
1, 1 & 1.68, 0.42 & 3.17, 0.32 & 1.93, 0.38 (L) & 1.70, 0.43 \\
1, $\frac{1}{2}$ & 1.77, 0.62 & 3.17, 0.56  & 2.00, 0.63 (L) & 1.79, 0.62 \\
1, $\frac{1}{3}$ & 1.79, 0.65 & 3.17, 0.68  & 2.03, 0.72 (L)& 1.80, 0.65 \\
1, $\frac{1}{5}$ & 1.83, 0.77 & 3.17, 0.79  & 2.04, 0.82 (L) & 1.84, 0.77 \\
1, $\frac{1}{8}$ & 1.86, 0.86 & 3.17, 0.87 & 2.05, 0.88 & 1.88, 0.86 \\
\hline
2, 1 & 1.86, 0.86 & 3.17, 1.26 & 2.21, 1.10 (L) & 1.88, 0.86 (L) \\
2, $\frac{1}{2}$ & 2.03, 1.39 & 3.17, 1.59  & 2.28, 1.48 (L) & 2.05, 1.39 (L) \\
2, $\frac{1}{3}$ & 2.07, 1.50 & 3.17, 1.72  & 2.29, 1.67 (L) & 2.08, 1.49 \\
2, $\frac{1}{5}$ & 2.16, 1.74 & 3.17, 1.82  & 2.30, 1.87 (L) & 2.18, 1.75 \\
2, $\frac{1}{8}$ & 2.2, 1.83 & 3.17, 1.89  & 2.30, 1.92  & 2.22, 1.84 \\
\hline
3, 1 & 2.2, 1.83 & 3.17, 2.84 & 2.77, 2.29 (L) & 2.46, 2.50 (L) \\
3, $\frac{1}{2}$ & 2.47, 2.49 & 3.17, 2.92  & 2.74, 2.62 (L) & 2.61, 2.66 (L) \\
3, $\frac{1}{3}$ & 2.47, 2.49 & 3.17, 2.95  & 2.68, 2.77 (L) & 2.62,  2.69 (L) \\
3, $\frac{1}{5}$ & 2.55, 2.68 & 3.17, 2.97  & 2.61, 2.91 (L) & 2.69, 2.85 \\
3, $\frac{1}{8}$ & 2.62, 2.85 & 3.17, 2.98 & 2.56, 2.99  & 2.71, 2.88 \\
\hline
4, 1 & 3.36, 4.00 & 3.17, 5.06 & 3.62, 3.72 (L) &  3.18, 3.52 (L) \\
4, $\frac{1}{2}$ & 3.36, 4.00 & 3.17, 4.5  & 3.42, 3.92 (L) & 3.36, 4.02 (L) \\
4, $\frac{1}{3}$ & 3.36, 4.00 & 3.17, 4.32   & 3.26, 4.01 (L) & 3.36, 4.03 \\
4, $\frac{1}{5}$ & 3.36, 4.00 & 3.17, 4.19 & 3.06, 4.06 (L) & 3.36, 4.03 \\
4, $\frac{1}{8}$ & 3.36, 4.00 & 3.17, 4.12 & 2.91, 4.09 & 3.36, 4.03 \\
\hline
5, 1 & 4.87, 5.87 & 3.17, 7.9 & 4.76, 5.25 (L) & 4.81, 5.85 \\
5, $\frac{1}{2}$ & 4.49, 5.57 & 3.17, 6.28 & 4.38, 5.34 (L) & 4.47, 5.57 \\
5, $\frac{1}{3}$ & 4.34, 5.44 & 3.17, 5.82 & 4.08, 5.35 (L) & 4.32, 5.44 \\
5, $\frac{1}{5}$ & 4.11, 5.25 & 3.17, 5.48 & 3.71, 5.31 (L) & 4.10, 5.25 \\
5, $\frac{1}{8}$ & 4.00, 5.14 & 3.17, 5.29 & 3.39, 5.25  & 3.98, 5.14 \\
\hline\hline
\end{tabular}
}
\caption{Deterministic and lottery allocations across types. Columns 2--3 report the deterministic optimum and the first-best allocation, both as $(c,y)$. Columns 4--5 report the optimal lottery allocation for different values of $a_{\max}$; entries are $(c,\mathbb{E}[y])$, where $\mathbb{E}[y]$ is expected pre-tax income under the lottery. (L) indicates that the allocation for that type is randomized (i.e., the lottery has nondegenerate support). Deterministic solutions are from \cite{judd2017optimal}.}
\label{Table_Judd}
\end{table}

To understand how lotteries can improve upon this deterministic solution, we compute the optimal value for the planner problem \eqref{eq:tax_obj2} subject to \eqref{eq:tax_rc} and \eqref{eq:tax_ic_new_eta}. Using Proposition \ref{prop:etaconv}, we establish the following properties of the problem.
\begin{remark}\label{prop:etatwomax}
There exists $ \bar a $, such that for all $a_{\max} \ge \bar a$, the optimal lottery solution requires lotteries for all agents $ (\omega,\eta) $ with $ \eta \ne \frac{1}{8} $. In these solutions, lotteries always require these agents to earn income $ a_{\max} \omega$ with some positive probability $\pi_{\theta}$.
\end{remark}
\begin{proof}
    See Appendix \ref{app:etatwomax}.
\end{proof}

The fourth column of Table \ref{Table_Judd} reports the lottery allocation as $a_{\max}\to\infty$. 
Three observations stand out. First, relative to the deterministic benchmark, bunching disappears: distinct types are no longer assigned identical allocations. Second, consumption is substantially smoother across types, moving toward the full-information benchmark in which consumption is constant. 
Third, the lottery primarily affects incentives through the income or labor supply margin. 
For types with productivity $ \omega \in \{ 1,2,3 \} $, the average required income is quite close to the first best solution. Lower $ \eta $ types have to work more. The planner uses lotteries to prevent low-$\eta$ types from mimicking high-$\eta$ types by introducing rare high-labor supply realizations that are especially costly for low-$\eta$ agents.
For high productivity types, the pattern differs. When $\omega=4$, income declines with $\eta$, and when $\omega=5$ it is nearly flat, whereas in the first best income increases with $\eta$.\footnote{This is due to the fact that high productivity agents work more than one unit and the objective function is maximized if high $ \eta $ (low curvature) agents work more. For $ a < 1 $ it is the other way around.} The high $ \eta $ high $ \omega $ types then have an incentive to mimic the high $ \eta $ low $ \omega $ types that have lower income (hence lower $a$). Although lotteries discourage low $ \eta $ types from doing so, the high $ \eta $ types' incentive constraint can only be satisfied if they are not required to earn too high an income.

\begin{figure}[ht!]
\centering
   \includegraphics[width=0.9\linewidth]{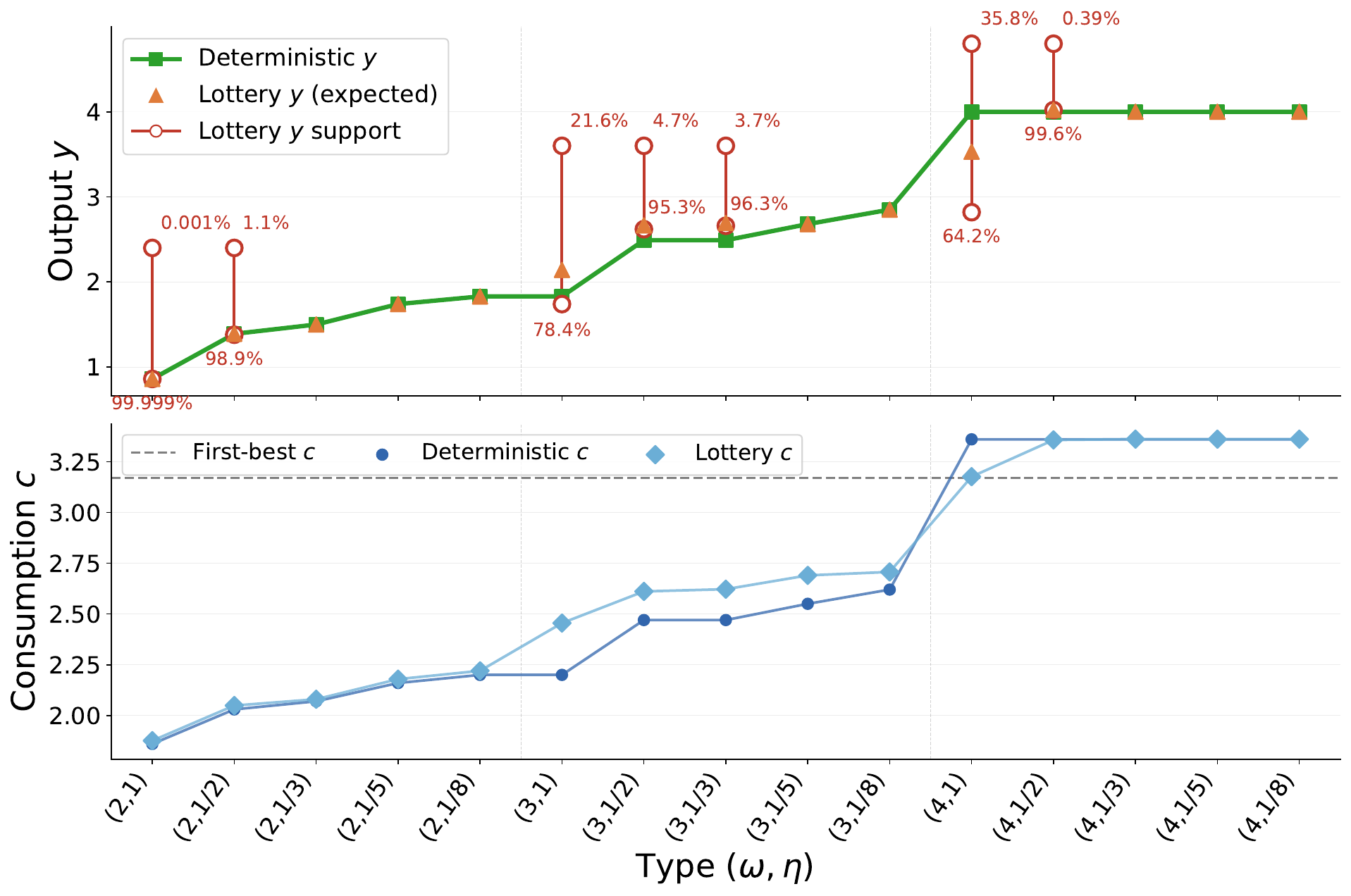}
\caption{Optimal consumption and output across types under various setups with $a_{\max} = 1.2$. Green squares and blue circles denote output and consumption in the optimal deterministic allocation. Orange triangles and light-blue diamonds denote expected output and expected consumption under the optimal lottery. Red open circles mark the support of the output lottery, with adjacent labels reporting probabilities. The gray dashed line denotes first-best consumption.}
\label{fig:omega_eta_lottery_sol_two_panel}
\end{figure}

With smaller values of $a_{\max}$, the planner’s ability to use rare large-income realizations is limited, and the structure of optimal lotteries becomes model-specific. 
We use our method to compute the optimal solutions for $a_{\max}=1.2$,\footnote{We set $a_{\max}=1.2$ to impose a relatively tight cap on labor supply while keeping the deterministic benchmark feasible. Under the deterministic solution, the highest-productivity types ($\omega=5$) earn income close to $5.9$, which corresponds to effort near the cap.} which is reported in the last column of Table \ref{Table_Judd} and visualized in Figure \ref{fig:omega_eta_lottery_sol_two_panel}.
Since the utility function over consumption is identical across all agents, each Lagrangian $ {\mathcal L}_{\theta}(.,y_\theta,\lambda,\gamma) $ is pseudo-concave in consumption, and consumption
can be determined via the first-order conditions. Step 1 of the Algorithm \ref{math_alg_mh} therefore reduces to $|\Theta|$ one-dimensional grid searches, and the example can be solved in about 3 minutes using $ 2 \times 10^6 $ Lagrangian iterations.\footnote{We take the optimal deterministic allocations from \cite{judd2017optimal} as a benchmark and compare them to our optimal lottery solutions, choosing the initial Lagrangian multipliers as $\gamma^1=0.5$ and $\lambda^1=0$. The learning rate is chosen as
$$
\mu^k=\begin{cases}
    \frac{1}{(k+1000)^{0.6}},&1\le k<10^6;\\
    \frac{1}{(k+1000)^{0.8}},&10^6\le k\le 2\times 10^{6},
\end{cases}
$$
and we use the last $10^5$ iteration rounds to construct the lottery solution.}

For $ a_{\max}=1.2$, no type with $ \omega=1,5$ faces a lottery. For $ \omega=3$ we still obtain three lotteries relaxing the incentive constraints across different $ \eta $. For $ \omega=4 $, similar to the deterministic solution we observe bunching. Many lotteries \lq\lq disappear\rq\rq \ as $ a_{\max} $ decreases. 
Figure \ref{fig:omega_eta_lottery_sol_two_panel} compares consumption and output lotteries for the optimal solution for $ a_{\max}=1.2 $ with the deterministic solution for $ \omega=2,3,4$ types. The figure shows how consumption becomes smoother across types (closer to the first best where consumption is constant across types) due to the fact that high $ \eta $ types now face lotteries. Overall only $ \eta=1 $ and $ \eta=1/2 $ types face significant lotteries (as opposed to the $ a_{\max}=\infty $  case).
 
\subsection{Welfare gain from randomization}
We now examine under which conditions lotteries can generate significant welfare gains.
To make welfare comparisons economically meaningful, we denote welfare differences in terms of the Hicksian ``compensating variation in resources'' relative to the first-best full-information problem. Formally, let \( \widetilde{W}(m) \) denote the highest level of social welfare that can be achieved in the full information setup when the economy has \( m \) fewer units of total resources, defined by:
\begin{align}
\widetilde{W} (m)=&\max_{x \in \mathcal {P}(C\times Y)} \int_{(c, y) \in C \times Y} \left( \sum_{\theta\in \Theta}  u_\theta(c_\theta,y_\theta) \right) x(dc, dy),\\
\text{s.t. }
&\int_{(c, y)\in C \times Y} \left( \sum_\theta  c_\theta - \sum_\theta  y_\theta\right) x(dc, dy) \le -m.
\end{align}
The function \( \widetilde{W}(m) \) is monotonic and decreases in \( m \). Let $ u^L $ denote the optimal welfare of the lottery problem, and let $ u^D $ denote the optimal welfare of the deterministic problem. For each of \( i \in \{L, D\} \), define \( m_i \) as the loss of resources in the full information problem that produces the same welfare level as \( u_i \); that is, \( \widetilde{W}(m_i) = u_i \). We then report welfare losses $\Delta_i$ as the relative compensating variation in resources compared to the total consumption in the first-best solution $C_{FB}$: 
\begin{equation}
    \Delta_i \equiv \frac{m_i}{C_{FB}}, \quad i=L, D.
\end{equation}

In our example, the deterministic mechanism yields a welfare loss of $\Delta_D = 7.16\%$ relative to the first best full-information benchmark.
As can be expected, potential welfare gains from lotteries depend crucially on $ a_{\max} $.
In the limit $a_{\max} \to \infty$, the welfare losses from asymmetric information shrink to $\Delta_L = 5.21\%$ under the optimal lottery mechanism. This implies that 27.23\% of the welfare loss from asymmetric information can be eliminated through lotteries and that the welfare gains from lotteries can be very large. These gains, however, are driven by lotteries that assign a very small probability to extremely high earnings requirements, which generate very large labor disutility in the tail and raise implementation concerns.

For $ a_{\max}=1.2 $, allowing for lotteries reduces this loss to $\Delta_L = 6.91\%$, implying a relative welfare gain of approximately 3.49\%. This rather small welfare gain is consistent with the pattern of output and consumption depicted in Figure \ref{fig:omega_eta_lottery_sol_two_panel}. Lotteries materially affect only a limited subset of types with $ \omega=3 $, so aggregate welfare improvements are correspondingly muted.

\section{Conclusion}
\label{sec:concl}
In this paper, we develop a new method to compute optimal randomization solutions in non-convex economies. Using a Lagrangian framework and iteratively decoupling optimization and complementary slackness conditions, our method provides a computationally efficient and memory-scalable approach to high-dimensional problems. Because non-convexities arise pervasively in economic theory, the method has broad potential applications. There is a large literature that examines problems with asymmetric information in a dynamic setting (see, e.g. \cite{fernandes2000recursive}, \cite{cole2001efficient}). These environments typically feature non-convex incentive constraints over continuation promises, making global computation particularly challenging. \citet{shen2025recursive} extends our approach to recursive contracts and proposes a general method for dynamic problems that directly addresses the long-standing critique that recursive contracts can mishandle non-convexities.

\singlespacing

\bibliographystyle{aer}
\bibliography{refe}

\onehalfspacing

\newpage
\appendix

\section*{Appendix}
\section{Proofs}
\label{app:proofs}

\subsection{Proof of Lemma \ref{lem:subgradient}}\label{app:lem:sub}
    Given $(\lambda,\gamma)\in\mathbb{R}_+^{m}\times\mathbb{R}_+^{\ell |A|}$. For any $(\lambda',\gamma')\in\mathbb{R}_+^{m}\times\mathbb{R}_+^{\ell |A|},$ we have
    \begin{equation}\label{sub_grad_prop}
    \begin{aligned}
        V(\lambda',\gamma')&=\max_{a\in A,\,c\in C}\mathcal{L}(a,c;\lambda',\gamma')\\
        &=\max_{a\in A,\,c\in C}f(a,c)-\sum_{i=1}^{m}\lambda'_i g_i(a,c)-\sum_{j=1}^{\ell}\gamma'_{j,a}h_j(a,c)\\
        &\ge f(a_{\lambda},c_{\lambda})-\sum_{i=1}^{m}\lambda'_i g_i(a_{\lambda},c_{\lambda})-\sum_{j=1}^{\ell}\gamma'_{j,a_{\lambda}}h_j(a_{\lambda},c_{\lambda})\\
        &=f(a_{\lambda},c_{\lambda})-\sum_{i=1}^{m}\lambda_ig_i(a_{\lambda},c_{\lambda})-\sum_{j=1}^{\ell}\gamma_{j,a_{\lambda}}h_j(a_{\lambda},c_{\lambda})+\left(\lambda'-\lambda,\gamma'-\gamma\right)\cdot\left(-(\Delta \lambda,\Delta \gamma)\right)\\
        &=V(\lambda,\gamma)+\left(\lambda'-\lambda,\gamma'-\gamma\right)\cdot\left(-(\Delta \lambda,\Delta \gamma)\right).
    \end{aligned}
    \end{equation}
    For any $(\lambda'_1,\gamma'_1),\,(\lambda_2',\gamma_2')\in \mathbb{R}_+^{m}\times\mathbb{R}_+^{\ell|A|}$ and $0\le \mu\le 1$, such that 
    $$
    \mu(\lambda'_1,\gamma'_1)+(1-\mu)(\lambda'_2,\gamma'_2)=(\lambda,\gamma),
    $$  \eqref{sub_grad_prop} implies that
    $$
    \begin{aligned}
    &\mu V(\lambda'_1,\gamma'_1)+(1-\mu)V(\lambda'_2,\gamma'_2)\\
    \ge &\mu \left(V(\lambda,\gamma)+\left((\lambda'_1,\gamma'_1)-(\lambda,\gamma)\right)\cdot\left(-(\Delta \lambda,\Delta \gamma)\right) \right)
+(1-\mu) \left(V(\lambda,\gamma)+\left((\lambda'_2,\gamma'_2)-(\lambda,\gamma)\right)\cdot\left(-(\Delta \lambda,\Delta \gamma)\right) \right)\\
    =&V(\lambda,\gamma),
    \end{aligned}
    $$
    implying that $V$ is convex. Furthermore, by the definition of a sub-gradient we know from \eqref{sub_grad_prop} that $(\Delta \lambda,\,\Delta \gamma)$ is a negative sub-gradient of $V$ at $(\lambda,\gamma)$.

\subsection{Proof of Theorem \ref{thm:relax}}\label{app:thm:relax}
For any $\epsilon>0$, we know that $x^{\epsilon}\in \mathcal{P}(A\times C)\subset\mathcal{M}(A\times C)$, and hence
$$
\sum_{a\in A}\int_{c\in C}|x^{\epsilon}(a,dc)|=\sum_{a\in A}\int_{c\in C}x^{\epsilon}(a,dc)=1,
$$
which is uniformly bounded with respect to $\epsilon$. According to Alaoglu's theorem, we can choose a sequence $\{\epsilon_n\}\rightarrow 0$, such that $x^{\epsilon_n}$ converges in the weak* topology to some $x^*\in \mathcal{M}(A\times C)$, i.e. for any $\varphi(a,c)\in C^{0}(A\times C)$, we have
\begin{equation}\label{eq:star_weak_conv}
\sum_{a\in A}\int_{c\in C}\varphi(a,c)x^{\epsilon_n}(a,dc)\rightarrow \sum_{a\in A}\int_{c\in C}\varphi(a,c)x^*(a,dc).
\end{equation}
Next, we prove that $x^*$ is the optimal solution to system \eqref{math_lot_mh}. 

We take $\varphi_1\equiv1$ in \eqref{eq:star_weak_conv}, and obtain
$$
\sum_{a\in A}\int_{c\in C}x^*(a,dc)=1,
$$
implying that $x^*\in \mathcal{P}(A\times C)$.

We take $\varphi_2\equiv g_i$ for some $i\in\{1,\cdots, m\}$ in \eqref{eq:star_weak_conv}, and obtain
$$
\begin{aligned}
&\sum_{a\in A}\int_{c\in C}g_i(a,c)x^*(a,dc)
=\lim_{\epsilon_{n}\rightarrow 0}\sum_{a\in A}\int_{c\in C}g_i(a,c)x^{\epsilon_n}(a,dc)
\le\lim_{\epsilon_{n}\rightarrow 0}\epsilon_n=0.
\end{aligned}
$$
Similarly, for $j\in\{1,\cdots, \ell\}, a'\in A$, we take
$$
\varphi(a,c)=\begin{cases}
    h_{j}(a',c),&a=a';\\
    0,&otherwise.
\end{cases}
$$
We can prove that
$$
\int_{c\in C}h_{j}(a',c)x^*(a',dc)\le 0.
$$
Therefore, $x^*$ satisfies all the constraints in \eqref{math_lot_mh} and is a feasible probability measure to \eqref{math_lot_mh}.

To show that $x^*$ is an optimal solution to system \eqref{math_lot_mh}, it then suffices to show that for any $x\in \mathcal{P}(A\times C)$ that is a feasible measure to \eqref{math_lot_mh}, we have
\begin{equation}\label{proof_relax_optcond}
\sum_{a\in A}\int_{c\in C}f(a,c)x^*(a,dc)\ge \sum_{a\in A}\int_{c\in C}f(a,c)x(a,dc).
\end{equation}
We prove \eqref{proof_relax_optcond} by contradiction. If \eqref{proof_relax_optcond} does not hold, then we can find $\tilde{x}\in\mathcal{P}(A\times C)$, such that $\tilde{x}$ satisfies all the constraints in \eqref{math_lot_mh} and 
\begin{equation}\label{proof_relax_mid}
\sum_{a\in A}\int_{c\in C}f(a,c)x^*(a,dc)< \sum_{a\in A}\int_{c\in C}f(a,c)\tilde{x}(a,dc).
\end{equation}
By the continuity of $f$ and the weak* convergence of $x^{\epsilon_n}$, we know that
$$
\sum_{a\in A}\int_{c\in C}f(a,c)x^*(a,dc)=\lim_{\epsilon_n\rightarrow 0}\sum_{a\in A}\int_{c\in C}f(a,c)x^{\epsilon_n}(a,dc),
$$
and hence we can find $\epsilon_{n_0}>0$, such that
$$
\sum_{a\in A}\int_{c\in C}f(a,c)x^{\epsilon_{n_0}}(a,dc)<\sum_{a\in A}\int_{c\in C}f(a,c)\tilde{x}(a,dc)
$$
according to \eqref{proof_relax_mid}. However, we know that $\tilde{x}$ is a feasible probability measure for \eqref{math_lot_mh} and hence a feasible probability measure for \eqref{math_lot_mh_relax} with $\epsilon=\epsilon_n$. This contradicts the fact that $x^{\epsilon_n}$ is the optimal solution to \eqref{math_lot_mh_relax} with $\epsilon=\epsilon_n$. Therefore, we finish the proof.

\subsection{Proof of Theorem \ref{thm:simp_L}}\label{app:thm:simp_L}

We first prove \eqref{eq:thm:sim_L_1}. By direct computation, we have
\begin{equation}
\begin{aligned}\label{eq:proof_thm_sim_L_1}
&L(x;\lambda,\gamma)\\
=&\sum_{a\in A}\int_{c\in C}f(a,c)x(a,dc)-\sum_{i=1}^{m}\lambda_i\sum_{a\in A}\int_{c\in C}g_i(a,c)x(a,dc)-\sum_{j=1}^{\ell}\sum_{a\in A}\gamma_{j,a}\int_{c\in C}h_{j}(a,c)x(a,dc)\\
=&\sum_{a\in A}\int_{c\in C}\left(f(a,c)-\sum_{i=1}^{m}\lambda_ig_i(a,c)-\sum_{j=1}^{\ell}\gamma_{j,a}h_j(a,c)\right)x(a,dc)\\
:=&\sum_{a\in A}\int_{c\in C}\mathcal{L}(a,c;\lambda,\gamma)x(a,dc)\\
\le& \max_{a\in A,c\in C}\mathcal{L}(a,c;\lambda,\gamma).
\end{aligned}
\end{equation}
Hence 
$$
\max_{x\in \mathcal{P}(A\times C)}L(x;\lambda,\gamma)\le \max_{a\in A,c\in C}\mathcal{L}(a,c;\lambda,\gamma).
$$
On the other hand, if $(a^*,c^*)\in\arg\max_{a\in A,c\in C}\mathcal{L}(a,c;\lambda,\gamma)$, then 
$$
L(\delta_{(a^*,c^*)};\lambda,\gamma)=\max_{a\in A,c\in C}\mathcal{L}(a,c;\lambda,\gamma),
$$
implying that
$$
\max_{x\in \mathcal{P}(A\times C)}L(x;\lambda,\gamma)\ge \max_{a\in A,c\in C}\mathcal{L}(a,c;\lambda,\gamma).
$$
Hence we can conclude that \eqref{eq:thm:sim_L_1} holds.

Now we prove \eqref{eq:thm:sim_L_2}. According to \eqref{eq:proof_thm_sim_L_1}, we know that
$$
L(x^*;\lambda,\gamma)=\sum_{a\in A}\int_{c\in C}\mathcal{L}(a,c;\lambda,\gamma)x^*(a,dc),
$$
and hence
$$
\begin{aligned}
& L(x^*;\lambda,\gamma)-\max_{x\in\mathcal{P}(A\times C)}L(x;\lambda,\gamma)\\
=&L(x^*;\lambda,\gamma)-\max_{a\in A,c\in C}\mathcal{L}(a,c;\lambda,\gamma)\\
=&\sum_{a\in A}\int_{c\in C}\mathcal{L}(a,c;\lambda,\gamma)x^*(a,dc)-\max_{a\in A,c\in C}\mathcal{L}(a,c;\lambda,\gamma)\\
=&\sum_{a\in A}\int_{c\in C}(\mathcal{L}(a,c;\lambda,\gamma)-\max_{a\in A,c\in C}\mathcal{L}(a,c;\lambda,\gamma))x^*(a,dc).
\end{aligned}
$$
Therefore, $x^*\in \arg\max_{x\in \mathcal{P}(A\times C)}L(x;\lambda,\gamma)$ if and only if $\mathcal{L}(a,c;\lambda,\gamma)=\max_{a\in A,c\in C}\mathcal{L}(a,c;\lambda,\gamma)$ a.s. with respect to probability $x^*$, which is equivalent to $(a,c)\in Z$ a.s. with respect to $x^*$.

\subsection{Proof of Theorem \ref{thm:dual=lot}}\label{app:dual=lot}

    We note that the equality \eqref{eq:thm:sim_L_1} in Theorem \ref{thm:simp_L} holds for any $(\lambda,\gamma)\in\mathbb{R}_+^{m}\times \mathbb{R}_+^{\ell |A|}$. Therefore, \eqref{eq:thm:sim_L_1} in Theorem  \ref{thm:simp_L} further yields that
     \begin{equation}\label{inf_sup_equiv}
    \inf_{(\lambda,\gamma)\in\mathbb{R}_+^{m}\times \mathbb{R}_+^{\ell |A|}}\max_{x\in \mathcal{P}(A\times C)}L(x;\lambda,\gamma)= \inf_{(\lambda,\gamma)\in\mathbb{R}_+^{m}\times \mathbb{R}_+^{\ell |A|}}\max_{a\in A,c\in C}\mathcal{L}(a,c;\lambda,\gamma),
    \end{equation}
    where $L$ is defined in \eqref{defL}. According to \eqref{eq:thm:sim_L_1} and \eqref{inf_sup_equiv}, $\mathcal{L}(a,c;\lambda,\gamma)$ in \eqref{dual} and \eqref{dual2} can be replaced by $L(x;\lambda,\gamma)$, i.e. to prove Theorem \ref{thm:dual=lot} it suffices to show that 
    \begin{equation}\label{dual_aim1}
        \max_{x\in\mathcal{P}(A\times C)}L(x;\lambda^*,\gamma^*)=\inf_{(\lambda,\gamma)\in\mathbb{R}_+^{m}\times \mathbb{R}_+^{\ell |A|}} \max_{x\in \mathcal{P}(A\times C)}L(x;\lambda,\gamma),
    \end{equation}
    and
    \begin{equation}\label{dual_aim2}
    \inf_{(\lambda,\gamma)\in\mathbb{R}_+^{m}\times \mathbb{R}_+^{\ell |A|}}\max_{x\in \mathcal{P}(A\times C)}L(x;\lambda,\gamma)=\sum_{a\in A}\int_{c\in C}f(a,c)x^*(a,dc).
    \end{equation}
    
    We first prove \eqref{dual_aim2}. Since $x^*$ satisfies all the constraints in \eqref{math_lot_mh}, for any $\lambda\ge 0,\,\gamma\ge0$, we know that
    $$
    \sum_{i=1}^{m}\lambda_i\sum_{a\in A}\int_{c\in C}g_i(a,c)x^*(a,dc)+\sum_{j=1}^{\ell}\sum_{a\in A}\gamma_{j,a}\int_{c\in C}h_{j}(a,c)x^*(a,dc)\le 0,
    $$
    and hence
    \begin{equation}\label{dual_mid1}
    \begin{aligned}
        &L(x^*;\lambda,\gamma)\\
        =&\sum_{a\in A}\int_{c\in C}f(a,c)x^*(a,dc)-\sum_{i=1}^{m}\lambda_i\sum_{a\in A}\int_{c\in C}g_i(a,c)x^*(a,dc)-\sum_{j=1}^{\ell}\sum_{a\in A}\gamma_{j,a}\int_{c\in C}h_{j}(a,c)x^*(a,dc) \\
        \ge &\sum_{a\in A}\int_{c\in C}f(a,c)x^*(a,dc).
    \end{aligned}
    \end{equation}
    Since $(\lambda,\,\gamma)$ is arbitrarily chosen in $\mathbb{R}_+^{m}\times \mathbb{R}^{\ell |A|}$, the inequality \eqref{dual_mid1} yields 
    $$
        \inf_{(\lambda,\gamma)\in\mathbb{R}_+^{m}\times \mathbb{R}_+^{\ell |A|}}L(x^*;\lambda,\gamma)\ge \sum_{a\in A}\int_{c\in C}f(a,c)x^*(a,dc),
    $$
    and further
    \begin{equation}\label{dual_ineq1}
         \max_{x\in\mathcal{P}(A\times C)}\inf_{(\lambda,\gamma)\in\mathbb{R}_+^{m}\times \mathbb{R}_+^{\ell |A|}}L(x;\lambda,\gamma)\ge  \inf_{(\lambda,\gamma)\in\mathbb{R}_+^{m}\times \mathbb{R}_+^{\ell |A|}}L(x^*;\lambda,\gamma)\ge \sum_{a\in A}\int_{c\in C}f(a,c)x^*(a,dc).
    \end{equation}
    Due to the general inf-sup inequality, \eqref{dual_ineq1} implies that
    \begin{equation}\label{dual_aim2_part1}
        \begin{aligned}
           &\inf_{(\lambda,\gamma)\in\mathbb{R}_+^{m}\times \mathbb{R}_+^{\ell |A|}}\max_{x\in\mathcal{P}(A\times C)}L(x;\lambda,\gamma)\\
           \text{ (Inf-sup Inequality)}\ge &  \max_{x\in\mathcal{P}(A\times C)}\inf_{(\lambda,\gamma)\in\mathbb{R}_+^{m}\times \mathbb{R}_+^{\ell |A|}}L(x;\lambda,\gamma)\\
           \text{ \eqref{dual_ineq1}}\ge& \sum_{a\in A}\int_{c\in C}f(a,c)x^*(a,dc).
        \end{aligned}
    \end{equation}
    On the other hand, since $(\lambda^*,\gamma^*)$ is the Lagrangian multipliers corresponding to $x^*$, we have
    \begin{equation}\label{dual_Lagmax}
        L(x^*;\lambda^*,\gamma^*)=\max_{x\in\mathcal{P}(A\times C)}L(x;\lambda^*,\gamma^*),
    \end{equation}
    implying that
     \begin{equation}\label{dual_aim2_part2_1}
        \inf_{(\lambda,\gamma)\in\mathbb{R}_+^{m}\times \mathbb{R}_+^{\ell |A|}}\max_{x\in\mathcal{P}(A\times C)}L(x;\lambda,\gamma)\\
    \le  \max_{x\in\mathcal{P}(A\times C)}L(x;\lambda^*,\gamma^*)
    =L(x^*;\lambda^*,\gamma^*).
    \end{equation}
    Due to the complementary conditions satisfied by the Lagrangian multipliers, we have
    \begin{equation}\label{dual_aim2_part2_2}
    \begin{aligned}
        &L(x^*;\lambda^*,\gamma^*)\\
        =&\sum_{a\in A}\int_{c\in C}f(a,c)x^*(a,dc)-\sum_{i=1}^{m}\lambda^*_i\sum_{a\in A}\int_{c\in C}g_i(a,c)x^*(a,dc)-\sum_{j=1}^{\ell}\sum_{a\in A}\gamma^*_{j,a}\int_{c\in C}h_{j}(a,c)x^*(a,dc) \\
        = &\sum_{a\in A}\int_{c\in C}f(a,c)x^*(a,dc).
    \end{aligned}
    \end{equation}
    We combine \eqref{dual_aim2_part2_1} and \eqref{dual_aim2_part2_2} to obtain
    \begin{equation}\label{dual_aim2_part2}
      \inf_{(\lambda,\gamma)\in\mathbb{R}_+^{m}\times \mathbb{R}_+^{\ell |A|}}\max_{x\in\mathcal{P}(A\times C)}L(x;\lambda,\gamma)\le   \sum_{a\in A}\int_{c\in C}f(a,c)x^*(a,dc).
    \end{equation}
    The equality \eqref{dual_aim2} is then deduced from \eqref{dual_aim2_part1} and \eqref{dual_aim2_part2}.
    Furthermore, the equality \eqref{dual_aim1} can be implied directly from \eqref{dual_aim2}, \eqref{dual_Lagmax}, and \eqref{dual_aim2_part2_2}. We can then combine \eqref{dual_aim1} and \eqref{dual_aim2} to finish the proof.

\subsection{Proof of Theorem \ref{thm:eps_optimal}}\label{app:eps_optimal}
Following the same steps in the proof for Theorem \ref{thm:relax}, we can choose a sequence $\{\epsilon_{n}\}_{n=1}^{\infty}\rightarrow 0$, such that $\tilde{x}^{\epsilon_n}$ converges in the weak* topology to some $\tilde{x}^*\in\mathcal{M}(A\times C)$, and $\tilde{x}^{*}$ satisfies all the constraints in the lottery system \eqref{math_lot_mh}. Suppose that $x^*$ is an optimal solution to \eqref{math_lot_mh}, it suffices to show that
\begin{equation}\label{proof:eps_optimal}
    \sum_{a\in A}\int_{c\in C}f(a,c)\tilde{x}^*(a,dc)\ge \sum_{a\in A}\int_{c\in C}f(a,c)x^*(a,dc).
\end{equation}
Indeed, by property 3 in the definition of an $\epsilon-$optimal solution, we know that
\begin{equation}\label{proof:eps_optimal_prop3}
\sum_{a\in A}\int_{c\in C}f(a,c)\tilde{x}^{\epsilon_n}(a,dc)\ge \sum_{a\in A}\int_{c\in C}f(a,c)x^*(a,dc)-\epsilon_n.
\end{equation}
The inequality \eqref{proof:eps_optimal} can then be directly obtained by taking limits on both sides of \eqref{proof:eps_optimal_prop3} .
\subsection{Proof of Theorem \ref{math_thm_final_mh}}\label{app: math_thm_final_mh}
The Theorem follows from the following two propositions
\begin{prop}\label{math_thm_heu_mh}
We assume that the sequence $(\mu^k)_{k=1}^\infty$ satisfies
$$
\sum_{k=1}^{\infty}\mu^k=\infty.
$$
We assume that $(\lambda^k,\gamma^k)$ converges to some $(\lambda^*,\gamma^*)$ when $k\rightarrow\infty$ and $x^*$ is the solution to system
\eqref{math_lot_mh}. Then for any $\epsilon>0$, there exists $N\in \mathbb{N}_+$, such that when $n>N$, $x^n$ is an $\epsilon$-optimal solution to system \eqref{math_lot_mh}.
\end{prop}

\begin{proof}
To prove the proposition, we need to prove the following properties for $x^n$:
\begin{enumerate}
    \item The constraint \eqref{math_lot_mh_relax1}(with $x$ replaced by $x^n$) holds;
    \item The constraint \eqref{math_lot_mh_relax2}(with $x$ replaced by $x^n$) holds;
    \item The property \eqref{eq:property3}(with $\tilde{x}^{\epsilon}$ replaced by $x^n$) holds.
\end{enumerate}
We first prove property 1, i.e. 
$$
\sum_{a\in A}\int_{c\in C}g_i(a,c)x^{n}(a,dc)\le \epsilon,\quad \text{for } i\in\{1,\cdots,m\}.
$$ 
For $i\in\{1,...,m\}$, we have
    $$
    \sum_{a\in A}\int_{c\in C}g_i(a,c)x^n(a,dc)=\frac{1}{\sum_{k=1}^{n}\mu_k}\sum_{k=1}^{n}\mu_kg_i(a^k,c^k).
    $$
by the definition of $x^n$ that
$$
x^n=\frac{1}{\sum_{k=1}^n\mu^k}\sum_{k=1}^{n}\mu^k\delta_{(a^k,c^k)}.
$$
By the updating rule for $\lambda_i$, written $$\lambda_i^{k+1}=\max\{\lambda_i^k+\mu^kg_i(a^k,c^k),0\}\ge \lambda_i^k+\mu^kg_i(a^k,c^k), \quad k=1,\cdots, n,$$ we know that
$$
\sum_{k=1}^{n}\lambda_i^{k+1}\ge \sum_{k=1}^{n}\left[\lambda_i^k+\mu^kg_i(a^k,c^k)\right],
$$
which can be simplified as
$$
\lambda_i^{n+1}\ge \lambda_i^1+\sum_{k=1}^{n}\mu_kg_i(a^k,c^k)=\lambda_i^{1}+\left(\sum_{k=1}^{n}\mu^k\right)\sum_{a\in A}\int_{c\in C}g_i(a,c)x^n(a,dc).
$$
Hence
\begin{equation}\label{proof_prop1_mh}
\sum_{a\in A}\int_{c\in C}g_i(a,c)x^n(a,dc)\le\frac{\lambda_i^{n+1}-\lambda_i^1}{\sum_{k=1}^{n}\mu^k}.
\end{equation}
Property 1 is then implied by the fact that $\lambda_i^{n+1}\rightarrow\lambda_i^{*}$ (which leads to the fact that the numerator $\lambda_i^{n+1}-\lambda_i^1$ is bounded) and $\sum_{k=1}^{n}\mu^k\rightarrow\infty$.

We then prove property 2. The spirit of the proof is the same as that for property 1. For $j\in\{1,\cdots,l\},a\in A$, we have
$$
\int_{c\in C}h_j(a,c)x^n(a,dc)=\frac{1}{\sum_{k=1}^{n}\mu^k}\sum_{\{k:a^k=a\}}\mu^kh_{j}(a,c^k),
$$
by the definition of $x^n$. By the updating rule of $\gamma_{j,a}$, we have
$$
\gamma^{n+1}_{j,a}\ge \gamma^{1}_{j,a}+\sum_{\{k:a^k=a\}}\mu^kh_{j} (a,c^k)=\gamma_{j,a}^1+\left(\sum_{k=1}^{n}\mu^k\right)\int_{c\in C}h_j(a,c)x^n(a,dc).
$$
Hence 
\begin{equation}\label{proof_prop2_mh}
\int_{c\in C}h_j(a,c)x^n(a,dc)\le\frac{\gamma_{j,a}^{n+1}-\gamma_{j,a}^1}{\sum_{k=1}^{n}\mu^k}.
\end{equation}
Property 2 holds by the fact that $\gamma_{j,a}^{n}\rightarrow\gamma_{j,a}^*$ and $\sum_{k=1}^{n}\mu^k\rightarrow \infty$.

Now we prove property 3, i.e.
\begin{equation}\label{eq:proof_prop3_mh}
\sum_{a\in A}\int_{c\in C}f(a,c)x^n(a,dc)\ge \sum_{a\in A}\int_{c\in C}f(a,c)x^*(a,dc)-\epsilon,
\end{equation}
where $x^*$ is the optimal solution of \eqref{math_lot_mh}. The proof for property 3 can be divided into two steps. 

\textbf{Step 1. }We define the sets
$$
\Lambda^g_1:=\{i\in\{1,2...,m\},\lambda_i^*>0\},\,\Lambda^h_1:=\{  {(j,a)}\in \{1,\cdots,l\}\times A,  {\gamma_{j,a}^*}>0\}.
$$
$$
\Lambda^g_2:=\{i\in\{1,2...,m\},\lambda_i^*=0\},\,\Lambda^h_2:=\{  {(j,a)}\in \{1,\cdots,l\}\times A,  {\gamma_{j,a}^*}=0\}.
$$
In this step, we prove that when $i\in \Lambda^g_1$, we have 
\begin{equation}\label{proof_prop3_claim_mh}
\sum_{a\in A}\int_{c\in C}g_i(a,c)x^{n}(a,dc)\rightarrow0,
\end{equation}
and when $(j,a)\in \Lambda_1^h$, we have
\begin{equation}\label{proof_prop3_claim_mh2}
\int_{c\in C}  {h_j(a,c)}x^{n}(a,dc)\rightarrow0.
\end{equation}
Before we prove \eqref{proof_prop3_claim_mh} and \eqref{proof_prop3_claim_mh2}, we first interpret these two formulas. Since we have already proven the property 1 and property 2 of definition \ref{def:eps_opt_sol} hold, we now have
$$
\sum_{a\in A}\int_{c\in C}g_i(a,c)x^{n}(a,dc)\le \epsilon, \quad \forall i\in \{1,\cdots,m\},
$$
and
$$
\int_{c\in C}  {h_j(a,c)}x^{n}(a,dc)\le \epsilon,\quad \forall (j,a)\in \{1,\cdots,\ell\}\times A.
$$
These two formulas \eqref{proof_prop3_claim_mh} and \eqref{proof_prop3_claim_mh2} further state that the complementary conditions
$$
\lambda_i^*\sum_{a\in A}\int_{c\in C}g_i(a,c)x^n(a,dc)=0, \quad \forall i\in \{1,\cdots, m\},
$$
and
$$
  {\gamma_{j,a}^*}\int_{c\in C}  {h_j(a,c)}x^{n}(a,dc)=0, \quad \forall (j,a)\in \{1,\cdots, \ell\}\times A,
$$
are approximately satisfied when $n$ is sufficiently large, which in a way implies that when $i\in \Lambda_1^g$ or   {$(j,a)\in \Lambda_1^h$}, the corresponding constraints are active.

Now we prove \eqref{proof_prop3_claim_mh}. For every $i\in \Lambda^g_1$, according to the fact that $\lambda_i^n\rightarrow \lambda_i^*>0$, we know that there exists $N_i\in \mathbb{N}_+$, such that when $n>N_i$, we have $\lambda_i^n>0$. Then by the updating rule of $\lambda_i^n$, when $n\ge N_i$ we have
 $$0<\lambda_i^{n+1}=\max\{\lambda_i^n+\mu^ng_i(a^n,c^n),0\},$$
and thus
$
\lambda_i^{n+1}=\lambda_i^{n}+\mu^ng_i(a^n,c^n),\quad \forall n\ge N_i.
$
Therefore 
$$
\begin{aligned}
\lambda_i^{N_i+n+1}&=\lambda_i^{N_i}+\sum_{k=1}^{n}\mu^{N_i+k}g_i(a^{N_i+k},c^{N_i+k})\\
&=\lambda_i^{N_i}+\sum_{k=1}^{N_i+n}\mu^{k}g_i(a^k,c^{k})-\sum_{k=1}^{N_i}\mu^{k}g_i(a^k,c^{k})\\
&=\lambda_i^{N_i}+\left(\sum_{k=1}^{N_i+n}\mu^k\right)\sum_{a\in A}\int_{c\in C}g_i(a,c)x^{N_i+n}(a,dc)-\left(\sum_{k=1}^{N_i}\mu^k\right)\sum_{a\in A}\int_{c\in C}g_i(a,c)x^{N_i}(a,dc),\quad \forall n\ge 0.
\end{aligned}
$$
Hence
\begin{equation}\label{proof_prop3_comple_mh}
\sum_{a\in A}\int_{c\in C}g_i(a,c)x^{N_i+n}(a,dc)=\frac{\lambda_i^{N_i+n+1}-\lambda_i^{N_i}}{\sum_{k=1}^{N_i+n}\mu^k}+\frac{\sum_{k=1}^{N_i}\mu^k}{\sum_{k=1}^{N_i+n}\mu^k}\sum_{a\in A}\int_{c\in C}g_i(a,c)x^{N_i}(a,dc).
\end{equation}
Since $\lambda_i^{n}\rightarrow\lambda_i^*$ (hence $\lambda_i^{N_i+n+1}-\lambda_i^{N_i}$ is bounded) and $\sum_{k=1}^{N_i+n}\mu^k\rightarrow\infty$, the right hand side of (\ref{proof_prop3_comple_mh}) converges to 0 as $n \rightarrow \infty$, and we have finished the claim (\ref{proof_prop3_claim_mh}). The proof for \eqref{proof_prop3_claim_mh2} is the same as the one for \eqref{proof_prop3_claim_mh} so we omit the details here. 

\textbf{Step 2. } We  utilize claims (\ref{proof_prop3_claim_mh}) and (\ref{proof_prop3_claim_mh2}) to prove \eqref{eq:proof_prop3_mh}. By the definition of $x^n$, we have
\begin{equation}\label{proof_prop3_fformula_mh}
\sum_{a\in A}\int_{c\in C}f(a,c)x^n(a,dc)=\frac{1}{\sum_{k=1}^n\mu^k}\sum_{k=1}^{n}\mu^kf(a^k,c^k)
\end{equation}
By the fact that
$$
(a^k,c^k)\in\arg\max_{a\in A,c\in C}f(a,c)-\sum_{i=1}^{m}\lambda_i^kg_i(a,c)-\sum_{j=1}^{\ell}\gamma_{j,a}^k  {h_j(a,c)},
$$
we have
\begin{equation}\label{eq:proof_prop3_maxi}
  {\begin{aligned}
&f(a^k,c^k)-\sum_{i=1}^{m}\lambda_i^kg_i(a^k,c^k)-\sum_{j=1}^{\ell}\gamma_{j,a^k}^kh_{j}(a^k,c^k)\\
\ge &f(a,c)-\sum_{i=1}^{m}\lambda_i^kg_i(a,c)-\sum_{j=1}^{\ell}\gamma_{j,a}^kh_j(a,c),\,\text{for all }a\in A,c\in C.
\end{aligned}}
\end{equation}
The key to the proof is to utilize the inequality \eqref{eq:proof_prop3_maxi} for each $1\le k\le n$ and give a lower bound for the right-hand side of \eqref{proof_prop3_fformula_mh}. To be precise, for any $a\in A,\,c\in C,$ we have
\begin{equation}\label{eq:proof_prop3_fineq1}
  {\begin{aligned}
    &\frac{1}{\sum_{k=1}^n\mu^k}\sum_{k=1}^n\mu^kf(a^k,c^k)\\
    =&\frac{1}{\sum_{k=1}^n\mu^k}\sum_{k=1}^n\mu^k\left(f(a^k,c^k)-\sum_{i=1}^{m}\lambda_i^kg_i(a^k,c^k)-\sum_{j=1}^{\ell}\gamma_{j,a^k}^kh_{j}(a^k,c^k)\right.\\
&\left.+\sum_{i=1}^{m}\lambda_i^kg_i(a^k,c^k)+\sum_{j=1}^{\ell}\gamma_{j,a^k}^kh_{j}(a^k,c^k)\right)\\
    \ge&\frac{1}{\sum_{k=1}^n\mu^k}\sum_{k=1}^n\mu^k\left(f(a,c)-\sum_{i=1}^{m}\lambda_i^kg_i(a,c)-\sum_{j=1}^{\ell}\gamma_{j,a}^kh_j(a,c)\right.\\
&\left.+\sum_{i=1}^{m}\lambda_i^kg_i(a^k,c^k)+\sum_{j=1}^{\ell}\gamma_{j,a^k}^kh_{j}(a^k,c^k)\right)\\
\end{aligned}}
\end{equation}
We utilize the relations $\lambda_i^k=(\lambda_i^k-\lambda_i^*)+\lambda_i^*$ and   {$\gamma_{j,a}^k=(\gamma_{j,a}^k-\gamma_{j,a}^*)+\gamma_{j,a}^*$} to simplify the terms $\sum_{i}^{m}\lambda_i^kg_i(a^k,c^k)$ and   {$\sum_{j=1}^{\ell}\gamma_{j,a^k}^kh_{j}(a^k,c^k)$} in the brackets on the right hand side of \eqref{eq:proof_prop3_fineq1}, and have
\begin{equation}\label{eq:proof_prop3_fineq2}
  {\begin{aligned}
    &\frac{1}{\sum_{k=1}^n\mu^k}\sum_{k=1}^n\mu^kf(a^k,c^k)\\
    \ge &\frac{1}{\sum_{k=1}^n\mu^k}\sum_{k=1}^n\mu^k\left(f(a,c)-\sum_{i=1}^{m}\lambda_i^kg_i(a,c)-\sum_{j=1}^{\ell}\gamma_{j,a}^k  {h_j(a,c)}
    +\sum_{i=1}^{m}\lambda_i^*g_i(a^k,c^k)+\sum_{j=1}^{\ell}\gamma_{j,a^k}^*h_{j}(a^k,c^k) \right.\\
    &\left.+\sum_{i=1}^{m}(\lambda_i^k-\lambda_i^*)g_i(a^k,c^k)+\sum_{j=1}^{\ell}(\gamma_{j,a^k}^k-\gamma_{j,a^k}^*)h_{j}(a^k,c^k)\right)\\
    =&f(a,c)-\frac{1}{\sum_{k=1}^{n}\mu^k}\sum_{k=1}^{n}\mu^k\sum_{i=1}^{m}\lambda_i^kg_i(a,c)-\frac{1}{\sum_{k=1}^{n}\mu^k}\sum_{k=1}^{n}\mu^k\sum_{j=1}^{\ell}\gamma_{j,a}^k  {h_j(a,c)}\\
    &+\sum_{i=1}^{m}\lambda_i^*\sum_{a\in A}\int_{c\in C}g_i(a,c)x^n(a,dc)+\sum_{j=1}^{\ell}\sum_{a\in A}\gamma_{j,a}^*\int_{c\in C}   {h_j(a,c)}x^n(a,dc)\\
    &+\frac{1}{\sum_{k=1}^n\mu^k}\sum_{k=1}^n\mu^k\left( \sum_{i=1}^{m}(\lambda_i^k-\lambda_i^*)g_i(a^k,c^k)\right)+\frac{1}{\sum_{k=1}^n\mu^k}\sum_{k=1}^n\mu^k\left( \sum_{j=1}^{\ell}(\gamma_{j,a^k}^k-\gamma_{j,a^k}^*)h_{j}(a^k,c^k)\right).
\end{aligned}}
\end{equation}
We further define that
$$
I_1=f(a,c)-\frac{1}{\sum_{k=1}^{n}\mu^k}\sum_{k=1}^{n}\mu^k\sum_{i=1}^{m}\lambda_i^kg_i(a,c)-\frac{1}{\sum_{k=1}^{n}\mu^k}\sum_{k=1}^{n}\mu^k\sum_{j=1}^{\ell}\gamma_{j,a}^k  {h_j(a,c)},
$$
$$
I_2=\sum_{i=1}^{m}\lambda_i^*\sum_{a\in A}\int_{c\in C}g_i(a,c)x^n(a,dc)+\sum_{j=1}^{\ell}\sum_{a\in A}\gamma_{j,a}^*\int_{c\in C}   {h_j(a,c)}x^n(a,dc),
$$
and
$$
I_3=\frac{1}{\sum_{k=1}^n\mu^k}\sum_{k=1}^n\mu^k\left( \sum_{i=1}^{m}(\lambda_i^k-\lambda_i^*)g_i(a^k,c^k)\right)+\frac{1}{\sum_{k=1}^n\mu^k}\sum_{k=1}^n\mu^k\left(\sum_{j=1}^{\ell}(\gamma_{j,a^k}^k-\gamma_{j,a^k}^*)h_{j}(a^k,c^k)\right).
$$
The inequality \eqref{eq:proof_prop3_fineq2} can then be written as
\begin{equation}\label{eq:proof_prop3_fineqfinal}
\begin{aligned}
    &\frac{1}{\sum_{k=1}^n\mu^k}\sum_{k=1}^n\mu^kf(a^k,c^k)\ge I_1+I_2+I_3.
\end{aligned}
\end{equation}
Next, we give a lower bound estimation of $I_3, I_2$, and $I_1$ respectively. For $I_3$, By the facts that $\lambda_i^k\rightarrow\lambda_i^*$ and $\gamma_{j,a}^k\rightarrow \gamma_{j,a}^*$, together with the boundedness of $g_i$ and $h_j$, we know that
\begin{equation}\label{eq:proof_prop3_I3est1}
 \sum_{i=1}^{m}(\lambda_i^k-\lambda_i^*)g_i(a^k,c^k)+
 \sum_{j=1}^{\ell}(\gamma_{j,a^k}^k-\gamma_{j,a^k}^*)h_{j}(a^k,c^k)\rightarrow 0.
\end{equation}
Since $\sum_{k=1}^{\infty}\mu^k=\infty$, \eqref{eq:proof_prop3_I3est1} yields
$$
I_3=\frac{1}{\sum_{k=1}^n\mu^k}\sum_{k=1}^n\mu^k\left( \sum_{i=1}^{m}(\lambda_i^k-\lambda_i^*)g_i(a^k,c^k)\right)+\frac{1}{\sum_{k=1}^n\mu^k}\sum_{k=1}^n\mu^k\left( \sum_{j=1}^{\ell}(\gamma_{j,a^k}^k-\gamma_{j,a^k}^*)h_{j}(a^k,c^k)\right)\rightarrow 0.
$$
Hence, for any $\epsilon>0$, there exists $N_1\in \mathbb{N}_+$, such that when $n>N_1$, 
\begin{equation}\label{eq:proof_prop3_I3estfinal}
I_3>-\frac{\epsilon}{2}.
\end{equation}
Next, we consider the lower bound of $I_2$. By the formulas \eqref{proof_prop3_claim_mh} and \eqref{proof_prop3_claim_mh2} proven in step 1, we have
$$
\begin{aligned}
I_2&=\sum_{i=1}^{m}\lambda_i^*\sum_{a\in A}\int_{c\in C}g_i(a,c)x^n(a,dc)+\sum_{j=1}^{\ell}\sum_{a\in A}\gamma_{j,a}^*\int_{c\in C}   {h_j(a,c)}x^n(a,dc)\\
&=\sum_{i\in\Lambda_1^g}\lambda_i^*\sum_{a\in A}\int_{c\in C}g_i(a,c)x^n(a,dc)+\sum_{(j,a)\in\Lambda_1^h}\gamma_{j,a}^*\int_{c\in C}   {h_j(a,c)}x^n(a,dc)\\
&\rightarrow 0.
\end{aligned}
$$
Therefore there exists $N_2\in \mathbb{N}_+$, such that when $n>N_2$, we have
\begin{equation}\label{eq:proof_prop3_I2estfinal}
    I_2>-\frac{\epsilon}{2}.
\end{equation}
We combine \eqref{eq:proof_prop3_fineqfinal}, \eqref{eq:proof_prop3_I3estfinal} and \eqref{eq:proof_prop3_I2estfinal} and see that when $n>\max\{N_1,N_2\}$, we have
\begin{equation}\label{proof_prop3_mainine2_mh}
\begin{aligned}
&\frac{1}{\sum_{k=1}^n\mu^k}\sum_{k=1}^n\mu^kf(a^k,c^k)
\ge  I_1-\epsilon\\
=&f(a,c)-\frac{1}{\sum_{k=1}^{n}\mu^k}\sum_{k=1}^{n}\mu^k\sum_{i=1}^{m}\lambda_i^kg_i(a,c)-\frac{1}{\sum_{k=1}^{n}\mu^k}\sum_{k=1}^{n}\mu^k\sum_{j=1}^{\ell}\gamma_{j,a}^k  {h_j(a,c)}-\epsilon
\end{aligned}
\end{equation}
We multiply $x^*(a,dc)$ to both sides of (\ref{proof_prop3_mainine2_mh}) and integrate over $A\times C$. The inequality (\ref{proof_prop3_mainine2_mh}) then implies
$$
\begin{aligned}
&\frac{1}{\sum_{k=1}^n\mu^k}\sum_{k=1}^n\mu^kf(a^k,c^k)\\
\ge &\sum_{a\in A}\int_{c\in C}f(a,c)x^*(a,dc)-\frac{1}{\sum_{k=1}^{n}\mu^k}\sum_{k=1}^{n}\mu^k\sum_{i=1}^{m}\lambda_i^k\sum_{a\in A}\int_{c\in C}g_i(a,c)x^*(a,dc)\\
&-\frac{1}{\sum_{k=1}^{n}\mu^k}\sum_{k=1}^{n}\mu^k\sum_{j=1}^{\ell}\sum_{a\in A}\gamma_{j,a}^k\int_{c\in C}  {h_j(a,c)}x^*(a,dc)-\epsilon\\
\ge &\sum_{a\in A}\int_{c\in C}f(a,c)x^*(a,dc)-\epsilon,
\end{aligned}
$$
where the last inequality holds by the fact that $x^*$ satisfies all the constraints in \eqref{math_lot_mh}. Hence, we prove Property 3 and finish the proof for Proposition  \ref{math_thm_heu_mh}.
\end{proof}
\begin{prop}\label{math_thm_lagconv_mh}
We assume that the sequence $(\mu^k)_{k=1}^\infty$ satisfies
$$
\sum_{k=1}^{\infty}\mu^k=\infty, \text{\quad and\quad}\sum_{k=1}^{\infty}(\mu^k)^2<\infty.
$$
Let $x^*$ be the solution to system \eqref{math_lot_mh}, and corresponding Lagrangian multipliers to $x^*$ exist. Then $(\lambda^k,\gamma^k)$ generated by Algorithm \ref{math_alg_mh} converge to some $(\lambda^*,\gamma^*)$, where $(\lambda^*,\gamma^*)$ is a minimizer of the dual problem 
$$
\inf_{(\lambda,\gamma)\in \mathbb{R}_+^m\times\mathbb{R}_+^{\ell |A|}} V(\lambda,\gamma):=\max_{a\in A,\,c\in C} \mathcal{L}(a,c;\lambda,\gamma),
$$
where $\mathcal{L}$ is defined in \eqref{defcalL}.
\end{prop}
\begin{proof}
We define the set of minimizers to the dual problem as
\begin{equation}\label{def:optimallambda}
\Lambda^*:=\left\{(\lambda^*,\gamma^*)\in\mathbb{R}_+^m\times\mathbb{R}_+^{\ell |A|},\,V(\lambda^*,\gamma^*)= \inf_{(\lambda,\gamma)\in\mathbb{R}_+^m\times\mathbb{R}_+^{\ell |A|}}V(\lambda,\gamma)\right\}.
\end{equation}
 \eqref{dual} in Theorem \ref{thm:dual=lot} implies that $\Lambda^*$ is a non-empty set. Furthermore, for any $(\lambda,\gamma)\in\mathbb{R}_+^m\times\mathbb{R}_+^{\ell |A|}$, according to Lemma \ref{lem:subgradient}, $(\Delta\lambda,\Delta \gamma)$ defined in \eqref{def_updating_direction} is a negative sub-gradient of $V(\lambda,\gamma)$, and 
$$
\|(\Delta \lambda,\Delta \gamma)\|_2^2=\sum_{i=1}^{m}|g_i(a_\lambda,c_{\lambda})|^2+\sum_{j=1}^{\ell}|h_j(a_{\lambda},c_{\lambda})|^2
$$
is uniformly bounded due to the continuity of the constraints $g_i,\,h_j$, and the compactness of the consumption set $C$ and the action set $A$. The theorem therefore follows from Proposition 2.7 in \cite{nedic2001convergence}. 
\end{proof}
\subsection{Proof of Proposition \ref{thm_complexity}}\label{app:thm_complexity}
The proof can be divided into two steps. At the first step, we show the properties of the lottery solution $x^n$ generated from Algorithm \ref{math_alg_mh} as follows:
\begin{equation}\label{proof_complexity_g}
\sum_{a \in A} \int_{c \in C} g_i(a, c) x^n(a, \mathrm{d}c) \le \frac{{\bar{\Lambda}}}{\sum_{k=1}^{n}\mu^k};
\end{equation}
\begin{equation}\label{proof_complexity_h}
\int_{c \in C} h_{j}(a, c) x^n(a, \mathrm{d}c) \le \frac{{\bar{\Lambda}}}{\sum_{k=1}^{n}\mu^k};
\end{equation}
and
\begin{equation}\label{proof_complexity_f}
\sum_{a \in A} \int_{c \in C} f(a, c) x^{n}(a, \mathrm{d}c)\ge \sum_{a \in A} \int_{c \in C} f(a, c) x^{*}(a, \mathrm{d}c)-\frac{M^2(m+\ell)\sum_{k=1}^{n}(\mu^k)^2+\bar{\Lambda}}{\sum_{k=1}^{n}\mu^k}.
\end{equation}
At the second step, we show that we can choose $\mu^k$ satisfying the condition $\mu^k\sim k^{-\frac{1}{2}(1+\rho)}$, such that when \eqref{eq:n_itertime_mh} holds, $x^n$ is an $\epsilon$-optimal solution.

\textbf{Step 1.} According to \eqref{proof_prop1_mh} and \eqref{proof_prop2_mh} in the proof of Proposition  \ref{math_thm_heu_mh}, we have:
$$
\sum_{a \in A} \int_{c \in C} g_i(a, c) x^n(a, \mathrm{d}c) \le \frac{\lambda_i^{n+1} - \lambda_i^1}{\sum_{k=1}^{n} \mu^k}\le \frac{{\lambda_i^{n+1}}}{\sum_{k=1}^{n}\mu^k}\le \frac{{\bar{\Lambda}}}{\sum_{k=1}^{n}\mu^k};
$$
and
$$
\int_{c \in C} h_{j}(a, c) x^n(a, \mathrm{d}c) \le \frac{\gamma_{ a,j}^{n+1} - \gamma_{a,j}^1}{\sum_{k=1}^{n} \mu^k}\le \frac{{\gamma_{a,j}^{n+1}}}{\sum_{k=1}^{n}\mu^k}\le \frac{{\bar{\Lambda}}}{\sum_{k=1}^{n}\mu^k}.
$$
Therefore, \eqref{proof_complexity_g} and \eqref{proof_complexity_h} hold.

It then remains to show that \eqref{proof_complexity_f} holds in this step. The proof for \eqref{proof_complexity_f} can be further divided into three sub-steps. First we show that \begin{equation}\label{thm_complexity_fdiff_final}
\begin{aligned}
    &\sum_{a \in A} \int_{c \in C} f(a, c) x^{n}(a, \mathrm{d}c)
    \ge\sum_{a \in A} \int_{c \in C} f(a, c) x^*(a, \mathrm{d}c) \\
    &+\frac{1}{\sum_{k=1}^{n}\mu^k}\sum_{i=1}^{m}\sum_{k=1}^{n}\mu^k\lambda_i^kg_i(a^k,c^k)+\frac{1}{\sum_{k=1}^n\mu^k}\sum_{j=1}^{\ell}\sum_{k=1}^{n}\mu^k\gamma_{j,a^k}^kh_{j}(a^k,c^k).
\end{aligned}
\end{equation}
Then, we give an estimate of the term
\begin{equation}\label{eq:complexity_estimate}
    \frac{1}{\sum_{k=1}^{n}\mu^k}\sum_{i=1}^{m}\sum_{k=1}^{n}\mu^k\lambda_i^kg_i(a^k,c^k)+\frac{1}{\sum_{k=1}^n\mu^k}\sum_{j=1}^{\ell}\sum_{k=1}^{n}\mu^k\gamma_{j,a^k}^kh_{j}(a^k,c^k)
\end{equation}
in the right hand side of \eqref{thm_complexity_fdiff_final}. Finally, we conclude \eqref{proof_complexity_f}.

\paragraph{1-1. Show that \eqref{thm_complexity_fdiff_final} holds.} According to \eqref{proof_prop3_fformula_mh} and \eqref{eq:proof_prop3_fineq1} in the proof of Proposition  \ref{math_thm_heu_mh}, for every $a\in A,\,c\in C$, we have\footnote{There will be a slightly abuse of notations here. The $(a,c)$ in the first line just means the variable of integration; and the $(a,c)$ in the rest refers to the arbitrary chosen $a\in A,\,c\in C$.}: 
\begin{equation}\label{thm_complexity_fdiff}
\begin{aligned}
    &\sum_{a \in A} \int_{c \in C} f(a, c) x^{n}(a, \mathrm{d}c)\\
    \ge&\frac{1}{\sum_{k=1}^n\mu^k}\sum_{k=1}^n\mu^k\left(f(a,c)-\sum_{i=1}^{m}\lambda_i^kg_i(a,c)-\sum_{j=1}^{\ell}\gamma_{j,a}^k  {h_j(a,c)}+\sum_{i=1}^{m}\lambda_i^kg_i(a^k,c^k)+\sum_{j=1}^{\ell}\gamma_{j,a^k}^kh_{j}(a^k,c^k)\right)\\
    =&f(a,c)-\frac{1}{\sum_{k=1}^{n}\mu^k}\sum_{i=1}^{m}\sum_{k=1}^{n}\mu^k\lambda_i^kg_i(a,c)-\frac{1}{\sum_{k=1}^n\mu^k}\sum_{j=1}^{\ell}\sum_{k=1}^{n}\mu^k\gamma_{j,a}^k  {h_j(a,c)}\\
    &+\frac{1}{\sum_{k=1}^{n}\mu^k}\sum_{i=1}^{m}\sum_{k=1}^{n}\mu^k\lambda_i^kg_i(a^k,c^k)+\frac{1}{\sum_{k=1}^n\mu^k}\sum_{j=1}^{\ell}\sum_{k=1}^{n}\mu^k\gamma_{j,a^k}^kh_{j}(a^k,c^k).
\end{aligned}
\end{equation}
\normalsize
We multiply $x^*(a,dc)$ to both sides in \eqref{thm_complexity_fdiff} and integral on $A\times C$ and have
\begin{equation}\label{thm_complexity_fdiff_midfinal}
\begin{aligned}
    &\sum_{a \in A} \int_{c \in C} f(a, c) x^{n}(a, \mathrm{d}c)\\
    \ge&\sum_{a \in A} \int_{c \in C} f(a, c) x^*(a, \mathrm{d}c) -\frac{1}{\sum_{k=1}^{n}\mu^k}\sum_{k=1}^{n}\mu^k\sum_{i=1}^{m}\lambda_i^k\sum_{a\in A}\int_{c\in C}g_i(a,c)x^*(a,dc)\\
    &-\frac{1}{\sum_{k=1}^{n}\mu^k}\sum_{k=1}^{n}\mu^k\sum_{j=1}^{\ell}\sum_{a\in A}\gamma_{j,a}^k\int_{c\in C}  {h_j(a,c)}x^*(a,dc)\\
    &+\frac{1}{\sum_{k=1}^{n}\mu^k}\sum_{i=1}^{m}\sum_{k=1}^{n}\mu^k\lambda_i^kg_i(a^k,c^k)+\frac{1}{\sum_{k=1}^n\mu^k}\sum_{j=1}^{\ell}\sum_{k=1}^{n}\mu^k\gamma_{j,a^k}^kh_{j}(a^k,c^k).
\end{aligned}
\end{equation}
Since $x^*(a,dc)$ satisfies all the constraints in \eqref{math_lot_mh}, i.e. 
$$
\sum_{a\in A}\int_{c\in C}g_i(a,c)x^*(a,dc)\le 0, \quad \int_{c\in C}  {h_j(a,c)}x^*(a,dc)\le 0,
$$
we thus obtain the estimate \eqref{thm_complexity_fdiff_final} from \eqref{thm_complexity_fdiff_midfinal}.

\paragraph{1-2. Estimate the term \eqref{eq:complexity_estimate}.} Recall that \eqref{eq:complexity_estimate} is
\[ 
    \frac{1}{\sum_{k=1}^{n}\mu^k}\sum_{i=1}^{m}\sum_{k=1}^{n}\mu^k\lambda_i^kg_i(a^k,c^k)+\frac{1}{\sum_{k=1}^n\mu^k}\sum_{j=1}^{\ell}\sum_{k=1}^{n}\mu^k\gamma_{j,a^k}^kh_{j}(a^k,c^k).
\]
By direct computation, we have
 \begin{equation}\label{thm_complexity_diff_mid1}
 \begin{aligned}
 &\Lambda(\lambda^{k+1},\gamma^{k+1})-\Lambda(\lambda^k,\gamma^k)\\
    = &\left(\sum_{i=1}^{m}2\lambda_i^k\left(\lambda_i^{k+1}-\lambda_i^{k}\right)+\sum_{j=1}^{\ell}2\gamma_{j,a^k}^k\left(\gamma^{k+1}_{j,a^k}-\gamma_{j,a^k}^k\right)\right)+\left(\sum_{i=1}^{m}\left(\lambda_i^{k+1}-\lambda_i^{k}\right)^2+\sum_{j=1}^{\ell}\left(\gamma^{k+1}_{j,a^k}-\gamma_{j,a^k}^k\right)^2\right).
 \end{aligned}
 \end{equation}
We recall that the updating rules for $\lambda_i$, written
$$
\lambda_i^{k+1}=\max\{\lambda_i^{k}+\mu^kg_i(a^k,c^k),0\},
$$
can be rewritten as
\begin{equation}\label{eq:lambda_update_newver}
\lambda_i^{k+1}-\lambda_i^{k}=\mu^k\tau_i^{k}g_i(a^k,c^k),
\end{equation}
for some $0\le \tau_i^{k}\le 1$, where $\tau_i^{k}<1$ if and only if $g_i(a^k,c^k)<0$ and $\lambda_i^{k+1}=0$. Similarly, there exists $0\le \tau_{j,a^k}^{k}\le 1$, such that 
\begin{equation}\label{eq:gamma_update_newver}
\gamma_{j,a^k}^{k+1}-\gamma_{j,a^k}^{k}=\mu^k\tau_{j,,a^k}^{k}h_j(a^k,c^k),
\end{equation}
where $\tau_{j,a^k}^{k}<1$ if and only if $h_j(a^k,c^k)<0$ and $\gamma_{j,a^k}^{k+1}=0$. Therefore we can utilize \eqref{eq:lambda_update_newver} and \eqref{eq:gamma_update_newver} to simplify the right hand side of \eqref{thm_complexity_diff_mid1} and obtain

\begin{equation}\label{thm_complexity_diff_mid2}
    \begin{aligned}
 &\Lambda(\lambda^{k+1},\gamma^{k+1})-\Lambda(\lambda^k,\gamma^k)
 =\left(\sum_{i=1}^{m}2\mu^k\tau^k_ig_i(a^k,c^k)\lambda_i^k+\sum_{j=1}^{\ell} 2\mu^{k}\tau^k_{j,a^k}h_{j}(a^k,c^k)\gamma_{j,a^k}^k\right)\\
&+\left(\sum_{i=1}^m(\mu^k)^2(\tau_i^k)^2g_i^2(a^k,c^k)+\sum_{j=1}^{\ell}(\mu^k)^2(\tau_{j,a^k}^k)^2h^2_{j}(a^k,c^k)\right)\\
=&2\mu^k\left(\sum_{i=1}^{m}g_i(a^k,c^k)\lambda_i^k+\sum_{j=1}^{\ell} h_{j}(a^k,c^k)\gamma_{j,a^k}^k\right)\\
 &-2\mu^k\left(\sum_{i=1}^{m}(1-\tau_i^k)g_i(a^k,c^k)\lambda_i^k+\sum_{j=1}^{\ell} (1-\tau_{j,a^k}^k)h_{j}(a^k,c^k)\gamma_{j,a^k}^k\right)\\
&+\left(\sum_{i=1}^m(\mu^k)^2(\tau_i^k)^2g_i^2(a^k,c^k)+\sum_{j=1}^{\ell}(\mu^k)^2(\tau_{j,a^k}^k)^2h^2_{j}(a^k,c^k)\right).
 \end{aligned}
\end{equation}
Furthermore, when $1-\tau_{i}^{k}>0$ , we know that $\lambda_i^{k+1}=0$, and according to \eqref{eq:lambda_update_newver}, we have
\begin{equation}\label{eq:lambdaeq_small}
    \lambda_i^k=-\mu^k\tau_i^kg_i(a^k,c^k).
\end{equation}
Similarly, when $1-\tau^k_{j,a^k}>0$, we have the following equality according to \eqref{eq:gamma_update_newver}
\begin{equation}\label{eq:gammaeq_small}
    \gamma_{j,a^k}^k=-\mu^k\tau_{j,a^k}^kh_{j}(a^k,c^k).
\end{equation}
We utilize the equality \eqref{eq:lambdaeq_small} and the equality \eqref{eq:gammaeq_small}  to obtain
\begin{equation}\label{ineq:esti_I2_step2}
\begin{aligned}
     &-2\mu^k\left(\sum_{i=1}^{m}(1-\tau_i^k)g_i(a^k,c^k)\lambda_i^{k}+\sum_{j=1}^{\ell}(1-\tau_{j,a^k}^k)h_{j}(a^k,c^k)\gamma_{j,a^k}^k\right)\\
    =&2(\mu^k)^2\left(\sum_{i=1}^m(1-\tau_i^k)\tau_i^kg_i^2(a^k,c^k)+\sum_{j=1}^{\ell}(1-\tau_{j,^k}^k)\tau_{j,a^k}^kh_{j}^2(a^k,c^k)\right).
\end{aligned}
\end{equation}
We combine \eqref{thm_complexity_diff_mid2} and \eqref{ineq:esti_I2_step2}, then we have
$$
\begin{aligned}
&\Lambda(\lambda^{k+1},\gamma^{k+1})-\Lambda(\lambda^{k},\gamma^{k})\\
=& 2\mu^k\left(\sum_{i=1}^{m}g_i(a^k,c^k)\lambda_i^k+\sum_{j=1}^{\ell} h_{j}(a^k,c^k)\gamma_{j,a^k}^k\right)\\
 &-2\mu^k\left(\sum_{i=1}^{m}(1-\tau_i^k)g_i(a^k,c^k)\lambda_i^k+\sum_{j=1}^{\ell} (1-\tau_{j,a^k}^k)h_{j}(a^k,c^k)\gamma_{j,a^k}^k\right)\\
& +\left(\sum_{i=1}^m(\mu^k)^2(\tau_i^k)^2g_i^2(a^k,c^k)+\sum_{j=1}^{\ell}(\mu^k)^2(\tau_{j,a^k}^k)^2h^2_{j}(a^k,c^k)\right)\\
=&2\mu^k\left(\sum_{i=1}^{m}g_i(a^k,c^k)\lambda_i^k+\sum_{j=1}^{\ell} h_{j}(a^k,c^k)\gamma_{j,a^k}^k\right)\\
 &+2(\mu^k)^2\left(\sum_{i=1}^m(1-\tau_i^k)\tau_i^kg_i^2(a^k,c^k)+\sum_{j=1}^{\ell}(1-\tau_{j,^k}^k)\tau_{j,a^k}^kh_{j}^2(a^k,c^k)\right)\\
&+\left(\sum_{i=1}^m(\mu^k)^2(\tau_i^k)^2g_i^2(a^k,c^k)+\sum_{j=1}^{\ell}(\mu^k)^2(\tau_{j,a^k}^k)^2h^2_{j}(a^k,c^k)\right).\\
=&2\mu^k\left(\sum_{i=1}^{m}g_i(a^k,c^k)\lambda_i^k+\sum_{j=1}^{\ell} h_{j}(a^k,c^k)\gamma_{j,a^k}^k\right)\\
    & +(\mu^k)^2\left(\sum_{i=1}^m(2-\tau_i^k)\tau_i^kg_i^2(a^k,c^k)+\sum_{j=1}^{\ell}(2-\tau_{j,a^k}^k)\tau_{j,a^k}^kh_{j}^2(a^k,c^k)\right),
\end{aligned}
$$
which implies that
\begin{equation}\label{improve_diff_final_complexity}
\begin{aligned}
&2\mu^k\left(\sum_{i=1}^{m}g_i(a^k,c^k)\lambda_i^k+\sum_{j=1}^{\ell} h_{j}(a^k,c^k)\gamma_{j,a^k}^k\right)\\
=&\left(\Lambda(\lambda^{k+1},\gamma^{k+1})-\Lambda(\lambda^{k},\gamma^{k})\right)\\
    & -(\mu^k)^2\left(\sum_{i=1}^m(2-\tau_i^k)\tau_i^kg_i^2(a^k,c^k)+\sum_{j=1}^{\ell}(2-\tau_{j,a^k}^k)\tau_{j,a^k}^kh_{j}^2(a^k,c^k)\right).
\end{aligned}
\end{equation}
 We add up equality \eqref{improve_diff_final_complexity} for $k=1,\cdots, n$. After some direct computation, we have an estimate to \eqref{eq:complexity_estimate} as
\begin{equation}\label{thm_complexity_fdiff_part1}
\begin{aligned}
&\frac{1}{\sum_{k=1}^{n}\mu^k}\sum_{i=1}^{m}\sum_{k=1}^{n}\mu^k\lambda_i^kg_i(a^k,c^k)+\frac{1}{\sum_{k=1}^n\mu^k}\sum_{j=1}^{\ell}\sum_{k=1}^{n}\mu^k\gamma_{j,a^k}^kh_{j}(a^k,c^k)\\
=& \frac{\Lambda(\lambda^{n+1}, \gamma^{n+1})- \Lambda(\lambda^{1}, \gamma^{1})}{2\sum_{k=1}^{n} \mu^k}\\
&-\frac{\sum_{k=1}^{n}(\mu^k)^2\left(\sum_{i=1}^m(2-\tau_i^k)\tau_i^kg_i^2(a^k,c^k)+\sum_{j=1}^{\ell}(2-\tau_{j,a^k}^k)\tau_{j,a^k}^kh_{j}^2(a^k,c^k)\right)}{2\sum_{k=1}^{n}\mu^k}.
\end{aligned}
\end{equation}
\normalsize
\paragraph{1-3. We then utilize the estimate \eqref{thm_complexity_fdiff_part1} to derive \eqref{proof_complexity_f}.} We combine \eqref{thm_complexity_fdiff_final} and \eqref{thm_complexity_fdiff_part1} and obtain
\begin{equation}\label{proof_property_f}
\begin{aligned}
&\sum_{a \in A} \int_{c \in C} f(a, c) x^{n}(a, \mathrm{d}c)\\
\ge &\sum_{a \in A} \int_{c \in C} f(a, c) x^*(a, \mathrm{d}c) + \frac{\Lambda(\lambda^{n+1}, \gamma^{n+1})- \Lambda(\lambda^{1}, \gamma^{1})}{2\sum_{k=1}^{n} \mu^k}\\
&-\frac{\sum_{k=1}^{n}(\mu^k)^2\left(\sum_{i=1}^m(2-\tau_i^k)\tau_i^kg_i^2(a^k,c^k)+\sum_{j=1}^{\ell}(2-\tau_{j,a^k}^k)\tau_{j,a^k}^kh_{j}^2(a^k,c^k)\right)}{2\sum_{k=1}^{n}\mu^k}.
\end{aligned}
\end{equation}
By the fact that $|g_i|\le M,\,|h_j|\le M,$ the estimate \eqref{proof_property_f} further implies
$$
\begin{aligned}
&\sum_{a \in A} \int_{c \in C} f(a, c) x^{n}(a, \mathrm{d}c)\\
\ge&  \sum_{a \in A} \int_{c \in C} f(a, c) x^{*
}(a, \mathrm{d}c)-\frac{M^2(  {m+\ell})\sum_{k=1}^{n}(\mu^k)^2+\Lambda(\lambda^1,\gamma^1)}{2\sum_{k=1}^{n}\mu^k}\\
\ge&  \sum_{a \in A} \int_{c \in C} f(a, c) x^{*
}(a, \mathrm{d}c)-\frac{M^2(  {m+\ell})\sum_{k=1}^{n}(\mu^k)^2+\bar{\Lambda}}{\sum_{k=1}^{n}\mu^k}.
\end{aligned}
$$
Therefore, we have shown that \eqref{proof_complexity_f} holds.

\textbf{Step 2.} Choosing \(\mu^k = \frac{1}{M\sqrt{m+\ell}k^{\frac{1}{2}(1 + \rho)}}\), where \(0 < \rho \le 1\), gives:
\begin{equation}\label{proof_complexity_firstord}
\sum_{k=1}^{n} \mu^k \sim
\begin{cases}
\frac{n^{\frac{1}{2}(1 - \rho)}}{M\sqrt{m+\ell}}, & \rho < 1 \\
\frac{\log n}{M\sqrt{m+\ell}}, & \rho = 1,
\end{cases}
\end{equation}
and
\begin{equation}\label{proof_complexity_secondord}
\sum_{k=1}^{n} (\mu^k)^2 \sim \frac{1}{M^2\rho (m+\ell)}.
\end{equation}
Thus, for any \(\epsilon > 0\), when \eqref{eq:n_itertime_mh} holds, i.e.
$$
n \succsim 
\begin{cases}
\left(\frac{M(\frac{1}{\rho}+\bar{\Lambda})}{\epsilon}\right)^{\frac{2}{1 - \rho}}(m+\ell)^{\frac{1}{1-\rho}}, & \rho < 1 \\
e^{\frac{M(\frac{1}{\rho}+\bar{\Lambda})\sqrt{m+\ell}}{\epsilon}}, & \rho = 1,
\end{cases}
$$
we can implement \eqref{proof_complexity_firstord} and \eqref{proof_complexity_secondord} to \eqref{proof_complexity_g},\eqref{proof_complexity_h} and \eqref{proof_complexity_f}, and hence obtain:
$$
\sum_{a \in A} \int_{c \in C} g_i(a, c) x^n(a, \mathrm{d}c) \le \epsilon; \quad
\int_{c \in C} h_{j}(a, c) x^n(a, \mathrm{d}c) \le \epsilon;
$$
and
$$
\begin{aligned}
&\sum_{a \in A} \int_{c \in C} f(a, c) x^{n}(a, \mathrm{d}c)
\ge \sum_{a \in A} \int_{c \in C} f(a, c) x^*(a, \mathrm{d}c) - \epsilon.
\end{aligned}
$$
\subsection{Proof of Lemma \ref{lem:pcmh}}\label{app:pcmh}
For all $ q \in Q $ the partial derivative of the Lagrangian with respect to $c(q) $ can be written as
$$ \frac{\partial {\mathcal L}(a,c;\lambda,\gamma)}{\partial c(q)}=
  -p(q|a) v'(q-c(q)) -\tilde{u}'(c(q)) \left(\sum_{\hat a \in A} \gamma_{a,\hat{a}} (p(q|\hat{a}) C(\hat{a}) -p(q|a) C(a))-\lambda p(q|a) C(a) \right) .$$
 If for a given $a\in A$ the term in the brackets is positive, the derivative is negative and the function is strictly differentiably monotone and hence (since one-dimensional) pseudo-concave.
    If the term in the brackets is negative the partial derivative is decreasing in $ c(q) $ and hence the Lagrangian is concave.

\subsection{Proof of Proposition \ref{prop:etaconv}}\label{app:etaconv}

To prove the proposition, we first need the following result which shows that if the original IC constraint \ref{eq:tax_ic2} is replaced by 
\ref{eq:tax_ic_new_eta}, the optimal solution never involves lotteries.
\begin{prop}\label{prop:noloteta}
         We consider the planner problem \eqref{eq:tax_obj2} subject to \eqref{eq:tax_rc} and \eqref{eq:tax_ic_new_eta}  with the consumption set $C=\times_{h\in H}[\epsilon,c_{\max}]$, the output set $Y=\times_{h\in H} [0,a_{\max}\omega_h]$, and the utility function defined by \eqref{eq:utility_func}. Assume that the Slater's condition holds, then the optimal solution to the planner problem does not involve lotteries.
\end{prop}
To proceed in the proofs of the propositions, we will need the following technical lemma.

\begin{lemma}\label{lem:technical2}
    Given $i\in\{1,\cdots, N\}$, and $1< p_1<\cdots<p_N$. We consider the function
$$
f_i(y)=\sum_{j=i+1}^{N}a_jy^{p_j}-a_iy^{p_i}+\sum_{j=1}^{i-1}a_jy^{p_j}+a_0y,
$$
where $a_j\ge0,\,\forall j\ne i$. and $a_0>0$. Then one of the following holds:
\begin{enumerate}
    \item $f_i$ admits a unique maximal point in the interval $[0,\omega_ha_{\max}]$.
    \item $f_i$ admits two maximal points in the interval $[0,\omega_ha_{\max}]$, and one of them is $w_ha_{\max}$.
\end{enumerate}
Furthermore, when $a_{i+1}=\cdots=a_{N}=0$, then $f_i$ admits a unique maximal point in the interval $[0,\omega_ha_{\max}]$.
\end{lemma} 
\begin{proof}
 If $a_i\le 0$, then $f_i(y)$ is strictly increasing in the interval $[0,\omega_h a_{\max}]$ and has a unique maximal point $y^*=\omega_h a_{\max}$.

      If $a_i >0$, we define $p_0=1$, and it is straightforward to have
    $$
    f_i'(y)=\sum_{j=i+1}^{N}a_{j}p_jy^{p_j-1}-a_ip_iy^{p_i-1}+\sum_{j=0}^{i-1}a_jp_jy^{p_j-1},
    $$
    and
    $$
    \begin{aligned}
    f_i''(y)=&\sum_{j=i+1}^{N}a_jp_j(p_j-1)y^{p_j-2}-a_ip_i(p_i-1)y^{p_i-2}+\sum_{j=1}^{i-1}a_jp_j(p_j-1)y^{p_j-2}\\
    =&\sum_{j=i+1}^{N}a_jp_j(p_j-1)y^{p_j-2}+\frac{p_i-1}{y}(f_i'(y)-\sum_{j=0}^{i-1}a_jp_jy^{p_j-1}-\sum_{j=i+1}^{N}a_jp_jy^{p_j-1})+\sum_{j=1}^{i-1}a_jp_j(p_j-1)y^{p_j-2}\\
    =&\frac{(p_i-1)f'_i(y)}{y}-\sum_{j=0}^{i-1}a_jp_j(p_i-p_j)y^{p_j-2}+\sum_{j=i+1}^{N}a_jp_j(p_j-p_i)y^{p_j-2}.
    \end{aligned}
    $$
Therefore, 
\begin{equation}\label{eq:fpibehavior}
\begin{aligned}
\left[\frac{f'_i(y)}{y^{p_i-1}}\right]'&=\frac{f''_i(y)}{y^{p_i-1}}-\frac{(p_i-1)f_i'(y)}{y^{p_i}}\\
&= y^{-1}\left[\sum_{j=i+1}^{N}a_jp_j(p_j-p_i)y^{p_j-p_i}-\sum_{j=0}^{i-1}a_jp_j(p_i-p_j)y^{p_j-p_i}\right]\\
&:=y^{-1} g_i(y) .
\end{aligned}
\end{equation}
Since $a_0>0$, $g_i(y)$ is strictly increasing in $y$, and $g_i(0)=-\infty$. Therefore, $f_i'(y)/y^{p_i-1}$ is either
\begin{enumerate}
    \item strictly decreasing, or
    \item strictly decreasing in $[0,\bar{y}]$ for some $\bar{y}\in[0,\omega_ha_{\max}]$ and then strictly increasing in $[\bar{y},\omega_ha_{\max}]$.
\end{enumerate}
We take $y=0$ in $f_i'(y)$ and obtain $f_i'(0)=a_0>0$. By the continuity of $f_i'(y)$, there exists $\delta>0$, such that when $y\in [0,\delta]$, $f'(y)>0$.  

For case 1, if for any $y\in [\delta,\omega_ha_{\max}]$, $\frac{f'_i(y)}{y^{p_i-1}}> 0$, then $f'_i(y)>0$, $f_i(y)$ is strictly increasing and has a unique maximal point $y^*=\omega_ha_{\max}$; if there exists $y_0\in [\delta,\omega_ha_{\max}]$, such that $\frac{f'_i(y_0)}{y_0^{p_i-1}}=0$, then $f_i'(y_0)=0$, $f_i'(y)>0$ for $y<y_0$, and $f_i'(y)<0$ for $y>y_0$, hence $f_i(y)$ has a unique maximal point $y^*=y_0$. 

For case 2, since $f_i'(y)/y^{p_i-1}$ first decreases and then increases, $f_i'(y)$ may have 0, 1 or 2 zeros in $[0,\omega_ha_{\max}]$. 
The case that $f_i'(y)$ has 0 or 1 zeros is discussed in case 1. When $f_i'(y)$ has two zeros $0<y_1<y_2<w_ha_{\max}$, $f_i(y)$ is increasing in $[0,y_1]$, decreasing in $[y_1,y_2]$, and increasing again in $[y_2,\omega_ha_{\max}]$, and hence has two local maximum $y_1$ and $\omega_ha_{\max}$. When $f_i(y_1)=f_i(\omega_ha_{\max})$, $f_i$ has two maximal points with one equal to $\omega_ha_{\max}$; otherwise $f_i(y)$ has a unique maximal point.

When $a_{i+1}=\cdots=a_{N}=0$, \eqref{eq:fpibehavior} yields that $f_i'(y)/y^{p_i-1}$ is strictly decreasing, and hence $f_i(y)$ has a unique maximal point. Therefore, we finish the proof.
\end{proof}

\begin{proof} ({\it Proof of Proposition \ref{prop:noloteta}})
To simplify notation, we denote types slightly differently and let
     $ \Theta=H=\cup_{i=1}^{N}\{h=(i,h_{-\eta})|h_{-\eta}\in H_i\}$, and for any $h=(i,h_{-\eta})\in \Theta$, $\eta_h=\eta_i$. Suppose that $ 0< \eta_N<\cdots<\eta_1$.
    We define
    $$
    \tilde{\psi}_h=\frac{\psi_h}{\omega_h^{\frac{1}{\eta_h}+1}\left(\frac{1}{\eta_h}+1\right)},\quad\forall h\in \Theta.
    $$
    Then the deterministic problem for \eqref{eq:tax_obj2} subject to \eqref{eq:tax_rc} and \eqref{eq:tax_ic_new_eta} can be written 
    \begin{equation}\label{problem_eta_simplified}
    \small
   \begin{aligned} &\max_{c\in C,\,y\in Y}\sum_{h\in H}\left(\log(c_{h})-\tilde{\psi}_{h}y_h^{\frac{1}{\eta_h}+1}\right),\\
   \text{s.t.} &\log(c_h)-\tilde{\psi}_hy_h^{\frac{1}{\eta_h}+1}\ge \log(c_{h'})-\tilde{\psi}_{h}y_{h'}^{\frac{1}{\eta_h}+1}, \quad\forall h=(i,h_{-\eta}),h'=(j,h_{-\eta}')\in H,\text{ s.t. }i\le j,\\
   &\sum_{h\in H} y_h\ge \sum_{h\in H} c_h;
   \end{aligned}
\end{equation}
and the Lagrangian function of problem \eqref{problem_eta_simplified} has the structure
\begin{equation}\label{equ:lag_alpha2}
    \begin{aligned}
    \mathcal{L}(c,y;\lambda,\mu)=&\sum_{(i,h_{-\eta})\in H} \left[-\left(\tilde{\psi}_h+\sum_{j\ge i}\sum_{h'_{-\eta}\in H_j}\tilde{\psi}_{h}\lambda_{h,h'}-\sum_{j=i}\sum_{h'_{-\eta}\in H_{j}}\tilde{\psi}_{h'}\lambda_{h',h}\right) y_h^{\frac{1}{\eta_i}+1}+\right.\\
   &\left.\sum_{j=1}^{i-1}\left(\sum_{h'_{-\eta}\in H_{j}}\tilde{\psi}_{h'}\lambda_{h',h}\right) y_h^{\frac{1}{\eta_j}+1}+\mu y_{h}\right]+\mathcal{L}^c(c;\lambda,\mu)\\
   :=&\sum_{h\in H}\left[-\mathcal{B}_{h,i}(\lambda)y_h^{\frac{1}{\eta_i}+1}+\sum_{j=1}^{i-1}\mathcal{B}_{h,j}(\lambda)y_h^{\frac{1}{\eta_j}+1}+\mu y_h\right]+\mathcal{L}^c(c;\lambda,\mu),
    \end{aligned}
\end{equation}
where $h'=(j,h_{-\eta}')$, $\mathcal{B}_{h,j}(\lambda)\ge 0,\,\forall 1\le j\le i-1$, and $\mathcal{L}^c(c;\lambda,\mu)$ is independent of $y\in Y$. It is then straightforward to verify that the minimizer of 
$$
V(\lambda,\mu):=\max_{c\in C,y\in Y}\mathcal{L}(c,y;\lambda,\mu)
$$
satisfies $\mu^*>0$ (when $\mu=0$, we can take $\bar{c}_h\equiv c_{\max}$ and $\underline{y}_h\equiv0$ to attain $\mathcal{L}(\bar{c},\underline{y};\lambda,0)$ higher than the objective value of the optimal lottery solution). 
For any $h\in H$, according to Lemma \ref{lem:technical2} with $p_i=\frac{1}{\eta_i}+1,\,p_j=\frac{1}{\eta_j}+1(\forall 1\le j\le i-1),\,a_i=\mathcal{B}_{h,i}(\lambda^*),\,a_j=\mathcal{B}_{h,j}(\lambda^*)(\forall 1\le j\le i-1),a_j=0(\forall i+1\le j\le N) \,$and $a_0=\mu^*$, we know that the maximal $y_h^*$ is unique. The uniqueness of maximal $c_h^*$ is straightforward since $\mathcal{L}^c(c;\lambda^*,\mu^*)$ has the form $\sum_{h}X_h\log c_h-Y_hc_h$, where $Y_h>0$ for any $h$. Hence according to Result 1 of Corollary \ref{cor:nondelot}, the optimal solution to the planner problem \eqref{eq:tax_obj2} subject to  \eqref{eq:tax_rc} and \eqref{eq:tax_ic_new_eta} does not involve lotteries.
\end{proof}

We introduce an additional Lemma needed for the proof of Proposition \ref{prop:etaconv}.
\begin{lemma}\label{lem:eta_largel_samedeter}
    There exists $\bar{a}>0$, s.t. when $a_{\max}\ge \bar{a}$, the solution to the planner problem \eqref{eq:tax_obj2} subject to \eqref{eq:tax_rc} and \eqref{eq:tax_ic_new_eta} remains invariant with respect to $a_{\max}$.
\end{lemma}

\begin{proof}
   We denote $v_2(a_{\max})$ as the maximal objective value of the planner problem \eqref{eq:tax_obj2} subject to \eqref{eq:tax_rc} and \eqref{eq:tax_ic_new_eta}. According to Proposition \ref{prop:noloteta}, for any $a_{\max}>0$, the solution to the planner problem \eqref{eq:tax_obj2} subject to \eqref{eq:tax_rc} and \eqref{eq:tax_ic_new_eta} does not involve lottery in labor supply. We arbitrarily take a $\underline{a}>0$. Since
    $$
    u_h(c_h,y_h)\le u_h(c_{\max},0),\quad\forall h\in H, c_h\in[\epsilon,c_{\max}],\,y_h\in[0,\omega_ha_{\max}],
    $$
    and
    $$
    \lim_{y_h\rightarrow\infty}u_h(c_h,y_h)\rightarrow-\infty,\quad\forall h\in H, c_h\in[\epsilon,c_{\max}],\,y_h \ge 0.
    $$ 
    there exists $\bar{a}(h)>0$,  s.t. when $y_h>\bar{a}(h)w_h$, 
    $$
    u_h(c_h,y_h)+\sum_{h'\ne h}u_{h'}(c_{h'},y_{h'})<\left(v_2(\underline{a})-\sum_{h'\ne h}u_{h'}(c_{\max},0)\right)+\sum_{h'\ne h}u_{h'}(c_{\max},0)=v_2(\underline{a}).
    $$ 
    Therefore, for any $a_{\max}>\underline{a}$, the solution $(c(a_{\max}),y(a_{\max}))$ to the planner problem \eqref{eq:tax_obj2} subject to \eqref{eq:tax_rc} and \eqref{eq:tax_ic_new_eta} satisfies $y_h(a_{\max})\le \bar{a}(h)w_h$. We take $\bar{a}=\max_{h\in H}\bar{a}(h)$. Then for any $a_{\max}>\bar{a}$, the solution to the planner problem \eqref{eq:tax_obj2} subject to \eqref{eq:tax_rc} and \eqref{eq:tax_ic2}, remains invariant to the solution to the problem with $a_{\max}=\bar{a}$. Hence we finish the proof. 
\end{proof}

    \begin{proof} {\bf of Proposition \ref{prop:etaconv}}\\
    We denote $v_1(a_{\max})$ as the maximal objective value of the planner problem \eqref{eq:tax_obj2} subject to  \eqref{eq:tax_rc} and \eqref{eq:tax_ic2}; and $v_2(a_{\max})$ as the maximal objective value of the planner problem \eqref{eq:tax_obj2} subject to \eqref{eq:tax_rc} and \eqref{eq:tax_ic_new_eta}. Since \eqref{eq:tax_ic_new_eta} involves less constraints than \eqref{eq:tax_ic2}, we know $v_1(a_{\max})\le v_2(a_{\max})$ for any $a_{\max}>0$. According to Lemma \ref{lem:eta_largel_samedeter}, we have $v_2(a_{\max})=v_2(\bar{a})$ for any $a_{\max}\ge \bar{a}$. Hence, to prove this lemma, it then suffices to show that 
    $$
    \begin{aligned}
    &\lim_{a_{\max}\rightarrow+\infty}v_2(a_{\max})-v_1(a_{\max})= \lim_{a_{\max}\rightarrow+\infty}(v_2(a_{\max})-v_1(a_{\max}))_+\\
    =&\lim_{a_{\max}\rightarrow
    +\infty}(v_2(\bar{a})-v_1(a_{\max}))_+=0.
    \end{aligned}$$
    
    \textbf{Step 1.} The Slater's condition implies that, there exists $M>0$, s.t. for any $a_{\max}>\bar{a}$, the Lagrangian multipliers for the planner problem \eqref{eq:tax_obj2} subject to \eqref{eq:tax_rc} and  \eqref{eq:tax_ic2} satisfies $\max\{\max_{h,\,h'\in H}\lambda^*_{h,h'},\,\mu^*\}<M$.
    
    To be precise, for $a_{\max}=\bar{a}$, according to the Slater's condition, there exists a probability distribution $\bar{p}\in\mathcal{P}(C\times Y)$, s.t. $\bar{p}$ satisfies all the constraints in \eqref{eq:tax_rc} and \eqref{eq:tax_ic2} strictly.  Hence, there exists $M>0$, such that when $\max\{\max_{h,h'\in H}\lambda_{h,h'},\mu\}\ge M$, we have $L(\bar{p};\lambda,\mu)>v_2(\bar{a})$, where $L$ denotes the Lagrangian function of the planner problem \eqref{eq:tax_obj2} subject to  \eqref{eq:tax_rc} and  \eqref{eq:tax_ic2}. Since $\max_{p}L(p;\lambda,v)\ge L(\bar{p};\lambda,v)>v_2(\bar{a})\ge v_1(a_{\max})=\max_pL(p;\lambda^*,\mu^*)$, we know that $(\lambda,\mu)\not= (\lambda^*,\mu^*)$. Therefore, $\max\{\max_{h,\,h'\in H}\lambda^*_{h,h'},\,\mu^*\}<M$.
    
    \textbf{Step 2.} We show that, for any $\epsilon>0$, there exists $a(\epsilon)>\bar{a}$, such that when $a_{\max}>a(\epsilon)$, the Lagrangian multipliers of the planner problem  \eqref{eq:tax_obj2} subject to \eqref{eq:tax_rc} and \eqref{eq:tax_ic2} satisfies $\lambda_{h,h'}^*<\epsilon,\,\forall h,h'\in H$, s.t. $\eta_{h}<\eta_{h'}$. 
    
     Suppose $H=\cup_{i=1}^{N}\{h=(i,h_{-\eta})|h_{-\eta}\in H_i\}$, and for any $h=(i,h_{-\eta})\in H$, $\eta_h=\eta_i$. Suppose that $ 0< \eta_N<\cdots<\eta_1$. We first define the set
    $$\mathcal{G}=\{(h,h')\in H^2,\text{ s.t. } \eta_{h}<\eta_{h'}\}=\{(h=(i,h_{-\eta}),h'=(j,h'_{-\eta}))\in H^2,\text{ s.t. } i>j\}.$$
    Assume that for a planner problem \eqref{eq:tax_obj2} subject to \eqref{eq:tax_rc} and \eqref{eq:tax_ic2} with $a_{\max}>0$, there exists some Lagrangian multipliers $\lambda_{h',h}^*\ge \epsilon,\, (h'=(j,h_{-\eta}'),h=(i,h_{-\eta}))\in\mathcal{G}$.  We consider the Lagrangian function of the deterministic problem for \eqref{eq:tax_obj2} subject to \eqref{eq:tax_rc} and \eqref{eq:tax_ic2},
    $$
    \begin{aligned}
    \mathcal{L}(c,y;\lambda^*,\mu^*)= &\left[-\left(\tilde{\psi}_h+\sum_{h'\in H}\tilde{\psi}_{h}\lambda^*_{h,h'}-\sum_{k=i}\sum_{h_{-\eta}'\in H_k}\tilde{\psi}_{h'}\lambda^*_{h',h}\right) y_h^{\frac{1}{\eta_i}+1}+\right.\\
   &\left.\sum_{k\ne i}\left(\sum_{h'_{-\eta}\in H_{k}}\tilde{\psi}_{h'}\lambda^*_{h',h}\right) y_h^{\frac{1}{\eta_k}+1}+\mu^*  y_{h}\right]+\mathcal{L}^{c,y_{-h}}(c,y_{-h};\lambda^*,\mu^*)\\
   :=&\left[\sum_{k=1}^{N}\mathcal{B}_{h,k}(\lambda^*)y_h^{\frac{1}{\eta_k}+1}+\mu^*y_h\right]+\mathcal{L}^{c,y_{-h}}(c,y_{-h};\lambda^*,\mu^*),
    \end{aligned}
    $$
    where $h'=(k,h_{-\eta}')$, $\mathcal{L}^{c,y_{-h}}(c,y_{-h};\lambda^*,\mu^*)$ is independent of $y_h$, $\mathcal{B}_{h,j}(\lambda^*)\ge \epsilon \tilde{\psi}_{h'^{*}}$ for some $h'^{*}=(j,h_{-\eta}'^{*})$, and $\mathcal{B}_{h,k}(\lambda^*)\ge 0,\,\forall k\ne i$. According to Step 1, the Lagrangian multipliers $\lambda^*,\mu^*$ are bounded, hence there exists $M>0$ that is independent of $\lambda^*,\mu^*$, s.t. $|\mathcal{B}_{h,k}(\lambda^*)|\le M$, $\mu^*\le M$ and $\mathcal{L}^{c,y_{-h}}(c,y_{-h};\lambda^*,\mu^*)|_{c\equiv c_{\max},y_{-h}\equiv 0}\ge -M$. Therefore
    \begin{equation}\label{tax_mid1}
    \begin{aligned}
    &\max_{c\in C, y\in Y}\mathcal{L}(c,y;\lambda^*,\mu^*)\\
    \ge &\max_{y_h\in [0,\omega_ha_{\max}]}\left[\sum_{k=1}^{N}\mathcal{B}_{h,k}(\lambda^*)y_h^{\frac{1}{\eta_k}+1}+\mu^*y_h\right]-M\\
    \ge &\max_{y_h\in [0,\omega_ha_{\max}]}\left[\epsilon\tilde{\psi}_{h'^*}y_h^{\frac{1}{\eta_j}+1}-My_h^{\frac{1}{\eta_i}+1}\right]-M\rightarrow\infty \text{ as }a_{\max}\rightarrow\infty.
    \end{aligned}
    \end{equation}   
    On the other hand,
    \begin{equation}\label{tax_mid2}
    \max_{c\in C,y\in Y}\mathcal{L}(c,y;\lambda^*,\mu^*)=v_1(a_{\max})\le v_2(\bar{a}).
    \end{equation}
    Hence there exists some $\bar{a}(M,\epsilon,h,h')>0$, s.t. $a_{\max}<\bar{a}(M,\epsilon,h,h')$. Since the set $\mathcal{G}$ is finite, we can define 
    $$
    \bar{a}(\epsilon)=\max_{(h,h')\in\mathcal{G}}\bar{a}(M,\epsilon,h,h').
    $$
    Therefore, when $a_{\max}>\bar{a}(\epsilon)$, the Lagrangian multipliers should satisfy $\lambda_{h,h'}^*<\epsilon$, $\forall (h,h')\in \mathcal{G}$.
    \textbf{Step 3.} We denote $(\bar{c},\bar{y})$ the solution to the planner problem \eqref{eq:tax_obj2} subject to \eqref{eq:tax_rc} and \eqref{eq:tax_ic_new_eta} with $a_{\max}\ge\bar{a}$. It is straightforward to check that, for any $\epsilon>0$, there exists $\delta >0$, such that when $\max_{(h,h')\in\mathcal{G}}\lambda_{h,h'}^*\le \delta$,
    $$
    \mathcal{L}(\bar{c},\bar{y};\lambda^*,\mu^*)\ge L(\bar{c},\bar{y};\lambda^{**},\mu^{**})-\epsilon\ge  v_2(\bar{a})-\epsilon,
    $$
    where 
    $$
    \lambda_{h,h'}^{**}=\begin{cases}
        \lambda_{h,h'}^*,&(h,h')\not \in \mathcal{G};\\
        0,&(h,h')\in \mathcal{G}.
    \end{cases}
    $$
    According to Step 2, there exists $\bar{a}(\delta)>0$, such that when $a_{\max}\ge \bar{a}(\delta)$, the Lagrangian multipliers of the planner problem \eqref{eq:tax_obj2} subject to \eqref{eq:tax_rc} and \eqref{eq:tax_ic2} satisfy $\max_{(h,h')\in\mathcal{G}}\lambda_{h,h'}^*\le \delta$. Hence 
    $$
    v_1(a_{\max})=\max_{c\in C,y\in Y}\mathcal{L}(c,y;\lambda^*,\mu^*)\ge \mathcal{L}(\bar{c},\bar{y};\lambda^*,\mu^*)\ge v_2(\bar{a})-\epsilon,
    $$
    implying that
    $$
    (v_2(\bar{a})-v_1(a_{\max}))_{+}\le \epsilon.
    $$
    Since $\epsilon$ is arbitrary, we can conclude that $\lim_{a_{\max}\rightarrow
    \infty}(v_2(\bar{a})-v_1(a_{\max}))_+=0$ and finish the proof. \end{proof}

\subsection{Proof of Remark \ref{prop:etatwomax}}\label{app:etatwomax}
To establish the remark, we first need the following result. It shows that for every $a_{\max}>0$, whenever the optimal solution of some agent $h$ involves lottery, the lottery requires the agent to earn income $a_{\max}\omega_{h}$ with some positive probability $\pi_{h}$.
\begin{prop}\label{prop:lotshape}
         We consider the planner problem \eqref{eq:tax_obj2} subject to \eqref{eq:tax_rc} and \eqref{eq:tax_ic2}  with the consumption set $C=\times_{h\in H}[\epsilon,c_{\max}]$, the output set $Y=\times_{h\in H} [0,a_{\max}w_h]$, and the utility function defined by \eqref{eq:utility_func}. Assume that the Slater's condition holds, then the optimal solution to the planner problem does not involve lotteries in consumption. If the optimal solution requires a lottery for agent $\theta$, then the lottery requires this agent to earn income $a_{\max}\omega_{h}$ with some positive probability $\pi_{h}$.
\end{prop}
\begin{proof}
    Similarly as in \eqref{equ:lag_alpha2}, we have the Lagrangian as
    $$
    \mathcal{L}(c,y;\lambda^*,\mu^*)=\sum_{h\in H}\left[\underbrace{\sum_{j=i+1}^{N}\mathcal{B}_{h,j}(\lambda^*)y_h^{\frac{1}{\eta_j}+1}-\mathcal{B}_{h,i}(\lambda^*)y_{h}^{\frac{1}{\eta_i}+1}+\sum_{j=1}^{i-1}\mathcal{B}_{h,j}(\lambda^*)y_h^{\frac{1}{\eta_j}+1}+\mu^*y_h}_{I}\right]+\mathcal{L}^{c}(c;\lambda^*,\mu^*),
    $$
    where $\mathcal{B}_{h,j}(\lambda^*)\ge 0(\forall j\ne i)$ and $a_0=\mu^*>0$. According to Lemma \ref{lem:technical2}, if $I$ has multiple maximizers $y_h^{*,1}<y_h^{*,2}$, then $y_h^{*,2}=a_{\max}\omega_h$. Hence according to Result 1 of Corollary \ref{cor:nondelot}, we can conclude this proposition.
\end{proof}

To prove the proposition, it then suffices to show that the optimal solution for agents with $\eta\ne1/8$ requires lotteries. To show this, we will need the following lemma.
\begin{lemma}\label{lem:asymptotic_lag}
    Assume that $(\bar{c},\bar{y})$ is the solution to the planner problem \eqref{eq:tax_obj2} subject to \eqref{eq:tax_rc} and \eqref{eq:tax_ic_new_eta} as $a_{\max}\rightarrow\infty$, as defined in Lemma \ref{lem:eta_largel_samedeter}.\footnote{By Proposition \ref{prop:noloteta}, the optimal solution to the problem with partial incentive constraints is always deterministic, and hence unique, implying the uniqueness of $(\bar{c},\bar{y})$.} When there exists $h,h'$ satisfying $\eta_h<\eta_{h'}$ such that 
    $$
    u_h(\bar{c}_h,\bar{y}_h)-u_h(\bar{c}_{h'},\bar{y}_{h'})< 0,
    $$
    then there exists $\bar{a}$, s.t. when $a_{\max}\ge \bar{a}$, the optimal solution for agent $h'$ in the planner problem \eqref{eq:tax_obj2} subject to \eqref{eq:tax_rc} and \eqref{eq:tax_ic2}  requires lottery.
\end{lemma}
\begin{proof}
    We prove by contradiction. If not, then there exists $h,h'$ satisfying $\eta_h<\eta_{h'}$, and $a_{\max}^n\rightarrow \infty$, s.t.
      $$
    u_h(\bar{c}_h,\bar{y}_h)-u_h(\bar{c}_{h'},\bar{y}_{h'})< 0,
    $$
    and the optimal solution with $a_{\max}=a_{\max}^n$ for agent $h'$ does not involve lottery. Assume that the Lagrangian multipliers for $a_{\max}=a_{\max}^{n}$ is $(\lambda^n,\mu^n)$. According to Step1 and Step2 of the proof for Proposition \ref{prop:etaconv}, $(\lambda^n,\mu^n)$ is bounded and $\lambda^n_{h,h'}\rightarrow 0$ for $h,h'\in H$ s.t. $\eta_h<\eta_{h'}$.

    Therefore, up to a subsequence, $(\lambda^n,\mu^n)\rightarrow (\lambda^*,\mu^*)$ s.t. $\lambda^*_{h,h'}=0$ for $h,h'\in H$ satisfying $\eta_h\le \eta_{h'}$. Then for any $a_{\max}>0$, we have
    $$
    \begin{aligned}
    \max_{c\in C,y\in \times_{h\in H}[0,a_{\max}w_h]}\mathcal{L}(c,y;\lambda^*,\mu^*)&=\lim_{n\rightarrow\infty}\max_{c\in C,y\in \times_{h\in H}[0,a_{\max}w_h]}\mathcal{L}(c,y;\lambda^n,\mu^n)\\
    &\le \lim_{n\rightarrow\infty}\max_{c\in C,y\in \times_{h\in H}[0,a^n_{\max}w_h]}\mathcal{L}(c,y;\lambda^n,\mu^n)\\
    &=\lim_{n\rightarrow\infty}v_2(a_{\max}^n),
    \end{aligned}
    $$
    where $v_2(a_{\max}^n)$ is defined as in Lemma \ref{lem:eta_largel_samedeter}. On the other hand, we have
    $$
   \begin{aligned}
       \max_{c\in C,y\in \times_{h\in H}[0,a_{\max}w_h]}\mathcal{L}(c,y;\lambda^*,\mu^*)\ge&\inf_{\lambda,\mu}\max_{c\in C,y\in \times_{h\in H}[0,a_{\max}w_h]}\mathcal{L}(c,y;\lambda^*,\mu^*)\\
       =&v_2(a_{\max}),
   \end{aligned}
   $$
   where the infimum is taken over multipliers such that $\lambda_{h,h'}=0$ for $h,h'\in H$ satisfying $\eta_h\le \eta_{h'}$.
   Therefore, according to Lemma \ref{lem:eta_largel_samedeter} there exists $\bar{a}_1>0$, s.t. when $a_{\max}>\bar{a}_1$, $$
   \max_{c\in C,y\in \times_{h\in H}[0,a_{\max}w_h]}\mathcal{L}(c,y;\lambda^*,\mu^*)=v_2(\bar{a}_1)=\lim_{n\rightarrow\infty}v_2(a_{\max}^n).
   $$
   Hence $\lambda^*,\mu^*$ is the Lagrangian multipliers for the planner problem \eqref{eq:tax_obj2} subject to \eqref{eq:tax_rc} and \eqref{eq:tax_ic2}, and
   $$
   \arg \max_{c\in C,y\in \mathbb{R}_+^{|H|}}\mathcal{L}(c,y;\lambda^*,\mu^*)=(\bar{c},\bar{y}),
   $$
   where $\bar{y}_h$ is the only zero of
   $$
   \begin{aligned}
   &\frac{d}{dy_h}\left[-\mathcal{B}_{h,i}(\lambda^*)y_h^{\frac{1}{\eta_i}+1}+\sum_{j=1}^{i-1}\mathcal{B}_{h,j}(\lambda^*)y_h^{\frac{1}{\eta_j}+1}+\mu^*y_h\right]\\
   =&-\mathcal{B}_{h,i}(\lambda^*)(\frac{1}{\eta_i}+1)y_h^{\frac{1}{\eta_i}}+\sum_{j=1}^{i-1}\mathcal{B}_{h,j}(\lambda^*)(\frac{1}{\eta_j}+1)y_h^{\frac{1}{\eta_j}}+\mu^*.
    \end{aligned}
   $$
   We define $y_{h'}^n$ as the optimal income for agent $h'$ with $a_{\max}=a_{\max}^n$. Since total welfare is $v_2(\bar{a}_1)$ when $n$ is sufficiently large, $y_{h'}^n$ are bounded by some $\bar{y}_1$. According to Lemma \ref{lem:technical2}, $y_{h'}^n$ is the smallest zero of
   $$
   \sum_{j=i+1}^{N}\mathcal{B}_{h',j}(\lambda^n)(\frac{1}{\eta_j}+1)y_{h'}^{\frac{1}{\eta_j}}-\mathcal{B}_{h',i}(\lambda^n)(\frac{1}{\eta_i}+1)y_{h'}^{\frac{1}{\eta_i}}+\sum_{j=1}^{i-1}\mathcal{B}_{h',j}(\lambda^n)(\frac{1}{\eta_j}+1)y_{h'}^{\frac{1}{\eta_j}}+\mu^{n}.
   $$
   By continuity, $y_{h'}^n\rightarrow\bar{y}_{h'}$. Assume that $c_h^n$ is the optimal consumption, $y_h^n$ is in the  support of the optimal solution, then by continuity we have $\lim_{n\rightarrow\infty}c^h_n=\bar{c}^h$, $\liminf_{n\rightarrow\infty} y_h^n\ge \bar{y}_h$\footnote{According to Lemma \ref{lem:technical2}, $y_h^n$ is either the smallest zero of$$
   \sum_{j=i+1}^{N}\mathcal{B}_{h,i}(\lambda^n)(\frac{1}{\eta_j}+1)-\mathcal{B}_{h,i}(\lambda^n)(\frac{1}{\eta_i}+1)y_{h}^{\frac{1}{\eta_i}}+\sum_{j=1}^{i-1}\mathcal{B}_{h,j}(\lambda^n)(\frac{1}{\eta_j}+1)y_{h}^{\frac{1}{\eta_j}}+\mu^{n}\omega_{h}
   $$ or $y_h^{n}=a_{\max}^n$.}. We can then verify that the incentive constraint for $h$ and $h'$ does not hold when $a_{\max}=a_{\max}^n$ as $n\rightarrow\infty$, which contradicts the assumption, and we finish the proof.
\end{proof}

The stated remark then follows from the fact that the solution to the planner problem \eqref{eq:tax_obj2} subject to \eqref{eq:tax_rc} and \eqref{eq:tax_ic_new_eta} satisfies: for any \((\omega',\eta')\) with \(\eta' \neq 1/8\), there exists \((\omega,\eta)\) such that \(\eta < \eta'\) and
\[
u_{\omega,\eta}(\bar{c}_{\omega,\eta},\bar{y}_{\omega,\eta}) - u_{\omega,\eta}(\bar{c}_{\omega',\eta'},\bar{y}_{\omega',\eta'}) < 0.
\]

\section{Worst case complexity bounds}
\label{app:complex}

\subsection{Lagrangian iteration}

Proposition \ref{thm_complexity} provides bounds on the total number of iterations needed to obtain a desired accuracy.
Next, we consider the computational complexity of each iteration. We discretize the set \(C\) into the same finite grid \(\hat{C}\) as in the linear programming approach. For each iteration, we first solve the problem
\begin{equation}\label{Lag_sub_problem}
(a^k, c^k) \in \arg \max_{a \in A, c \in \hat{C}} \left[ f(a, c) - \sum_{i=1}^{m} \lambda_i^k g_i(a, c) - \sum_{j=1}^{\ell} \gamma^k_{j, a} h_{j}(a, c) \right].
\end{equation}
The computation of this optimization problem involves evaluating the Lagrangian function at \( |A||\hat{C}| \) points, with each evaluation requiring \( O(m+\ell) \) computations. Thus, the complexity of finding the maximum point of the Lagrangian function is \( O\left(|A||\hat{C}|(m+\ell)\right) \). The complexity of updating the Lagrange multipliers is \( O\left(m+\ell\right) \). Therefore, the complexity for each iteration is:
\begin{equation}\label{iter_complexity}
O\left(|A||\hat{C}|(m+\ell)\right).
\end{equation}

Combining \eqref{eq:n_itertime_mh} and \eqref{iter_complexity}, we have the theorem for the overall computational complexity for finding an \(\epsilon\)-optimal lottery solution as follows:
\begin{theorem}\label{thm:overall_complexity}
 We take all the assumptions in Theorem \ref{math_thm_final_mh}. Additionally, we assume there exist two constants $M\ge 0, \bar{\Lambda}\ge 0$ such that $\|g_i\|_{\infty}\le M\,(i=1,\cdots,m),\,\|h_{j}\|_{\infty}\le M\,(j=1,\cdots,\ell),\,$ and   \(\|(\lambda^k,\gamma^k)\|_{\infty}+\Lambda(\lambda^1,\gamma^1)\le \bar{\Lambda} \) throughout the iterations. We take $\mu^k\sim k^{-\frac{1}{2}(1+\rho)}$ for $0<\rho\le1$. Then for $\epsilon>0$, the overall computational complexity for  finding an \(\epsilon\)-optimal lottery solution is
 \begin{equation}\label{overall_complexity}
\begin{cases}
O\left(\left(\frac{M(\frac{1}{\rho}+\bar{\Lambda})}{\epsilon}\right)^{\frac{2}{1 - \rho}} (|A||\hat{C}|(m+\ell)^{1+\frac{1}{1-\rho}})\right), & \rho < 1; \\
O\left(e^{\frac{M(\frac{1}{\rho}+\bar{\Lambda})\sqrt{m+\ell}}{\epsilon}} (|A||\hat{C}|(m+\ell))\right), & \rho = 1.
\end{cases}
\end{equation}
\end{theorem}
\begin{proof}
 The overall complexity \eqref{overall_complexity} can be directly obtained by multiplying \eqref{eq:n_itertime_mh} and \eqref{iter_complexity}. 
\end{proof}
\subsubsection{Estimate for Partial First-order Approach}
\label{complex:pfoc}
As indicated above, a crucial part of our algorithm consists of solving the maximization problem in Step 1. So far, our complexity analysis assumed that this is done by discretizing $ C $ and simple grid search.
However, as explained in Section \ref{sec:comper}, in many economic applications, the subproblem \eqref{Lag_sub_problem} exhibits specific structural properties that allow for more efficient solution methods. For example, in the principal-agent problem  with utility functions for agents that are separable in action and consumption, the subproblem\eqref{Lag_sub_problem} satisfies the following properties
\begin{equation}\label{FOC_property}
\begin{aligned}
    g_i(a,c)&=u_{i,0}(a)+\sum_{r=1}^{d}u_{i,r}(a)w_r(c_r),\quad \forall i\in\{1,\cdots, m\},\, a\in A,\,c=(c_1,\cdots,c_d)\in C\subset \mathbb{R}^d\\
    h_j(a,c)&= v_{j,0}(a)+\sum_{r=1}^{d} v_{j,r}(a) w_r(c_r),\quad \forall j\in\{1,\cdots,\ell\},\,a\in A,\,c=(c_1,\cdots,c_d)\in C\subset \mathbb{R}^d,
\end{aligned}
\end{equation}
where $C=\times_{r=1}^{d}[c_{r,\text{min}},c_{r,\text{max}}]$, for some strictly increasing and strictly concave functions $w_r(c_r)$\footnote{This property holds because each consumption $c_r$ appears in $g$ and $h$ in a single, specifically concave form. See Section \ref{sec:pri-ag} for more details. The case there is more specific, as $w_1=\cdots=\omega_{|Q|}$.}.
The following Lemma implies that, with the separable property \eqref{FOC_property}, and the assumption that the principal's utility function $f(a,c)=v(a,c)$ is strictly decreasing and weakly concave with respect to each $c_r$, the first-order approach for \(c\) can be implemented without the need to grid the set \(C\) for solving the subproblem \eqref{Lag_sub_problem}, further enhancing the algorithm's efficiency.
\begin{lemma}\label{FOC_lemma} 
Given $\lambda> 0 \,(\lambda\ge 0 \text{ and there exists $i\in\{1,\cdots,m\}$ such that $\lambda_i>0$}),\gamma\ge 0$. We assume that the separable property \eqref{FOC_property} holds and that the principal's utility $f(a,c)=v(a,c)$ is strictly decreasing and weakly concave with respect to each $c_r$. For any $a\in A$, we consider the following system:  
\begin{equation}\label{Lag_FOC_cond}
\begin{cases}
    \mathcal{A}(a,\lambda,\gamma,r)\frac{\partial}{\partial c_r}w_r(c_r)=\frac{\partial}{\partial c_r}v(a,c),&\mathcal{A}(a,\lambda,\gamma,r)<0;\\
    c_r=c_{r,\text{min}},&\mathcal{A}(a,\lambda,\gamma,r)\ge 0,
\end{cases}
\end{equation}
where
$$
\mathcal{A}(a,\lambda,\gamma,r)=\sum_{i=1}^{m}\lambda_iu_{i,r}(a)+\sum_{j=1}^{\ell}\gamma_{j,a}v_{j,r}(a).
$$
If \eqref{Lag_FOC_cond} admits a solution $c(a)\in C$ for any $a\in A$, then the solution $(a^*,c^*)$ to the subproblem \eqref{Lag_sub_problem} should satisfy $c^*=c(a^*)$. 
\end{lemma}
\begin{proof}
With the property \eqref{FOC_property},  the subproblem \eqref{Lag_sub_problem} in the continuous consumption space $C$ can be written as
    \begin{equation}\label{Lag_sub_problem_FOC}
        (a^*, c^*) \in \arg \max_{a \in A, c \in C} \mathcal{L}(a,c;\lambda,\gamma):=v(a,c)-\mathcal{A}(a,\lambda,\gamma,0)-\sum_{r=1}^{d}\mathcal{A}(a,\lambda,\gamma,r)w_r(c_r).
    \end{equation}
    We define the function on $c$ as
    $$
    \mathcal{F}(c;a^*,\lambda,\gamma):=\mathcal{L}(a^*,c;\lambda,\gamma).
    $$
    Then \eqref{Lag_sub_problem_FOC} yields that $c^*$ is the maximizer of $\mathcal{F}$, given $a^*,\,\lambda,\,$ and $\gamma$. It is direct to show that
    \begin{equation}\label{Lag_sub_problem_FOCond}
        \frac{\partial \mathcal{F}(c;a^*,\lambda,\gamma)}{\partial c_r}=\frac{\partial\mathcal{L}(a^*,c;\lambda,\gamma)}{\partial c_r}=  \frac{\partial}{\partial c_r}v(a^*,c)-\mathcal{A}(a^*,\lambda,\gamma,r)\frac{\partial}{\partial c_r}w_r(c_r).
    \end{equation}
    When $\mathcal{A}(a^*,\lambda,\gamma,r)\ge 0$, the right-hand side of \eqref{Lag_sub_problem_FOCond} is negative, implying that $\mathcal{F}$ is strictly decreasing with respect to $c_r$, and hence $c_r^*=c_r(a^*)=c_{r,\text{min}}$. 
    
    We then define $\hat{c}:=(c_r)_{r\in\{\mathcal{A}(a^*,\lambda,\gamma,r)<0\}}$ and $\tilde{c}:=(c_r)_{r\in\{\mathcal{A}(a^*,\lambda,\gamma,r)\ge0\}}$, and hence $c=(\hat{c},\tilde{c}).$ According to the discussion above, we know that $\tilde{c}^*=\tilde{c}(a^*)=(c_{r,\text{min}})_{r\in\{A(a^*,\lambda,\gamma,r)\ge0\}}$. By the fact that $c^*$ maximizes $\mathcal{F}(c;a^*,\lambda,\gamma)$, it is then straightforward to show that $\hat{c}^*$ should maximize $\mathcal{F}(\hat{c},\tilde{c}^*;a^*,\lambda,\gamma)$, or equivalently, $\hat{c}^*$ should maximize
    \begin{equation}\label{equiv_F}
    v(a^*,\hat{c},\tilde{c}^*)-\sum_{r\in\{\mathcal{A}(a^*,\lambda,\gamma,r)<0\}}\mathcal A(a^*,\lambda,\gamma,r)w_r(c_r),
    \end{equation}
    which is a strictly concave function in $\hat{c}$. Therefore, the following first order condition to \eqref{equiv_F},
    $$
    \mathcal{A}(a^*,\lambda,\gamma,r)\frac{\partial}{\partial c_r}w_r(c_r)=\frac{\partial}{\partial c_r}v(a^*,\hat{c},\tilde{c}^*),\quad r\in\{\mathcal{A}(a^*,\lambda,\gamma,r)<0\},
    $$
    admits at most one solution $\hat{c}(a^*)$. From the assumption we know that $(\hat{c}(a^*),\tilde{c}^*)\in C$, then we have $c^*=(\hat{c}(a^*),\tilde{c}^*)$. Hence, we have shown that $(a^*,c^*)$ satisfies the system \eqref{Lag_sub_problem_FOCond}.
\end{proof}
 If all assumptions in Lemma \ref{FOC_lemma} hold, then \eqref{Lag_FOC_cond} can be applied for solving the subproblem \eqref{Lag_sub_problem}. To be precise, the subproblem \eqref{Lag_sub_problem} can be solved by the following three-step procedure:
\begin{enumerate}
    \item Computing $\mathcal{A}(a,\lambda,\gamma,r)$ for all $a\in A$ and $r\in\{1,\cdots,d\}$;
    \item Determining $c(a)$ for all $a\in A$ by
    \begin{equation}\label{FOC_equation}
    \begin{cases}
     \mathcal{A}(a,\lambda,\gamma,r)\frac{\partial}{\partial c_r}w_r(c_r(a))=\frac{\partial}{\partial c_r}v(a,c(a)), &\mathcal{A}(a,\lambda,\gamma,r)<0;\\
    c_r=c_{r,\text{min}},&\mathcal{A}(a,\lambda,\gamma,r)\ge 0.
     \end{cases}
    \end{equation}
    \item Determining optimal $a\in A$ for the subproblem \eqref{Lag_sub_problem} with $c=c(a)$.
\end{enumerate}
If additionally the principal's utility function $v$ is linear with respect to $c$\footnote{This assumption holds for a broad class of problems, including the principal-agent problem discussed in the section \ref{sec:pri-ag}, where the principal's utility is defined as the difference between the total output and the total consumption amount, see Section \ref{sec:pri-ag} for details. In general, equation \eqref{FOC_equation} can be solved using standard methods, such as Newton's method. In this context, the complexity of solving the nonlinear equation \eqref{FOC_equation} for a given \(a \in A\) is polynomial in \(d\), which is still significantly more efficient than discretizing the consumption set. The resulting complexity analysis would follow a structure similar to Theorem \ref{thm:overall_complexity_FOC}, where \( |\hat{C}| \) is replaced by an expression in terms of \(d\). For simplicity, we omit a rigorous analysis of these cases.}, then  $\frac{\partial}{\partial c_r}v(a,c)=\mathcal{B}(a,r)$ is independent of $c$ and \eqref{FOC_equation} can be solved by
$$
\begin{cases}
c_r=\left(\frac{\partial w_r}{\partial c_r}\right)^{-1}\left(\frac{\mathcal{B}(a,r)}{\mathcal{A}(a,\lambda,\gamma,r)}\right), &\mathcal{A}(a,\lambda,\gamma,r)<0;\\
c_r=c_{r,\text{min}},&\mathcal{A}(a,\lambda,\gamma,r)\ge 0.
\end{cases}
$$
It is directly to check that the complexity for this three-step procedure is
\begin{equation}\label{iter_complexity_FOC}
O\left(\underbrace{d(m+\ell) |A|}_{\text{STEP1}}+\underbrace{d|A|}_{\text{STEP2}}+\underbrace{d(m+\ell )|A|}_{\text{STEP3}}\right)\sim O\left(d (m+\ell )|A|\right), 
\end{equation}
where $d$ is the dimension of the consumption space $C$. 

We combine \eqref{iter_complexity_FOC} and proposition \ref{thm_complexity} to obtain the complexity estimate for the first-order approach as following:
\begin{theorem}\label{thm:overall_complexity_FOC}
We take all the assumptions in Theorem \ref{math_thm_final_mh}. Additionally, we assume there exist two constants $M\ge 0, \bar{\Lambda}\ge 0$ such that $\|g_i\|_{\infty}\le M\,(i=1,\cdots,m),\,\|h_{j}\|_{\infty}\le M\,(j=1,\cdots,\ell),\,$ and   \(\|(\lambda^k,\gamma^k)\|_{\infty}+\Lambda(\lambda^1,\gamma^1)\le \bar{\Lambda} \) throughout the iterations. Also, we assume that each $f(a,c)=v(a,c)$ is linear in $c$. We take $\mu^k\sim k^{-\frac{1}{2}(1+\rho)}$ for $0<\rho\le1$. Then for $\epsilon>0$, the overall computational complexity for  finding an \(\epsilon\)-optimal lottery solution is
 \begin{equation}\label{overall_complexity_FOC}
\begin{cases}
O\left(\left(\frac{M(\frac{1}{\rho}+\bar{\Lambda})}{\epsilon}\right)^{\frac{2}{1 - \rho}} d |A|(m+\ell )^{1+\frac{1}{1-\rho}}\right), & \rho < 1; \\
O\left(e^{\frac{M(\frac{1}{\rho}+\bar{\Lambda})\sqrt{m+\ell}}{\epsilon}} d |A|(m+\ell )\right), & \rho = 1.
\end{cases}
\end{equation}
where $d$ is dimension of the consumption space $C$.
\end{theorem}
\begin{proof}
The overall complexity \eqref{overall_complexity_FOC} can be directly obtained by multiplying \eqref{eq:n_itertime_mh} and \eqref{iter_complexity_FOC}.   
\end{proof}
\begin{remark}
    By using the first-order approach, we can replace $|\hat{C}|$ in Theorem \ref{thm:overall_complexity} by $d$ in the complexity estimate.
\end{remark}

\subsection{Comparison to Existing Methods}\label{sect_Bcompare}
In the literature, the most commonly used method for solving the lottery problem \eqref{math_lot_mh} is to formulate it as a linear program and to use standard algorithms for linear programming \citep{prescott1984general,prescott2004computing}. To use linear programming, we first need to discretize the set $C$ into a finite grid $\hat{C}$. This transforms \eqref{math_lot_mh} into a finite-dimensional problem:
\begin{equation}\label{math_lot_mh_lp}
\begin{aligned}
    &\max_{\hat{x}\in \mathcal{P}(A\times \hat{C})}\sum_{a\in A}\sum_{\,c\in \hat{C}}f(a,c)\hat{x}(a,c),\\
    \textbf{s.t. }&\sum_{a\in A}\sum_{c\in \hat{C}}g_i(a,c)\hat{x}(a,c)\le0\text{ }(i\in\{1,\cdots,m\}),\\
    &\sum_{c\in \hat{C}} h_j(a,c)\hat{x}(a,c)\le0 \text{ }(j\in\{1,\cdots,l\},\,a\in A).
\end{aligned}
\end{equation}

This formulation constitutes a standard linear programming problem, where $\hat{x}$ is a vector of dimension $|A||\hat{C}|$, and the number of inequality constraints is $\ell|A| + m$. According to classical theory on interior-point algorithms for linear programming \citep{nocedal2006numerical}, the computational complexity of solving this problem is $O\left(\log(1/\delta) \cdot (|A||\hat{C}| + \ell|A| + m)^{3.5}\right)$, where $\delta$ denotes the error for complementary conditions in interior-point algorithms.\footnote{In \cite{prescott1984general,prescott2004computing}, the Dantzig-Wolfe decomposition is employed to enhance the linear programming algorithm, addressing both efficiency and memory concerns. However, this algorithm requires the computation of extreme points of the linear constraints and the construction of iterations based on these extreme points. This suggests that the method effectively relies on an extended version of the simplex algorithm, which lacks well-established theoretical complexity results. Consequently, to provide a theoretical comparison of computational complexity, we focus on traditional interior-point methods.}

Comparing the complexity of our approach in Theorem \ref{thm:overall_complexity} to that of the linear programming method, we see that 
in the case that $|\hat{C}|\sim |A|\sim \ell\gg m$, 
$$
|A||\hat{C}|(m + \ell)^{1+\frac{1}{1-\rho}}\sim |A|^{3+\frac{1}{1-\rho}}\ll |A|^7\sim (|A||\hat{C}| + \ell|A| + m)^{3.5} ,
$$
for $\rho$ close to 0. Considering the fact that Algorithm \ref{math_alg_mh} actually needs less iterations than the estimation in Proposition \ref{thm_complexity}, the Lagrangian iteration algorithm can perform significantly better than using linear programming directly when facing problems with large $|A|$ and $|\hat{C}|$ and the required accuracy $\epsilon$ is not very stringent. 
We note that for a problem that is discretized from a continuous problem, the scales \( |A| \) and \( |\hat{C}| \) are chosen in relation to \( \epsilon \), such that the exact solution to the discretized problem is an \( \epsilon \)-optimal solution to the continuous problem. Therefore, the complexity term involving \( \epsilon \), specifically \( \left(\frac{\frac{M}{\rho} + \bar{\Lambda} + 1}{\epsilon}\right)^{\frac{2}{1 - \rho}} \), cannot be neglected when analyzing the overall complexity of our methods. However, for problems with high-dimensional action and consumption sets, this complexity term is small compared to the term we analyzed earlier, namely $$ |A| |\hat{C}| (m + \ell)^{1 + \frac{1}{1 - \rho}} .$$

For instance, consider a problem with a \( d \)-dimensional action set and a \( d \)-dimensional consumption set, i.e., \( A = \times_{i=1}^{d} A_i \) and \( C = \times_{i=1}^{d} C_i \). A typical choice would be \( |A| = \prod_{i=1}^{d} |A_i| \sim \left( \frac{1}{\epsilon} \right)^{\frac{d}{2}} \) and \( |\hat{C}| = \prod_{i=1}^{d} |\hat{C}_i| \sim \left( \frac{1}{\epsilon} \right)^{\frac{d}{2}} \). \footnote{Heuristically, we define the function \( f(x_1, \dots, x_d) = -\frac{1}{d} \sum_{k=1}^{d} x_k^2 \) over the domain \( [-1, 1]^d \). The maximum value of \( f \) should be 0. If we choose \( h = \frac{1}{2N+1} \) and evaluate \( f \) on the discrete set \( \{ (2k+1)h : -N-1 \le k \le N \}^d \), the maximum value of \( f \) should be \( -h^2 \). Thus, to achieve an \( \epsilon \)-optimal solution, we select \( h \sim \sqrt{\epsilon} \), leading to the number of points \( |\{ (2k+1)h : -N-1 \le k \le N \}^d| \sim \left( \frac{1}{\epsilon} \right)^{\frac{d}{2}} \).} Therefore, when $d$ is large, 
$$
 \left(\frac{\frac{M}{\rho} + \bar{\Lambda} + 1}{\epsilon}\right)^{\frac{2}{1 - \rho}} \ll \left(\frac{1}{\epsilon}\right)^{\frac{d}{2}} =|A|,
$$
and the complexity of our methods is smaller than the complexity of the linear programming method. Furthermore, for problems satisfying assumptions in Theorem \eqref{thm:overall_complexity_FOC}, such as the moral hazard problems, the first-order approach for \(c\) can be implemented without the need to grid the set \(C\), further improving the algorithm's efficiency by a factor of approximately $\frac{|\hat{C}|}{d}$, where $d$ is the dimension of $C$.

\subsubsection{Decomposability in linear programming}
\label{sec: decoLP}
 For a decomposable problem, the linear programming formulation \eqref{math_lot_mh_lp} can be equivalently rewritten as
\begin{equation}\label{math_lot_mh_lp_decom}
    \begin{aligned}
    &\max_{\left(\hat{x}_k\in \mathcal{P}(A\times \hat{C}_k)\right)_{1\le k\le d}}\sum_{a\in A}\sum_{k=1}^d\sum_{c_k\in \hat{C}_k}f_k(a,c_k)\hat{x}_k(a,c_k),\\
    \textbf{s.t. }&\sum_{a\in A}\sum_{k=1}^{d}\sum_{c_k\in \hat{C}_k}g_{i,k}(a,c_k)\hat{x}_k(a,c_k)\le0\text{ }(i\in\{1,\cdots,m\}),\\
    &\sum_{k=1}^d\sum_{c_k\in \hat{C}_k} h_{j,k}(a,c_k)\hat{x}_k(a,c_k)\le0 \text{ }(j\in\{1,\cdots,\ell\},\,a\in A),\\
    &\sum_{c_1\in\hat{C}_1}\hat{x}_1(a,c_1)=\cdots=\sum_{c_d\in\hat{C}_d}\hat{x}_d(a,c_d):=\pi(a),(a\in A).
\end{aligned}
\end{equation}
where the joint distribution $\hat{x}(a,c)$ is recovered as $\hat{x}(a,c) = \pi(a)\prod_{k=1}^{d} \hat{x}_k(c_k|a)$. A   concrete illustration of a similar reformulation is given in Section \ref{sec:pri-ag}.
\end{document}